\documentclass[12pt]{article}
\usepackage[utf8]{inputenc}
\usepackage[english]{babel}
\usepackage{geometry}
\usepackage[dvipsnames]{xcolor}

\usepackage{tikz-cd}
\usepackage{tikz}
\usepackage{tikz-network}
\usepackage{physics}
\usepackage{amssymb}
\usepackage{amsthm}
\usepackage{mathtools}
\usepackage{breqn}
\usepackage{amsthm}
\usepackage{amsmath}
\usepackage{graphicx}
\usetikzlibrary{shapes.geometric} 
\usepackage{mathrsfs}
\usepackage{array}	
\usepackage{xcolor}
\usepackage{orcidlink}

\usepackage{authblk}

\usepackage[makeroom]{cancel}
\DeclareMathOperator{\id}{id}
\DeclareMathOperator{\Ima}{Im}
\DeclareMathOperator{\Mat}{Mat}

\newtheorem{defn}{Definition}[section]	
\newtheorem{prop}{Proposition}[section]
\newtheorem{exmp}{Example}[section]
\newtheorem{thm}{Theorem}[section]

\newtheorem{lm}{Lemma}[section]
\newtheorem{rem}{Remark}[section]
\newtheorem{conj}{Conjecture}[section]

\newcommand\hb{\hat{b}}
\newcommand\w{\wedge}
\newcommand\dl{\dd \log}
\newcommand\te{\tilde{e}}

\newcommand{\C}{{\mathbb C}}

\newcommand{\rD}{\mathrm{D}}

\newcommand\om{{\omega}}

\newcommand\bx{{\bf x}}
\newcommand\by{{\bf y}}

\newcommand\tx{\tilde{x}}
\newcommand\im{\mathrm{im}\,} 

\newcommand\hatom{\hat{\omega}}

\newcommand{\htau}{\hat{\tau}}

\newcommand{\cT}{\mathcal{ T}}

\newcommand{\cL}{{\cal L}}

\newcommand{\ts}{\textsuperscript}

\newcommand{\Mod}[1]{\ (\mathrm{mod}\ #1)}

\usepackage{float}

\usetikzlibrary{decorations.shapes}
\tikzset{decorate sep/.style 2 args=
{decorate,decoration={shape backgrounds,shape=circle,shape size=#1,shape sep=#2}}}

\title{Integrable deformations of cluster maps of type $D_{2N}$ }
\author{Wookyung Kim \footnote{Corresponding author e-mail: kimw@ms.u-tokyo.ac.jp} \,\orcidlink{0009-0004-8512-3014}}

\affil[]{Graduate School of Mathematical Sciences  \protect\\ 
University of Tokyo, 3-8-1 Komaba, Tokyo 153-8914, Japan.
}

\date{\today}

\begin{document}

\maketitle

\begin{abstract}\

In this paper, we extend one of the main results from our joint work \cite{hkm24} with Hone and Mase, in which we studied a deformed type $D_{4}$ map, to the general case of the type $D_{2N}$ for $N\geq3$. This can be achieved through a ``local expansion" operation, introduced in our joint work \cite{grab} with Grabowski and Hone. This operation involves inserting a specific subquiver into the quiver arising from the Laurentification of the deformed type $D_{4}$ map. This insertion yields a new quiver, obtained through the Laurentification of the deformed type $D_{6}$ map and thus enables systematic generalization to higher ranks $D_{2N}$. We  also study the degree growth of deformed type $D_{2N}$ map via the tropical method and conjecture that, for each $N$, the deformed map is an integrable, as indicated by the algebraic entropy test, the criterion for detecting integrability in the discrete dynamical systems.
\end{abstract}

\section{Introduction}

Cluster algebras are commutative algebras which was introduced by Fomin and Zelevensky \cite{FZ2001}. These algebras are defined by collection of distinguished generators, \textit{cluster variables} which are produced through the iterative process of  \textit{mutation}. Performing mutation $\mu_{j}$ on a \textit{cluster} $\vb{x} = (x_{1},\dots, x_{n})$, $n$-tuple of cluster variables, replaces variable $x_j$ with new variable $x_j'$ given by the following difference equation,
$$x_j'x_j = F(x_1,\dots,x_{j-1},x_{j+1},\dots, x_{n})$$
where $x_j'=\mu_{j}(x_j)$ and $F$ is subtraction-free polynomial in initial cluster variables. In the cluster algebra, this equation is known as \textit{exchange relation}, and it is determined by the \textit{quiver}, $Q$, that is a finite directed graph, consisting of no loops or 2-cycles. Moreover the quiver itself transforms under mutation, together with cluster variables. Thus more precisely, a cluster algebra is a collection of \textit{seeds}, where each seed is formed by a set of distinguished cluster variables and associated quiver. These seeds are induced from  the initial seed $(\vb{x},Q)$ by recursively applying the mutations. 

The exchange relation is said to possess the \textit{Laurent property}, that is, every variable induced by the exchange relation can be expressed as a Laurent polynomial in initial variables. In other words, new cluster variables induced by any composition of mutations $\mu_{i_{n}}\cdot \mu_{i_{n-1}} \cdots \mu_{i_{2}} \cdot \mu_{i_{1}}$ are elements of the Laurent polynomial ring $\mathbb{Z}_{\geq 0}[x_{1}^{\pm},x_{2}^{\pm},\dots, x_{n}^{\pm}   ]$. This is a key distinguished feature of cluster algebra.

In a particular case, if there exists a composition of mutations that fixes the exchange matrix, then it can be regarded as a birational map which maps between clusters,$\varphi:\vb{x} \to \vb{x}' = (x_1', x_2', \dots, x_n')$, which is referred to as \textit{cluster map}. This enables the description of discrete dynamical systems within the framework of cluster algebra, which leads to the study of the area of discrete integrable systems \cite{Shigeru,FH2013}. 

For instance, non-linear recurrence relation in the form of 
\begin{equation}\label{lyness5}
x_{n+2}x_{n} = 1 +  x_{n+1}
\end{equation}
is known as the \textit{Lyness recurrence} relation. The recurrence is equivalent to the iteration of the birational map given by 
\begin{equation}
\varphi:(x_{1},x_{2})\to \qty(x_{2},\frac{1+x_2}{x_1})
\end{equation}
This map can be constructed by composition of mutations in the cluster algebra with quiver which is identical to Dynkin type $A_{2}$. This map preserves the symplectic 2-form 
\begin{equation}
\omega = \frac{1}{x_1x_2}\dd x_1 \w \dd x_2
\end{equation}
 and there exists a symmetric function (first integral) which is invariant under the action of the map. Thus it is the Liouville integrable map whose iteration represents the discrete integrable system. 
One can observe that the orbit of recurrence relation is periodic with period 5 i.e. $(x_{n+5} = x_{n})$. This is the simplest example of Zamolodchikov periodicity, that is, the periodicity of the orbit of birational maps for which iterates returns to initial data after a fixed number of evolutions. This periodicity was first observed by Zamolodchikov \cite{zam} in the dynamics induced by the difference equation known as the \textit{Y-system}, which provides solutions to the Bethe ansatz equations associated with conformal field theories of the ADE scattering type. This Y-system was later  realized within the framework of cluster algebra by introducing \textit{coefficient variable} defined over a semifield. Subsequently, Fomin and Zelevinsky showed that both dynamics of coefficient variable and cluster variable exhibit same periodicity in the cluster algebras whose quiver has same structure as the finite Dynkin diagram \cite{fz2}.

The following parameter modification on the recurrence relation
\begin{equation}\label{lynessdeformed}
x_{n+2}x_{n} = b+ax_{n+1}
\end{equation}
does preserve the symplectic structure and integrability of the map. In other words, one can generalise the integrable cluster map by imposing arbitrary parameters in the exchange relation. This deformation provides the motivation behind the development of the results presented by Hone and Kouloukas in  \cite{hk}. However, in exchange for admitting deformation of the cluster map, the iteration no longer satisfies the Laurent property. This implies  that the map cannot be defined in the cluster algebras. One way to restore this property is to lift the deformed map to the higher dimensional space where the Laurent property holds. This method is so called Laurentification (a term coined by Hamad et al. in \cite{hkha17}), which is established by constructing the variable transformation through the $p$-adic methods
introduced by Kanki in \cite{kanki} which is analogous to the singularity confinement test in discrete systems \cite{GRP}. Through the following variable transformation, 
\begin{equation}
x_{n}= \frac{\tau_{n+5}\tau_{n}}{\tau_{n+2}\tau_{n+3}}
\end{equation}
the recurrence relation can be rewritten as 
\begin{equation}\label{specialsomos7}
\tau_{n+7}\tau_{n} = a \tau_{n+1}\tau_{n+6} + b \tau_{n+3}\tau_{n+4}
\end{equation}
which can be realised via mutations in cluster algebra associated with the quiver shown in Figure \ref{Qsomos7}.  This means that the parametric map can be generated from the cluster variables in cluster algebra of higher rank.

\begin{figure}[h!]
\begin{center}
\resizebox{0.4\textwidth}{!}{%
 \begin{tikzpicture}[every circle node/.style={draw,scale=0.6,thick},node distance=15mm]
  \node [draw,circle,fill=red!50,"$5$"] (5) at (0,0) {};
  \node [draw,circle,fill=red!50,"$6$"] (6) [right=of 5] {};
  \node [draw,circle,fill=red!50,"$7$" below] (7) [below right=of 6] {};
  \node [draw,circle,fill=red!50,"$4$" below] (4) [below left=of 5] {};
    \node [draw,circle,fill=red!50,"$3$" below] (3) [below=of 4] {};
      \node [draw,circle,fill=red!50,"$1$" below] (1) [below=of 7] {};
      \node (a) [left=of 5]{};
      \node (b) [right=of 6]{};
      \node  [draw,circle,fill=blue!50,"$8$" above]  (8) at (0.9,1.2) {};
      \node [draw,circle,fill=blue!50,"$9$" below] (9) at (-2,-4){};

 \node [draw,circle,fill=red!50,"$2$" below] (2) at (1,-4) {};
  
  \begin{scope}[>=Latex]
  \draw[-> , thick]  (5) edge (6);  
   \draw[-> , thick]  (4) edge (7); 
    \draw[-> , thick]  (3) edge (1); 
    
     \draw[-> , thick]  (4) edge (1); 
      \draw[-> , thick]  (1) edge (5); 
       \draw[-> , thick]  (1) edge (2); 
       
        \draw[-> , thick]  (7) edge (5); 
         \draw[-> , thick]  (7) edge (2); 
          \draw[-> , thick]  (4) edge (7); 
          \draw[-> , thick]  (3) edge (7); 
          
          \draw[-> , thick]  (6) edge (4); 
          \draw[-> , thick]  (6) edge (3); 
          \draw[-> , thick]  (2) edge (6); 
          
          \draw[-> , thick]  (5) edge (2); 
          \draw[-> , thick]  (2) edge (3); 
          \draw[-> , thick]  (2) edge (4); 
          
           \draw[-> , thick]  (5) edge (8); 
           \draw[-> , thick]  (8) edge (6); 
            \draw[-> , thick]  (8) edge (3);
             \draw[-> , thick]  (8) edge (1);  
              \draw[-> , thick]  (4) edge[bend left=25] (8); 
               \draw[-> , thick]  (7) edge[bend right=25] (8); 
               
                \draw[-> , thick]  (9) edge (4); 
                 \draw[-> , thick]  (1) edge (9);

    \end{scope}

\end{tikzpicture}
}
\end{center}
\caption{Quiver associated with special Somos-7 recurrence \eqref{specialsomos7} }\label{Qsomos7}
\end{figure}
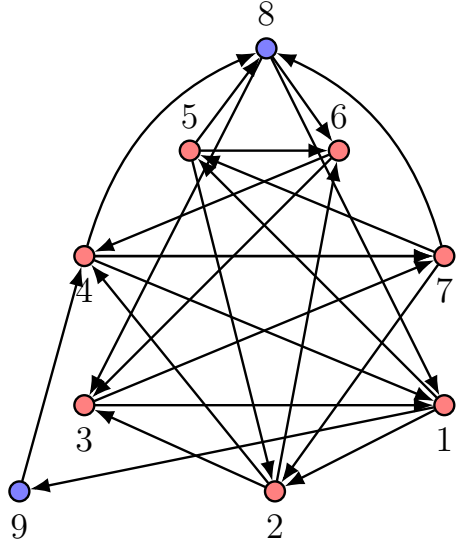

There is another form of cluster map which generates the same sequence of cluster variables as the Lyness recurrence. The iteration of the map is equivalent to the system of difference equations
\begin{equation}
\begin{split}
x_{1,n+1}x_{1,n}& = 1 + x_{2,n}\\
x_{2,n+1}x_{2,n}& = 1 + x_{1,n+1}\\
\end{split}
\end{equation}
Choosing a different deformation (compared to \eqref{lynessdeformed})
\begin{equation}\label{typeA2map}
\begin{split}
x_{1,n+1}x_{1,n}& = 1 + a_1x_{2,n}\\
x_{2,n+1}x_{2,n}& = 1 + a_2x_{1,n+1}\\
\end{split}
\end{equation}
gives rise a deformed map which preserves same symplectic form as Lyness map and yields a discrete integrable system. Under Laurentification, a variable transformation consisting a different form from that of the Lyness map, the deformed map is lifted to cluster algebra whose quiver is identical to the quiver in Figure \ref{Qsomos7} but with its frozen vertices connected differently (shown in Figure \ref{QA2first}). 
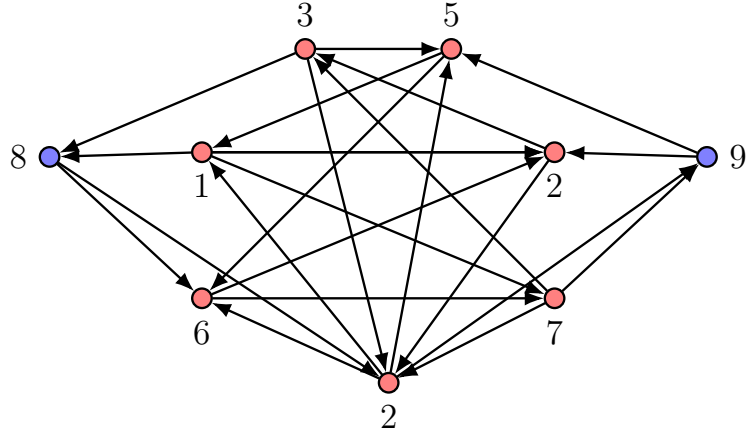
\begin{figure}[h!]
\begin{center}
\resizebox{0.6\textwidth}{!}{%
 \begin{tikzpicture}[every circle node/.style={draw,scale=0.6,thick},node distance=15mm]
  \node [draw,circle,fill=red!50,"$3$"] (5) at (0,0) {};
  \node [draw,circle,fill=red!50,"$5$"] (6) [right=of 5] {};
  \node [draw,circle,fill=red!50,"$2$" below] (7) [below right=of 6] {};
  \node [draw,circle,fill=red!50,"$1$" below] (4) [below left=of 5] {};
    \node [draw,circle,fill=red!50,"$6$" below] (3) [below=of 4] {};
      \node [draw,circle,fill=red!50,"$7$" below] (1) [below=of 7] {};
      \node (a) [left=of 5]{};
      \node (b) [right=of 6]{};
      \node  [draw,circle,fill=blue!50,"$8$" left]  (8) [below left=of a]{};
      \node [draw,circle,fill=blue!50,"$9$" right] (9) [below right=of b]{};

 \node [draw,circle,fill=red!50,"$2$" below] (2) at (1,-4) {};
  
  \begin{scope}[>=Latex]
  \draw[-> , thick]  (5) edge (6);  
   \draw[-> , thick]  (4) edge (7); 
    \draw[-> , thick]  (3) edge (1); 
    
     \draw[-> , thick]  (4) edge (1); 
      \draw[-> , thick]  (1) edge (5); 
       \draw[-> , thick]  (1) edge (2); 
       
        \draw[-> , thick]  (7) edge (5); 
         \draw[-> , thick]  (7) edge (2); 
          \draw[-> , thick]  (4) edge (7); 
          \draw[-> , thick]  (3) edge (7); 
          
          \draw[-> , thick]  (6) edge (4); 
          \draw[-> , thick]  (6) edge (3); 
          \draw[-> , thick]  (2) edge (6); 
          
          \draw[-> , thick]  (5) edge (2); 
          \draw[-> , thick]  (2) edge (3); 
          \draw[-> , thick]  (2) edge (4); 
          
          \draw[-> , thick]  (8) edge (2); 
          \draw[-> , thick]  (8) edge (3); 
          \draw[-> , thick]  (4) edge (8); 
          \draw[-> , thick]  (5) edge (8); 
          
          \draw[-> , thick]  (9) edge (6); 
           \draw[-> , thick]  (9) edge (7);
            \draw[-> , thick]  (2) edge (9);  
             \draw[-> , thick]  (1) edge (9);

    \end{scope}

\end{tikzpicture}
}
\end{center}
\caption{The extended quiver $Q_{A_{2}}$ constructed through Laurentification of \eqref{typeA2map} }\label{QA2first}
\end{figure}

In the one of previous work \cite{grab}, we showed that the quiver arising from deformation of type $A_{4}$ cluster, can be arranged in the form of $Q_{A_{4}}$ ((b) in Figure \ref{Q2toQ4}). This configuration has structure that allows the $Q_{A_{2}}$ to be extended into $Q_{A_{4}}$ as shown in Figure  \ref{Q2toQ4}  by insertion of particular quiver, referred to as local expansion, illustrated in the figure \ref{localexpansionA2toA4}. 
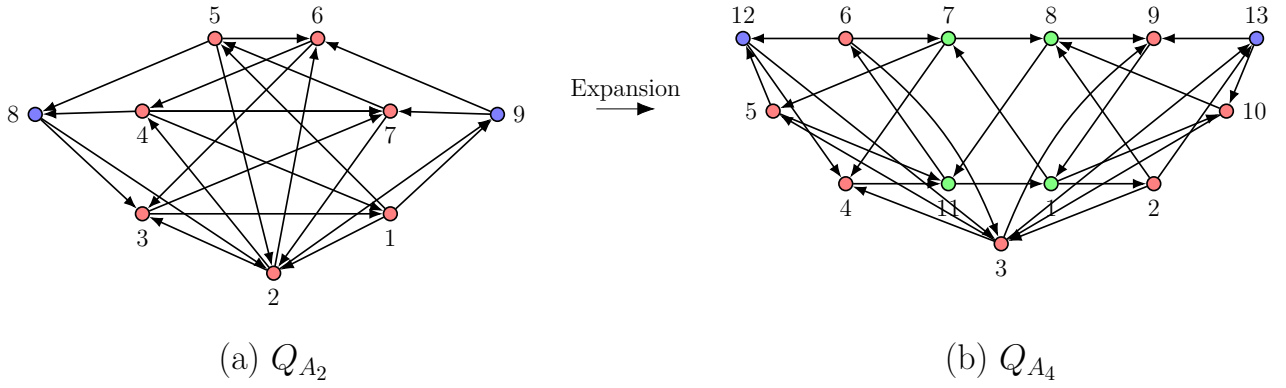
\begin{figure}[h!]
\begin{center}
\resizebox{1 \textwidth}{!}{%
 \begin{tikzpicture}[every circle node/.style={draw,scale=0.6,thick},node distance=15mm]

  \node [draw,circle,fill=red!50,"$5$"] (5) at (0,0) {};
  \node [draw,circle,fill=red!50,"$6$"] (6) [right=of 5] {};
  \node [draw,circle,fill=red!50,"$7$" below] (7) [below right=of 6] {};
  \node [draw,circle,fill=red!50,"$4$" below] (4) [below left=of 5] {};
    \node [draw,circle,fill=red!50,"$3$" below] (3) [below=of 4] {};
      \node [draw,circle,fill=red!50,"$1$" below] (1) [below=of 7] {};
      \node (a) [left=of 5]{};
      \node (b) [right=of 6]{};
      \node  [draw,circle,fill=blue!50,"$8$" left]  (8) [below left=of a]{};
      \node [draw,circle,fill=blue!50,"$9$" right] (9) [below right=of b]{};

 \node [draw,circle,fill=red!50,"$2$" below] (2) at (1,-4) {};
 \node (a) at (1,-5.5) {\Large (a) ${ Q_{A_{2}}}$};
  
  \begin{scope}[>=Latex]
  \draw[-> , thick]  (5) edge (6);  
   \draw[-> , thick]  (4) edge (7); 
    \draw[-> , thick]  (3) edge (1); 
    
     \draw[-> , thick]  (4) edge (1); 
      \draw[-> , thick]  (1) edge (5); 
       \draw[-> , thick]  (1) edge (2); 
       
        \draw[-> , thick]  (7) edge (5); 
         \draw[-> , thick]  (7) edge (2); 
          \draw[-> , thick]  (4) edge (7); 
          \draw[-> , thick]  (3) edge (7); 
          
          \draw[-> , thick]  (6) edge (4); 
          \draw[-> , thick]  (6) edge (3); 
          \draw[-> , thick]  (2) edge (6); 
          
          \draw[-> , thick]  (5) edge (2); 
          \draw[-> , thick]  (2) edge (3); 
          \draw[-> , thick]  (2) edge (4); 
          
          \draw[-> , thick]  (8) edge (2); 
          \draw[-> , thick]  (8) edge (3); 
          \draw[-> , thick]  (4) edge (8); 
          \draw[-> , thick]  (5) edge (8); 
          
          \draw[-> , thick]  (9) edge (6); 
           \draw[-> , thick]  (9) edge (7);
            \draw[-> , thick]  (2) edge (9);  
             \draw[-> , thick]  (1) edge (9);

    \end{scope}

\draw [-{Latex[length=3mm]}] (6.5,-1.2) -- (7.5,-1.2) node[midway,sloped,above] {Expansion};

  \node [draw,circle,fill=blue!50,"$12$"] (12) at (9,0) {};
  
     \node [draw,circle,fill=red!50,"$6$"] (6) [right= of 12] {};
      \node [draw,circle,fill=green!50,"$7$"] (7) [right=of 6] {};
      \node [draw,circle,fill=green!50,"$8$"] (8) [right=of 7] {};
      \node [draw,circle,fill=red!50,"$9$"] (9) [right=of 8] {};
      \node [draw,circle,fill=red!50,"$10$"right] (10) [below right=of 9] {};
      \node [draw,circle,fill=red!50,"$5$" left] (5) [below left=of 6] {};

      \node [draw,circle,fill=blue!50,"$13$"] (13) [right=of 9] {};

       \node [draw,circle,fill=red!50,"$4$"below] (4) [below right=of 5] {};
       \node [draw,circle,fill=green!50,"$11$"below] (11) [right=of 4] {};
       \node [draw,circle,fill=green!50,"$1$"below] (1) [right=of 11] {};
       \node [draw,circle,fill=red!50,"$2$"below] (2) [right=of 1] {};
       
        \node [draw,circle,fill=red!50,"$3$" below] (3) at (13.4,-3.5) {};
       
       \node (a) at (13.4,-5.5) {\Large (b) $Q_{A_{4}}$};

  \begin{scope}[>=Latex]
  
  \draw[-> , thick]  (1) edge (2); 
  \draw[-> , thick]  (1) edge (7); 
  \draw[-> , thick]  (9) edge (1);
  \draw[-> , thick]  (11) edge (1);
   \draw[-> , thick]  (1) edge (10);

 \draw[-> , thick]  (2) edge (3);
 \draw[-> , thick]  (2) edge (8);
 \draw[-> , thick]  (2) edge (13);

 \draw[-> , thick]  (4) edge (11);
  \draw[-> , thick]  (12) edge (4);
   \draw[-> , thick]  (7) edge (4);
    \draw[-> , thick]  (3) edge (4);
    
 \draw[-> , thick]  (5) edge (12);
  \draw[-> , thick]  (5) edge (11);
   \draw[-> , thick]  (3) edge (5);
    \draw[-> , thick]  (7) edge (5);
   
   \draw[-> , thick]  (6) edge (12);
    \draw[-> , thick]  (6) edge (7);
     \draw[-> , thick]  (11) edge (6);
      \draw[-> , thick]  (6) edge[bend left= 15] (3);
      
    \draw[-> , thick]  (7) edge (8);
    
     \draw[-> , thick]  (8) edge (11);
      \draw[-> , thick]  (8) edge (9);
       \draw[-> , thick]  (10) edge (8);
       
        \draw[-> , thick]  (13) edge (9);
         \draw[-> , thick]  (3) edge[bend left=15] (9);
         
          \draw[-> , thick]  (13) edge (10);
           \draw[-> , thick]  (10) edge (3);
           
            \draw[-> , thick]  (3) edge (13);
            
             \draw[-> , thick]  (12) edge (3);

    \end{scope}

\end{tikzpicture}
}
\end{center}
\caption{Extension from $Q_{A_{2}}$ to $Q_{A_{4}}$. Green nodes are new vertices which are inserted in $Q_{A_{2}}$ in the form of local expansion (shown in the figure \ref{localexpansionA2toA4})}
\label{Q2toQ4}
\end{figure}

  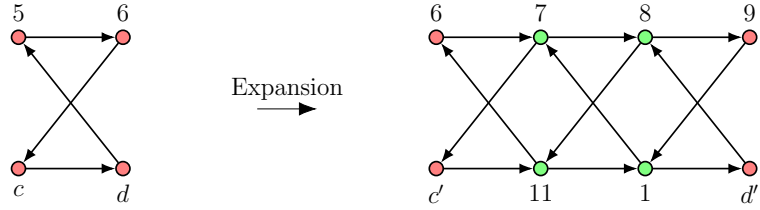
\begin{figure}[h!]
\begin{center}
\resizebox{0.6\textwidth}{!}{%
 \begin{tikzpicture}[every circle node/.style={draw,scale=0.6,thick},node distance=15mm]

  \node [draw,circle,fill=red!50,"$5$"] (5) at (0,0) {};
   \node [draw,circle,fill=red!50,"$6$"] (6)[right=of 5]{};
  
  \node [draw,circle,fill=red!50,"$c$"below] (4) at (0,-2.2) {};
   \node [draw,circle,fill=red!50,"$d$"below] (7) [right=of 4] {};
    
  \begin{scope}[>=Latex]
            
       \draw[-> , thick]  (5) edge (6);
        \draw[-> , thick]  (6) edge (4);
         \draw[-> , thick]  (4) edge (7);
          \draw[-> , thick]  (7) edge (5);

    \end{scope}
\draw [-{Latex[length=3mm]}] (4,-1.2) -- (5,-1.2) node[midway,sloped,above] {Expansion };
    \node [draw,circle,fill=red!50,"$6$"] (6) at (7,0) {};
  \node [draw,circle,fill=red!50,"$c'$" below] (4) at (7,-2.2) {};
   \node [draw,circle,fill=green!50,"$7$"] (7) [right=of 6] {};
  \node [draw,circle,fill=green!50,"$11$"below] (11) [right=of 4 ]  {};
   \node [draw,circle,fill=green!50,"$8$"] (8) [right=of 7] {};
  \node [draw,circle,fill=green!50,"$1$" below] (1) [right=of 11 ]  {};
    \node [draw,circle,fill=red!50,"$9$"] (9) [right=of 8] {};
  \node [draw,circle,fill=red!50,"$d'$" below] (2) [right=of 1 ]  {};

  \begin{scope}[>=Latex]
   
    \draw[-> , thick]  (7) edge (8);
     \draw[-> , thick]  (8) edge (11);
      \draw[-> , thick]  (11) edge (1);
       \draw[-> , thick]  (1) edge (7);
       
        \draw[-> , thick]  (8) edge (9);
         \draw[-> , thick]  (9) edge (1);
          \draw[-> , thick]  (1) edge (2);
           \draw[-> , thick]  (2) edge (8);
       
        \draw[-> , thick]  (6) edge (7);
         \draw[-> , thick]  (7) edge (4);
          \draw[-> , thick]  (4) edge (11);
           \draw[-> , thick]  (11) edge (6);

    \end{scope}

\end{tikzpicture}
}
\caption{Local expansion of the subquiver in $Q_{A_{2}}$. The nodes $c = \qty{3,4}$, $d=\qty{1,7}$, $c'=\qty{4,5}$ and $d'=\qty{2,10}$  }
\label{localexpansionA2toA4}
\end{center}
\end{figure}

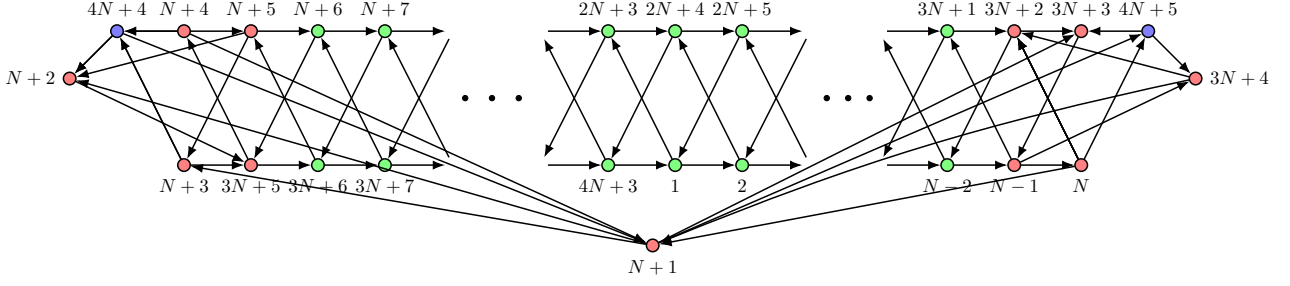
\begin{figure}[H]

\begin{center}
\resizebox{1\textwidth}{!}{%
 \begin{tikzpicture}[every circle node/.style={draw,scale=0.6,thick},node distance=10mm]

\node [draw,circle,fill=blue!50,"\footnotesize{$4N+4$}"] (aa1) at (0,0) {};
  \node [draw,circle,fill=red!50,"\footnotesize{$N+4$}"] (aa2) [right=of aa1] {};
  \node [draw,circle,fill=red!50,"\footnotesize{$N+5$}"] (aa3) [right=of aa2] {};
  \node [draw,circle,fill=green!50,"\footnotesize{$N+6$}"] (aa4) [right=of aa3] {};
  \node [draw,circle,fill=green!50,"\footnotesize{$N+7$}"] (aa5) [right=of aa4] {};
  \node [draw,circle,fill=red!50,"\footnotesize{$N+2$}" left]  (cc1) [below left=of aa1] {};
  \node  (aa6) [right=of aa5] {};

  \node [draw,circle,fill=red!50,"\footnotesize{$N+3$}" below] (bb1) at (1.25,-2.5) {} ; 
  \node [draw,circle,fill=red!50,"\footnotesize{$3N+5$}" below] (bb2) [right=of bb1] {};
  \node [draw,circle,fill=green!50,"\footnotesize{$3N+6$}" below] (bb3) [right=of bb2] {};
   \node [draw,circle,fill=green!50,"\footnotesize{$3N+7$}" below] (bb4) [right=of bb3] {};
   \node (bb5) [right=of bb4] {};

    \node [draw,circle,fill=red!50,"\footnotesize{$N+1$}" below] (dd1) at (10,-4) {} ;
    
  \begin{scope}[>=Latex]
  
\draw[-> , thick]  (aa1) edge (cc1);  
\draw[-> , thick]  (aa1) edge (dd1);  
\draw[-> , thick]  (aa2) edge (aa1);  
\draw[-> , thick]  (bb1) edge (aa1);  

\draw[-> , thick]  (cc1) edge (bb2);  
\draw[-> , thick]  (dd1) edge (cc1);  
\draw[-> , thick]  (aa1) edge (cc1);  
\draw[-> , thick]  (aa3) edge (cc1);  

\draw[-> , thick]  (bb1) edge (aa1);  
\draw[-> , thick]  (bb1) edge (bb2);  
\draw[-> , thick]  (dd1) edge (bb1);  
\draw[-> , thick]  (aa3) edge (bb1);  

\draw[-> , thick]  (aa2) edge (aa1);  
\draw[-> , thick]  (aa2) edge (dd1);  
\draw[-> , thick]  (aa2) edge (aa3);  
\draw[-> , thick]  (bb2) edge (aa2); 

\draw[-> , thick]  (aa3) edge (aa4);  
\draw[-> , thick]  (bb3) edge (aa3);  

\draw[-> , thick]  (aa4) edge (bb2);  
\draw[-> , thick]  (aa4) edge (aa5);  
\draw[-> , thick]  (bb4) edge (aa4); 

\draw[-> , thick]  (aa5) edge (bb3);
\draw[-> , thick]  (aa5) edge (aa6);

\draw[-> , thick]  (bb2) edge (bb3); 
\draw[-> , thick]  (bb3) edge (bb4);
\draw[-> , thick]  (bb4) edge (bb5);

\draw[-> , thick]  (aa6) edge (bb4);
\draw[-> , thick]  (bb5) edge (aa5);

\end{scope}

 \node [draw,circle,fill=green!50,"\footnotesize{$3N+1$}"] (a3) at (15.5,0) {};
  \node [draw,circle,fill=red!50,"\footnotesize{$3N+2$}"] (a4) [right=of a3] {};
  \node [draw,circle,fill=red!50,"\footnotesize{$3N+3$}"] (a5) [right=of a4] {};
  \node [draw,circle,fill=blue!50,"\footnotesize{$4N+5$}"] (a6) [right=of a5] {};
  \node [draw,circle,fill=red!50,"\footnotesize{$3N+4$}" right] (c2) [below right=of a6] {};
  \node (a1) [left=of a3] {};
   \node (a0) [left=of a1] {};
   \node (A0) [draw,circle,fill=green!50,"\footnotesize{$2N+5$}"] [left=of a0] {};
   \node (A1) [draw,circle,fill=green!50,"\footnotesize{$2N+4$}"][left=of A0] {};
   \node (A2) [draw,circle,fill=green!50,"\footnotesize{$2N+3$}"][left=of A1] {};
   \node (A3)[left=of A2] {};

  \node [draw,circle,fill=green!50,"\footnotesize{$N-2$}" below](b2) at (15.5,-2.5) {};
  \node (b1) [left=of b2] {};
  \node (b0) [left=of b1] {};
  \node (B0)[draw,circle,fill=green!50,"\footnotesize{$2$}" below] [left=of b0] {};
  \node (B1)[draw,circle,fill=green!50,"\footnotesize{$1$}" below] [left=of B0] {};
  \node (B2)[draw,circle,fill=green!50,"\footnotesize{$4N+3$}" below] [left=of B1] {};
  \node (B3) [left=of B2] {};
  \node [draw,circle,fill=red!50,"\footnotesize{$N-1$}" below] (b3) [right=of b2] {};
   \node [draw,circle,fill=red!50,"\footnotesize{$N$}" below] (b4) [right=of b3] {};

\draw[decorate sep={1mm}{5mm},fill] (6.5,-1.25) -- (7.5,-1.25);
\draw[decorate sep={1mm}{4mm},fill] (13.25,-1.25) -- (14.25,-1.25);

  \begin{scope}[>=Latex]

    \draw[-> , thick]  (a1) edge  (a3);
   \draw[-> , thick]  (a3) edge  (a4);
    \draw[-> , thick]  (a4) edge  (a5);
    \draw[-> , thick]  (a6) edge  (a5);

     \draw[-> , thick]  (b1) edge  (b2);
   \draw[-> , thick]  (b2) edge  (b3);
   \draw[-> , thick]  (b3) edge  (b4);

 \draw[-> , thick]  (b2) edge  (a1);
   \draw[-> , thick]  (a4) edge  (b2);

 \draw[-> , thick]  (b3) edge  (a3);
   \draw[-> , thick]  (b3) edge  (c2);
    \draw[-> , thick]  (a5) edge  (b3);

     \draw[-> , thick]  (b4) edge  (dd1);
   \draw[-> , thick]  (b4) edge  (a4);
   \draw[-> , thick]  (b4) edge  (a4);
\draw[-> , thick]  (b4) edge  (a6);

\draw[-> , thick]  (c2) edge  (a4);
\draw[-> , thick]  (c2) edge[bend right=5]  (dd1);
\draw[-> , thick]  (b4) edge  (a4);
\draw[-> , thick]  (a6) edge  (c2);
    
   \draw[-> , thick]  (dd1) edge  (a5); 
   \draw[-> , thick]  (b4) edge  (a4);

\draw[-> , thick]  (dd1) edge  (a6); 

\draw[-> , thick]  (a3) edge  (b1); 

\draw[-> , thick]  (B0) edge  (b0); 
\draw[-> , thick]  (B1) edge  (B0);
\draw[-> , thick]  (B2) edge  (B1); 
\draw[-> , thick]  (B3) edge  (B2);

\draw[-> , thick]  (b0) edge  (A0); 
\draw[-> , thick]  (a0) edge  (B0);
\draw[-> , thick]  (B0) edge  (A1); 
\draw[-> , thick]  (A0) edge  (B1);
\draw[-> , thick]  (B1) edge  (A2); 
\draw[-> , thick]  (A1) edge  (B2);
\draw[-> , thick]  (B2) edge  (A3); 
\draw[-> , thick]  (A2) edge  (B3);

\draw[-> , thick]  (A3) edge  (A2); 
\draw[-> , thick]  (A2) edge  (A1); 
\draw[-> , thick]  (A1) edge  (A0); 
\draw[-> , thick]  (A0) edge  (a0);

    \end{scope}

\end{tikzpicture}
}
\end{center}
 \caption{Quiver, $Q_{A_{2N}}$ associated to deformation of type $A_{2N}$ cluster map \cite{grab}. The red vertices indicate the nodes that already existed in $Q_{A_{4}}$ while the green vertices represent those newly included through local expansion. }
 \label{typeA2N}
\end{figure}
By applying this expansion recursively, one can build class of quiver $Q_{A_{2N}}$ (represented in Figure \ref{typeA2N}), together with associated family of parametric cluster map $\psi_{A_{2N}}$. As a result, Laurent property of the deformed type $A_{2N}$ maps is restored by defining it on the cluster algebra of higher rank.

In this paper, we follow a similar approach used in \cite{grab} to construct a family of quivers that yields a two-parameter cluster map associated with the deformation of a type $D_{2N}$ cluster map.  The paper is outlined as follows. 

In Section \ref{s:Preliminaries1}, we begin by reviewing a brief background of cluster algebras and consider the cluster algebra of type $D_{4}$ to introduce features such as  Laurent property, mutation periodicity and Zamolodchikov periodicity.  We then review the discrete integrable system in the cluster algebra. 

In Section \ref{s:DeformD4}, we review the integrable deformation of the type $D_{4}$ cluster map from  \cite{hkm24}. Under the parameter conditions for integrability, the map reduces to two distinct parametric Liouville integrable maps on the plane, each preserving a different first integral. Both maps admit Laurentifications, which are inequivalent to each other, that lift to cluster algebras of rank $8$. 
The two quivers share the same configuration for mutable nodes, but the frozen nodes are connected differently. 

In Section  \ref{deformedtypeD6} we examine the deformation of the cluster map associated with type $D_{6}$ and show that it admits Laurentification, which lifts to a cluster algebra of rank 12 with quiver consisting of two frozen nodes. Using the max-plus form of the exchange relations, we determine the degree growth and compute the \textit{algebraic entropy}, a standard criterion for integrability.

Subsequently, in Section \ref{S:localexpand}, we compare quiver induced in type $D_{6}$ case with one of the extended quiver arising from the deformation in cluster algebra of type $D_{4}$ and show that they are related through the insertion of specific subquiver in the form of local expansion. Iterating this expansion yields a quiver with $4N$ nodes with two frozen nodes, which corresponds to the Laurentification of the deformed type $D_{2N}$ map.

In section \ref{S: D2Nmap}, we prove that the extended quiver, constructed in the previous section, is mutation periodic up to a specific permutation of labelling of nodes, thereby  defining the cluster map associated with the deformed type $D_{2N}$ map. We further show that, for all $N\geq 1$, degree growth of the map is quadratic.

\section{Preliminaries}\label{s:Preliminaries1}

Cluster algebra of rank $n$ is generated by distinguished generators, known as \textit{cluster variables}, induced by applying mutation recursively starting from initial seed, $(\vb{x},Q)$ consisting of a set of cluster variables $\vb{x} = (x_{1},\dots, x_{n})$ and quiver $Q$, that is, finite directed graph with $n$ vertices, $Q$, which do not does not contain any loops (1-cycles) or 2-cycles. The mutation $\mu_{k}$ applies to each in the seed separately, $\mu_{k}(\vb{x},Q) = (\mu_{k}(\vb{x}),\mu_{k}(Q)) = (\vb{x}',Q')$, following individual transformation rules. Let us first consider \textit{quiver mutation}.

\textit{Quiver mutation}: The mutation $\mu_{k}(Q)$ at the node $k$ gives rises to new quiver $Q'$ by taking  following rule as shown below, 
\begin{enumerate}

\item For each path  $i \xrightarrow{r} k \xrightarrow{s} j $ in $Q$, multiply  the number of incoming and outgoing arrows at the node $k$ and insert the multiplied number of edges between $i$ and $j$,

\begin{center}
\resizebox{0.5\textwidth}{!}{%
 \begin{tikzpicture}[every circle node/.style={draw,scale=0.6,thick},node distance=15mm]
 
\draw [-{Latex[length=3mm]}] (4.5,0) -- (5.5,0) node[midway,sloped,above] {$\mu_{k}$};
    \node [draw,circle,fill=red!50,"$i$"below] (1) at (0,0) {};
  \node [draw,circle,fill=blue!50,"$k$" below] (2)[right=of 1] {};
   \node [draw,circle,fill=red!50,"$j$"below] (3) [right=of 2] {};
  
  \begin{scope}[>=Latex]
  
    \draw[-> , thick]  (1) edge["r"] (2);
    \draw[-, thick]  (2) edge["s"] (3);

    \end{scope}

  \node [draw,circle,fill=red!50,"$i$" below] (1) at (7,0) {};
  \node [draw,circle,fill=blue!50,"$k$"below ] (2)[right=of 1] {};
   \node [draw,circle,fill=red!50,"$j$" below] (3) [right=of 2] {};
  
  \begin{scope}[>=Latex]
  
    \draw[-> , thick,]  (1) edge["r"] (2);
    \draw[->, thick]  (2) edge["s"] (3);    
    \draw[->, thick]  (1) edge["$r\cdot s$",bend left=35] (3);  
    
   \end{scope}  
\end{tikzpicture}
}
\end{center}

\item Reverse edge  which are incident to node $k$

\item Remove any 2-cycles formed by step 1.

\end{enumerate}

The quiver with $n$ nodes can be encoded into $n\times n$ skew-symmetric matrix, $B\in \Mat(\mathbb{Z}) $ in which each entry $b_{ij}$ represents the number of edges from $i$ to $j$ i.e.
\begin{center}
\resizebox{0.23\textwidth}{!}{%
 \begin{tikzpicture}[every circle node/.style={draw,scale=0.6,thick},node distance=15mm]

    \node [draw,circle,fill=red!50,"$i$"] (1) at (0,0) {};
  \node [draw,circle,fill=red!50,"$j$" ] (2)[right=of 1] {};
  
  \node(a) at (-1,0) {$b_{ij}=$};
  
  \begin{scope}[>=Latex]
  
    \draw[-> , thick]  (1) edge (2);
 \end{scope}

\end{tikzpicture}
}
\end{center}
The corresponding matrix is known as \textit{exchange matrix}. The quiver mutation  $\mu_{k}$ at the node $k$, can be formulated by using the entries of the exchange matrix $B$ as shown below. 

\textit{Matrix mutation} $\mu_{k}(B)$: The new exchange matrix $B' = \mu_{k}(B)$  is given by 

  \begin{equation}\label{matrixmu}
       (\mu_{k}(B))_{ij} = (B')_{ij}= b'_{ij} =
        \begin{cases}
            -b_{ij} & \text{if} \ i= k \ \text{or} \ j=k  \\ 
            b_{ij} + \frac{1}{2} \qty(\abs{b_{ik}}b_{kj} + b_{ik}\abs{b_{kj}})& \text{otherwise}  \\ 
        \end{cases}
    \end{equation} 

Together with quiver (or matrix mutation), cluster variable $x_{k}$ in the cluster $\vb{x}$ is updated to new variable which is given by the equation below. 

 \textit{Cluster mutation} $\mu_{k}(\vb{x})$: The \textit{cluster mutation} of $\vb{x}$ in direction $k$, $\mu_{k}(\vb{x}) = \qty(x_{1},\dots,x_{k-1},x'_{k},\dots x_{n})$, give rises to new cluster variable $x'_{k}$ which is given by following expression,
    \begin{equation}\label{mu1}
    \mu_{k}(x_{k}) =x'_{k} =  \dfrac{1}{x_{k}}\qty( \prod^{n}\limits_{\substack{j=1 \\ b_{jk} > 0}}x^{b_{jk}}  + \prod^{n}\limits_{\substack{j=1 \\ b_{jk} < 0}}x^{-b_{jk}} )  
\end{equation}
This expression is known as a (coefficient free) \textit{exchange relation}. Note that the mutation is an involution such that the quiver remains unchanged when the same mutation is applied twice in the same direction i.e. $\mu_{k}\cdot \mu_{k} (\vb{x},Q)= \mu^{2}_{k}(\vb{x},Q) = (\vb{x},Q)$. The sequence of mutations starting from the initial seed $(\vb{x}_{0},B_{0})$,
\begin{align*}
(\vb{x}_{0},B_{0}) \xmapsto{\mu_{i_{1}}} (\vb{x}_{1},B_{1}) \xmapsto{\mu_{i_{2}}} (\vb{x}_{2},B_{2}) \xmapsto{\mu_{i_{3}}} \cdots \xmapsto{\mu_{i_{m}}}(\vb{x}_{m},B_{m})
\end{align*}
produces new cluster variables and quiver, resulting in construction of the cluster algebra.
\begin{defn}[Cluster algebra] Let $\mathcal{F}$ be the field of rational functions in $n$ independent variables $x_{1},\dots,x_{n}$ over $\mathbb{C}$ and set $\vb{x} = (x_{1},\dots,x_{n}) \in \mathcal{F}^{n}$. Let $B$ be $n\times n$ skew-symmetrizable matrix. Then cluster algebra of rank $n$,  $\mathcal{A}(\vb{x},B)$, is set of distinguished cluster variables in all of seeds that are reachable by sequence of mutations starting from the initial seed.
\end{defn}

One can extend the notion of cluster algebra by  incorporating non-mutable variables, \textit{frozen variables}, in cluster, together with their corresponding \textit{frozen nodes} in the quiver. As the name suggests, they never mutate. For the precise detail, let us consider the cluster algebra of rank $n$ with initial cluster $\vb{x}= (x_1,x_2,\dots, x_n)$ and initial quiver $Q$ with $n$ mutable nodes. To generalise the structure, we introduce the $m$ frozen variables $\tilde{x}={n+1},\tilde{x}={n+2},\cdots ,\tilde{x}={n+m}$ and corresponding $m$ frozen nodes. These variables are added to cluster forming  extended cluster \textit{extended cluster}  $\tilde{\vb{x}}=(x_{1},x_{2},\dots,x_{n},x_{n+1},\dots,x_{n+m})$ and the quiver is extended to $\tilde{Q}$ with $n+m$ nodes, where $m$ nodes are designated as frozen. Then the extended cluster algebra can be built from the extended initial seed $(\tilde{\vb{x}},\tilde{Q})$  by applying a sequence of mutations 
\begin{equation}
    x_{k}'x_{k} = \alpha_{k}\prod^{n}\limits_{\substack{j=1 \\ b_{jk} > 0}}x_{j}^{b_{jk}}  + \beta_{k}\prod^{n}\limits_{\substack{j=1 \\ b_{jk} < 0}}x_{j}^{-b_{jk}} 
\end{equation}
where the coefficients are 
\begin{equation}\label{coeffvar}
    \alpha_{k} = \prod^{n+m}\limits_{\substack{j=n+1 \\ b_{j,k} > 0}}x_{j}^{b_{j,k}}, \quad \beta_{k} = \prod^{n+m}\limits_{\substack{j=n+1 \\ b_{j,k} < 0}}x_{j}^{-b_{j,k}}.
\end{equation}
Note that mutation rule for the extended quiver is same as previously given. We refer to these more general cluster algebras as \textit{cluster algebras of geometric type}.

Let us consider the example of coefficient free cluster algebra. 

\begin{exmp}\label{TypeA2ex} Begin with the initial seed $(\vb{x},B)$, where $\vb{x} = (x_{1},x_{2})$ and the exchange matrix  
    \begin{align*}
    B = \mqty(0 & 1 \\ -1 & 0)
    \end{align*}
The mutation in the direction 1, $\mu_{1}(\vb{x},B) = (\vb{x}',B')$, where $\mqty(x_{1}',x_{2}')$
\begin{align*}
\mqty(x_{1}',x_{2}') = \left( \frac{1 + x_{2}}{x_{1}},x_2 \right), \quad \  B' = \mqty(0 & -1 \\ 1 & 0 ) = - B
\end{align*}
Since the mutation returns the initial seed if we apply the mutation in the same direction as previously, the next mutation should be in direction 2. The mutation $\mu_{2}(\vb{x}') =\mqty(x_{1}'',x_{2}'')$, 
\begin{align*}
    &x_{1}'' = x_{1}' \\ 
    &  x_{2}''= \frac{1 + x_{1}'}{x_{2}'} = \frac{1 + \left(  \frac{1 + x_{2}}{x_{1}} \right)}{x_2} = \frac{1 + x_1 + x_2}{x_1 x_2}
\end{align*}
and 
\begin{align*}
    \mu_{2}(B') = B
\end{align*}
Applying mutation in sequence yields the following cluster variables, 
\begin{center}
\resizebox{1\textwidth}{!}{%
 \begin{tikzpicture}[every circle node/.style={draw,scale=0.6,thick},node distance=10mm]

 \draw [-{Stealth[width=2mm,length=1.2mm]}] (1.7,0) -- (2.5,0) node[midway,sloped,above] {$\mu_{1}$};  
 
  \node [draw,circle,fill=red!50,"\footnotesize{$1$}"] (1) at (0,0) {};
  \node [draw,circle,fill=green!50,"\footnotesize{$2$}"] (2) [right=of 1] {};
  \node (1a) at (0.5,-0.3) {$(x_{1},x_{2})$};

  \begin{scope}[>=Latex]
\draw[-> , thick]  (1) edge (2);  
\end{scope}

 \draw [-{Stealth[width=2mm,length=1.2mm]}] (4.7,0) -- (5.5,0) node[midway,sloped,above] {$\mu_{2}$};

\node [draw,circle,fill=red!50,"\footnotesize{$1$}"] (1) at (3,0) {};
  \node [draw,circle,fill=green!50,"\footnotesize{$2$}"] (2) [right=of 1] {};
  \node (1a) at (3.5,-0.5) { $\qty(\frac{1+x_{2}}{x_{1}},x_{2})$};

 \begin{scope}[>=Latex]
\draw[-> , thick]  (2) edge (1);  
\end{scope}

 \draw [-{Stealth[width=2mm,length=1.2mm]}] (7.7,0) -- (8.5,0) node[midway,sloped,above] {$\mu_{1}$};

\node [draw,circle,fill=red!50,"\footnotesize{$1$}"] (1) at (6,0) {};
  \node [draw,circle,fill=green!50,"\footnotesize{$2$}"] (2) [right=of 1] {};
  \node (1a) at (6.5,-0.5) {$\qty(\frac{1+x_{2}}{x_{1}},\frac{1+x_{1} + x_{2}}{x_{1}x_{2}})$};

 \begin{scope}[>=Latex]
\draw[-> , thick]  (1) edge (2);  
\end{scope}

 \draw [-{Stealth[width=2mm,length=1.2mm]}] (10.7,0) -- (11.5,0) node[midway,sloped,above] {$\mu_{2}$};

\node [draw,circle,fill=red!50,"\footnotesize{$1$}"] (1) at (9,0) {};
  \node [draw,circle,fill=green!50,"\footnotesize{$2$}"] (2) [right=of 1] {};
  \node (1a) at (9.5,-0.5) {\small$\qty(\frac{1+x_1}{x_2},\frac{1+x_{1} + x_{2}}{x_{1}x_{2}})$};

 \begin{scope}[>=Latex]
\draw[-> , thick]  (2) edge (1);  
\end{scope}

\node [draw,circle,fill=red!50,"\footnotesize{$1$}"] (1) at (12,0) {};
  \node [draw,circle,fill=green!50,"\footnotesize{$2$}"] (2) [right=of 1] {};
  \node (1a) at (12.5,-0.5) {$\qty(\frac{1+x_1}{x_2}, x_1)$};
  
 \begin{scope}[>=Latex]
\draw[-> , thick]  (1) edge (2);  
\end{scope}

\draw [-{Stealth[width=2mm,length=1.2mm]}] (0.5,-1.5) -- (0.5,-1) node[midway,sloped," $\mu_{2}$ " ,left] {};
 
  \draw [-{Stealth[width=2mm,length=1.2mm]}] (2.5,-2) -- (1.7,-2) node[midway,sloped,above] {$\mu_{1}$};

\node [draw,circle,fill=red!50,"\footnotesize{$1$}"] (1) at (0,-2) {};
  \node [draw,circle,fill=green!50,"\footnotesize{$2$}"] (2) [right=of 1] {};
  \node (1a) at (0.5,-2.5) {$\qty(x_{1},\frac{1+x_{1}}{x_{2}})$};
   
 \begin{scope}[>=Latex]
\draw[-> , thick]  (1) edge (2);  
\end{scope}

\draw [-{Stealth[width=2mm,length=1.2mm]}] (5.5,-2) -- (4.7,-2) node[midway,sloped,above] {$\mu_{2}$};

\node [draw,circle,fill=red!50,"\footnotesize{$1$}"] (1) at (3,-2) {};
  \node [draw,circle,fill=green!50,"\footnotesize{$2$}"] (2) [right=of 1] {};
  \node (1a) at (3.5,-2.5) {\small $\qty(\frac{1+x_1+x_2}{x_1x_2},\frac{1+x_{1}}{x_{2}})$};
  
 \begin{scope}[>=Latex]
\draw[-> , thick]  (1) edge (2);  
\end{scope}

 \draw [-{Stealth[width=2mm,length=1.2mm]}] (8.5,-2) -- (7.7,-2) node[midway,sloped,above] {$\mu_{1}$};

\node [draw,circle,fill=red!50,"\footnotesize{$1$}"] (1) at (6,-2) {};
  \node [draw,circle,fill=green!50,"\footnotesize{$2$}"] (2) [right=of 1] {};
  \node (1a) at (6.5,-2.5) {\small$\qty(\frac{1+x_1+x_2}{x_1x_2},\frac{1+x_{2}}{x_{1}})$};
  
   \begin{scope}[>=Latex]
\draw[-> , thick]  (1) edge (2);  
\end{scope}

 \draw [-{Stealth[width=2mm,length=1.2mm]}] (11.5,-2) -- (10.7,-2) node[midway,sloped,above] {$\mu_{2}$};

 \node [draw,circle,fill=red!50,"\footnotesize{$1$}"] (1) at (9,-2) {};
  \node [draw,circle,fill=green!50,"\footnotesize{$2$}"] (2) [right=of 1] {};
  \node (1a) at (9.5,-2.5) {$\qty(x_{2},\frac{1+x_{2}}{x_{1}})$};
  
 \begin{scope}[>=Latex]
\draw[-> , thick]  (1) edge (2);  
\end{scope}

 \draw [-{Stealth[width=2mm,length=1.2mm]}] (12.5,-1) -- (12.5,-1.5) node[midway,sloped,"$\mu_{1}$", right] {};

 \node [draw,circle,fill=red!50,"\footnotesize{$1$}"] (1) at (12,-2) {};
  \node [draw,circle,fill=green!50,"\footnotesize{$2$}"] (2) [right=of 1] {};
  \node (1a) at (12.5,-2.5) {$\displaystyle \qty(x_2,x_1)$};
  
 \begin{scope}[>=Latex]
\draw[-> , thick]  (2) edge (1);  
\end{scope}

\end{tikzpicture}
}
\end{center}%
One can see that the seed returns to the initial seed after a sequence of 10 mutations. This means that no new cluster variables, except the variables above, can be generated by mutations in any directions. Thus, from the sequence of mutations above, we see that the distinct cluster variables are given by 
$$\text{Distinct cluster variables} = \qty{x_{1},x_{2},\dfrac{1+x_{2}}{x_{1}}, \dfrac{1+x_{1}+x_{2}}{x_{1}x_{2}}, \dfrac{1+x_{1}}{x_{2}}} $$
Hence associated cluster algebra is given by 
\begin{align*}
    \mathcal{A}(\vb{x}) = \mathbb{C}\qty[x_{1},x_{2},\frac{1 + x_{2}}{x_{1}}, \frac{1 + x_{1}}{x_{2}},\frac{1 + x_1 + x_2}{x_1 x_2}]
\end{align*}
\end{exmp}
In the example above, we observe that three distinct phenomenon arise. Each of these corresponds to features of the cluster algebra and  play an important role in the main result of this paper. The first phenomenon is that every generated cluster variables are expressed in Laurent polynomials in initial variables $x_{i}$ for $i=1,\dots n$.  This is known as the \textit{Laurent phenomenon}, stated as follows.
\begin{thm}[Laurent phenomenon (see \cite{FZ2002} )]
Every cluster variable generated by the cluster mutations is in the Laurent polynomial ring in its initial cluster variables. 
\end{thm}
The second phenomenon is that the quiver returns to its initial state in two ways: one is by applying mutation $\mu_{1}$ followed by the $\mu_{2}$ (i.e. $\mu_{2} \cdot \mu_{1}$); the other is by applying $\mu_{1}$ and then perform permutation of the two nodes (i.e. $(1,2)\cdot \mu_{1}$). Such phenomenon is referred as \textit{mutation periodic} up to permutation, stated as following. 
\begin{defn}[Mutation periodic] Let $Q$ be a quiver with $n$ vertices. Then the quiver is \textit{mutation periodic} up to permutation $\rho$ if there exists a sequence of quiver mutations which is equivalent to applying permutation $\rho$ of the labelling of the quiver $Q$ i.e. 
\begin{align*}
            \mu_{i_{m}}\mu_{i_{m-1}}\cdots \mu_{i_{2}}\mu_{i_{1}}(Q) = \rho(Q),
        \end{align*}
where $\rho:(1,2,\dots, N) \to (\rho(1),\rho(2),\dots, \rho(N))$ is corresponding permutation. 
\end{defn}
This notion leads to definition of certain birational map between cluster seeds,
\begin{defn}[Cluster map] Let $(\vb{x},Q)$ be an initial seed with an initial cluster $\vb{x}=(x_{1},\dots, x_{n}) $ and mutation periodic quiver $Q$ i.e.  $\mu_{i_{m}}\cdots\mu_{i_{1}}(Q) = \rho(Q)$ for some $m > 0$ and some permutation $\rho$. Then 
\begin{align*}
	\varphi = \rho^{-1}\mu_{i_{m}}\mu_{i_{m-1}}\cdots \mu_{i_{2}}\mu_{i_{1}}
\end{align*}
is a birational map such that 
\begin{equation}
        \begin{array}{lcrcl}
   & &   \varphi \ : \ \C^n&\to&\C^{n}\\
    & &\ \vb{x}=(x_1,\dots,x_{n})&\mapsto&\vb{x}'=(x_{1}',\dots,x_{n}') \; 
  \end{array}
    \end{equation}
where $\vb{x}'=(x_{1}',\dots,x_{n}')$ is new cluster given by $\varphi(\vb{x})$. Such a map is called a \textit{cluster map}. 
\end{defn}

This implies that if there exists composition of mutations that restores the  associated quiver (or exchange matrix) after permutation of nodes, then it can be regarded as birational map which maps between cluster variables, which is referred to as \textit{cluster map}. This enables the description of discrete dynamical systems within the framework of cluster algebra, which leads to the study of the area of discrete integrable systems.

Finally, the third (and last) phenomenon, which can be obeserved in Example \ref{TypeA2ex}, is that the seed returns to the initial seed under the repeated application of the composition of mutations $\mu_{2}\mu_{1}$. In particular, the cluster algebra associated with finite Dynkin diagram exhibits a remarkable periodicity pheomenon, which is referred to as \textit{Zamolodchikov periodicity}. 

\begin{thm}[Zamolodchikov periodicity \cite{fz2,fz4,fr} ]\label{zam} Let $\mathcal{A}(\vb{x}, Q)$ be a cluster algebra together a quiver $Q$ corresponding to the finite Dynkin diagram. Let $h$ be the Coxeter number of the corresponding finite Dynkin diagram and let $\varphi$ be the composition of cluster mutations that satisfy $\varphi(Q)= Q$. Then there are two cases of periodicity, such as 
\begin{itemize}
    \item period $\frac{h+2}{2}$ 
    \begin{equation}
        \mu^{\frac{h+2}{2}}(\vb{x},Q) = (\vb{x},Q) 
    \end{equation}
    if associated roots $\alpha_{i}$ satisfies $w_{0}(\alpha_{i}) = -\alpha_{i}$ where $\omega_{0}$ is the longest element of the corresponding Weyl group,
    \item period $h+2$,
    \begin{equation}
        \mu^{h+2}(\vb{x},Q) = (\vb{x},Q) 
    \end{equation}
    otherwise.
\end{itemize}

\end{thm}

By combining all three phenomenons we observed above, we can define a cluster map of the Dynkin diagram that preserves the associated quiver and exhibits Zamolodchikov periodicity.  
\begin{equation}
\varphi = (1,2)\cdot \mu_{1} : (x_{1},x_{2}) \to \qty( x_{2}, \dfrac{1+x_{2}}{x_{1}}), \ \text{or} \quad \varphi =  \mu_{2}\mu_{1} : (x_{1},x_{2}) \to \qty(\dfrac{1+x_{2}}{x_{1}}, \dfrac{1+x_{1} + x_{2}}{x_{1}x_{2}})
\end{equation}
satisfying $\varphi^{5} = \id$. We remark that although the two maps may appear distinct, but they produce same set of cluster variables. 

The iteration of cluster maps gives rise to  a discrete dynamical system, which allows one to adopt the notion of Liouville integrable map, introduced in \cite{Shigeru}, by using  the \textit{log-canonical Poisson bracket and (pre) symplectic form}. Both of these structures are compatible with cluster algebra in the sense that they retain their structure under the cluster mutations, i.e.
 \begin{equation}
        \begin{array}{lcrcl}
         & &  (\vb{x},B) &\mapsto& \mu_{k}(\vb{x},B) =(\vb{x}',B') = (x_{i}') \\[1.5em]
  & &  \qty{x_{i},x_{j}} = P_{ij}x_{i}x_{j} &\to& \qty{x'_{i},x'_{j}} = P'_{ij}x'_{i}x'_{j}   \\
   & &  \om = \sum_{i<j} \dfrac{b_{ij}}{x_{i}x_{j}} \dd x_{i} \w \dd x_{j}  & \to&\om' = \sum_{i<j} \dfrac{b'_{ij}}{x'_{i}x'_{j}}\dd x_{i} \w \dd x_{j} 
  \end{array}
    \end{equation}
where $P = (P_{ij})$ is a skew-symmetric matrix, which is directly proportional to inverse of an exchange matrix  $B=(b_{ij})$. The two-form $\om$ is referred to as a \textit{presymplectic form} when $B$ is degenerate, and \textit{symplectic form} otherwise.  Since  the cluster map leaves exchange matrix $B$ invariant, the notion of cluster maps can be extended to poisson (or symplectic) maps. Hence we have following 
\begin{defn} 
Let $\mathcal{A}(\vb{x}, B)$ be a cluster algebra of rank $N$ and let $P$ be skew-symmetric matrix of rank $2r<N$ whose entries are coefficient of the \textit{log-canonical Poisson bracket}, $\qty{x_{i},x_{j}} = P_{ij} x_{i}x_{j}$. Then the $\varphi$ is \textit{ integrable} if there exists 
        \begin{enumerate}
        \item {\small $N-2r$ Casimir functions $\mathcal{C}_{k}$, i.e. $\varphi^{*}(\mathcal{C}_{k}) = \mathcal{C}_{k} $ satisfying $\qty{\mathcal{C}_{k},f(\vb{x})} =0 $ for all function $f(\vb{x})$ }
        \item {\small $r$ first integrals $h_{j}$, $j=1,\dots, r$, $(\varphi^{*}(h_{j}) = h_{j}$ such that $ \qty{h_{i},h_{j}} = 0$}
        \end{enumerate}

\end{defn}
For the non-degenerate case with $\rank(B) = n$, the coefficient matrix of the symplectic form $\omega$ is the inverse of the Poisson matrix,  i.e. $P= B^{-1}$. In this setting, the cluster map is both a Poisson and a symplectic map. On the other hand, in the degenerate case $\rank(B) < n$, the map is Poisson but not symplectic. A Poisson manifold equipped with a degenerate Poisson bracket admits a  foliation by symplectic leaves, which enables one to study the dynamical system on each leaf where the symplectic structure is non-degenerate. In this context, there is a canonical way to reduce a Poisson map to a symplectic map via the symplectic form. The following result is stated as Theorem 2.6 in \cite{FH2013}.

\begin{thm}[Symplectic reduction] \label{symreduction}
    In the degenerate case, $\rank B = 2r < N$, there exists rational map $\pi$,
    \begin{equation}
        \begin{array}{lcrcl}
    & &\pi : \C^N&\to&\C^{2r}\\
    & &\ \vb{x}=(x_1,\dots,x_{N})&\mapsto&\vb{y}=(y_{1},\dots,y_{2r}) \; 
  \end{array}
    \end{equation}
equivalent to
    \begin{equation}\label{reducy}
\begin{split}
    &y_{j}= \vb{x}^{\vb{v^{(j)}}} = \prod_{i} x_{i}^{v^{(j)}_{i}}, \quad \vb{v^{(j)}}\in \Ima B
    \end{split}
\end{equation}
which reduces the $\varphi$ to the symplectic map $\hat{\varphi}$ satisfying the following relation, 
    \begin{align*}
        \pi \cdot \varphi = \hat{\varphi} \cdot \pi
    \end{align*}
depicted by 
    \begin{center}
    \begin{tikzcd}
\mathbb{C}^{N} \arrow[r, "\varphi"] \arrow[d, "\pi"]
& \mathbb{C}^{N} \arrow[d, "\pi"] \\
\mathbb{C}^{2r} \arrow[r,  "\hat{\varphi}"]
&  \mathbb{C}^{2r}
\end{tikzcd}
    \end{center}
Let  $M$ be a matrix  whose first $2r$ rows are $\vb{v^{(i)}} \in  \Ima B$ for $1 \leq i \leq 2r$, and whose remaining  $N-2r$ rows are $\vb{u^{(j)}} \in  \ker B$ for $N-r \leq j \leq N$. The symplectic form $\hat{\omega}$ associated with $\hat{\varphi}$ is expressed by following
\begin{equation}
    \hat{\omega} = \sum\limits_{i<j}\frac{\hat{b}_{ij}}{y_{i}y_{j}}\dd y_{i} \wedge \dd y_{j}.
\end{equation}
The coefficients $\hat{b}_{ij}$ are entries of the new exchange matrix which satisfies 
\begin{equation}\label{sympmatrix}
M^{-T} BM^{-1} = \mqty(\hat{B} & 0 \\ 0& 0),
\end{equation}
and  the pull-back relation $\pi^* \hat{\omega} = \omega$ holds.
\end{thm}

\section{Deformation integrable cluster map associated with type $D_{4}$} \label{s:DeformD4}

In this section, we revisit one of the results in \cite{hkm24}, that is, construction of the two parameter family of integrable map on the plane, which is obtained by deformed mutations in the cluster algebra of type $D_{4}$.

\subsection{Periodic cluster map associated with type $D_{4}$ } 

Let us consider Cartan matrix of Dynkin type $D_{4}$, 
\begin{equation}
    \mqty(2 & -1 & 0 & 0 \\-1 & 2 & -1 & -1 \\ 0 & -1 & 2 & 0 \\ 0 & -1 & 0 & 2)
\end{equation}
The exchange matrix corresponding to the Cartan matrix above is given by 
\begin{equation}
  B= \mqty(0 & -1 & 0 & 0 \\1 & 0 & -1 & -1 \\ 0 & 1 & 0 & 0 \\ 0 & 1 & 0 & 0)
\end{equation}
Starting with setting the initial seed  $\qty(\vb{x},B)$, consisting of initial cluster $\vb{x} = (x_{1},x_{2},x_{3},x_{4})$ and the exchange matrix, we consider the specific sequence of mutations in the cluster algebra $\mathcal{A}((\vb{x},B))$. One can verify that the exchange matrix is mutation periodic under the specific composition of matrix mutations, that is to say, 
\begin{equation}
\mu_{4}\mu_{3}\mu_{2}\mu_{1} (B) = B 
\end{equation}
Under the same sequence of mutations, the seed go through following  
\begin{align*}
(x_{1},x_{2},x_{3},x_{4}) \xmapsto{\mu_{1}} (x_{1}',x_{2},x_{3},x_{4}) \xmapsto{\mu_{2}} (x_{1}',x_{2}',x_{3},x_{4}) \xmapsto{\mu_{3}} (x_{1}',x_{2}',x_{3}',x_{4}) \xmapsto{\mu_{4}} (x_{1}',x_{2}',x_{3}',x_{4}')
\end{align*}
where each variable $x_{i}$ is replaced by a new cluster variable $x_{i}'$ in each mutations which is given by the 
\begin{equation}\label{typeD4:clmu}
\begin{split}
\mu_{1}: (x_{1},x_{2},x_{3},x_{4}) \mapsto (x_{1}',x_{2},x_{3},x_{4}): & \quad x_{1}'x_{1} = 1 + x_{2} \\ 
\mu_{2}: (x_{1}',x_{2},x_{3},x_{4}) \mapsto (x_{1}',x_{2}',x_{3},x_{4}): & \quad x_{2}'x_{2} = 1 + x_{1}' x_{3}x_{4} \\ 
\mu_{3}: (x_{1}',x_{2}',x_{3},x_{4}) \mapsto (x_{1}',x_{2}',x_{3}',x_{4}): & \quad x_{3}'x_{3} = 1 + x_{2}' \\ 
\mu_{4}: (x_{1}',x_{2}',x_{3}',x_{4}) \mapsto (x_{1}',x_{2}',x_{3}',x_{4}'): & \quad x_{4}'x_{4} = 1 + x_{2} '\\ 
\end{split}
\end{equation}
From these facts above,  one can construct corresponding cluster map $\varphi_{D_{4}} = \mu_{4}\mu_{3}\mu_{2}\mu_{1}$
\begin{equation}
\begin{array}{lcrcl}
    \varphi_{D_{4}}&:\bx=(x_1,x_2,x_3,x_4)&\mapsto&\bx'=(x_{1}',x_{2}',x_{3}',x_{4}') \; 
  \end{array}
\end{equation}
defined by the exchange relations shown in \eqref{typeD4:clmu}. Recall that sequence of mutations in cluster algebra of finite type has property of Zamolodchikov periodicity with period $(h+2)/2$ as stated in Theorem \ref{zam}. Since the Coxeter number of type $D_{4}$ is 6, the cluster map is periodic with period $4$, 
\begin{align*}
\varphi^{4}_{D_{4}}(\vb{x}) = \vb{x}
\end{align*}
Furthermore it preserves following presymplectic 2-form 
\begin{equation}\label{psympD4}
\om = \frac{1}{x_{1}x_{2}} \dd x_{1} \w \dd x_{2} + \frac{1}{x_{2}x_{3}} \dd x_{2} \w \dd x_{3} + \frac{1}{x_{2}x_{4}} \dd x_{2} \w \dd x_{4} 
\end{equation}
Note that the presymplectic form is degenerate as the rank of exchange matrix $B$ is 2 which is not full rank, i.e. $\rank(B)=2 < 4 $. Therefore we can perform a sympletic reduction (see Theorem \ref{symreduction}), which reduces the map to symplectic map defined on the plane. We begin with considering the null space, $\ker B$, and image, $\im B$, given by  
\begin{equation}
\ker B = < (1,0,0,1)^{T}, (1,0,1,0)^{T}> , \quad \im B = <(0,1,0,0)^{T}, (-1,0,1,1)^{T} > 
\end{equation}
Using the basis of $\im B$, we define rational map $\pi : \mathbb{C}^{4} \to \mathbb{C}^{2}$,
\begin{equation}
\begin{array}{lcrcl}
    \pi&:&\mathbb{C}^4&\to&\mathbb{C}^2\\
    & &\vb{x}=(x_1,x_2,x_3,x_4)&\mapsto&\vb{u}=(u_1,u_2) \\ 
  \end{array}
\end{equation} 
which projects $\vb{x}$ onto the plane of reduced coordinates, 
\begin{equation}
u_{1} = x_{2},\quad  u_{2} = \frac{x_{3}x_{4}}{x_{1}}
\end{equation}
Through this reduction $\pi$, the map $\varphi_{D_{4}}$ reduces to the 2D symplectic map $\hat{\varphi}_{D_{4}}$, given by 
\begin{equation}\label{phihatmapD4} 
 \begin{array}{lcrcl}
    \hat{\varphi}_{D_{4}}&:&\C^2&\to&\C^2\\
    & &\by=(u_1,u_2)&\mapsto&\left(\dfrac{1 + u_{2} + u_1u_2 }{u_1}, \dfrac{(1+u_{1}+u_{2}+u_{1}u_{2})^2}{u_2u_1^2(1 + u_1)}\right), 
  \end{array}
\end{equation}
which satisfies $\hat{\varphi}\cdot\pi = \pi\cdot\varphi $ and preserves the symplectic form, 
\begin{equation}\label{sympformD4}
\hatom = \dfrac{1}{u_1 u_2}\dd u_{1} \w \dd u_{2}, \quad \pi^{*}(\hatom) = \om.
\end{equation}

To show that $\hat{\varphi}_{D_{4}}$ is Liouville integrable, it is sufficient to find a single function (a first integral) that is invariant under the action of the map. Earlier in this section, we pointed out that the cluster map $\varphi_{D_{4}}$ is periodic with period $4$. Consequently, $\hat{\varphi}_{D_{4}}$ is also periodic, since it is induced from $\varphi_{D_{4}}$ via the projection $\pi$. Using this periodicity, one can construct a following symmetric function.
\begin{equation}\label{firstintD4}
\begin{split}
I(\vb{u}) &= \sum_{i=0}^{3} (\hat{\varphi}_{D_{4}}^{*})^{i} (u_{1}) \\ 
&=   u_1 + u_2 + \frac{ u_1}{u_2} + \frac{ u_2}{u_1} + \frac{1}{u_2} + \frac{1}{u_1} + \frac{ u_2}{u_1^2} + \frac{1}{u_2u_1} + \frac{1}{u_1^2} + \frac{1}{u_2u_1^2}
\end{split}
\end{equation}
which satisifies $\hat{\varphi}_{D_{4}}^{*}(I(\vb{u})) = I(\vb{u})$. 

\begin{rem}
The reason we prefer working on the reduced space is that there are examples which are not integrable in the original space they are defined on, but turns out to be integrable on the reduced space. For further detail please see \cite{FH2013,hone2019cluster}
\end{rem}

In the next section, we adjust the cluster mutations (deformation introduced in \cite{hk}) in \eqref{typeD4:clmu}  and find the condition under which the map remains to be integrable.


\subsection{Deformed map of type $D_{4}$}\label{DefmapD4}
We make a minor modification of the cluster mutations in the cluster map $\varphi_{D_{4}}$ in a way that an arbitrary coefficient is attached to each monomials of right hand side of expressions in \eqref{typeD4:clmu} as shown below. 
\begin{equation}\label{typeD4:dclmu}
\begin{split}
\tilde{\mu}_{1}: (x_{1},x_{2},x_{3},x_{4}) \mapsto (x_{1}',x_{2},x_{3},x_{4}): & \quad x_{1}'x_{1} = b_{1} + a_{1}x_{2} \\ 
\tilde{\mu}_{2}: (x_{1}',x_{2},x_{3},x_{4}) \mapsto (x_{1}',x_{2}',x_{3},x_{4}): & \quad x_{2}'x_{2} = b_{2} + a_{2}x_{1}' x_{3}x_{4} \\ 
\tilde{\mu}_{3}: (x_{1}',x_{2}',x_{3},x_{4}) \mapsto (x_{1}',x_{2}',x_{3}',x_{4}): & \quad x_{3}'x_{3} = b_{3} + a_{3}x_{2}' \\ 
\tilde{\mu}_{4}: (x_{1}',x_{2}',x_{3}',x_{4}) \mapsto (x_{1}',x_{2}',x_{3}',x_{4}'): & \quad x_{4}'x_{4} = b_{4} + a_{4}x_{2} '\\ 
\end{split}
\end{equation}
which we refer this as deformed map of type $D_{4}$, $\phi_{D_{4}} = \tilde{\mu}_{4}\tilde{\mu}_{3}\tilde{\mu}_{2}\tilde{\mu}_{1}$. We can reduce the number of parameters by scaling each cluster variables $ x_{i} \to \lambda_{i} x_{i}$ with an appropriate choice of parameters $(\lambda_{1},\lambda_{2},\lambda_{3},\lambda_{4}) \in (\mathbb{C}^{*})^4$.  which simplifies the process of the calculations in the later stage. The iteration of the deformed map can be written as  
\begin{equation}\label{typeD4:dclmu1}
\begin{split}
 x_{1,n+1}x_{1,n}& = b_{1} + x_{2,n} \\ 
x_{2,n+1}x_{2,n} &= b_{2} + x_{1,n+1} x_{3,n}x_{4,n} \\ 
 x_{3,n+1}x_{3,n} &= b_{3} + x_{2,n+1}\\ 
 x_{4,n+1}x_{4,n} &= b_{4} + x_{2,n+1} '\\ 
\end{split}
\end{equation}
where $n$ denotes the number of iterates $\tilde{\phi}^{n}(\vb{x})$. By the theorem 1.3 in \cite{hk}. Following the same steps as previously, we can project this map to 2D parametric birational map via the rational map $\pi:\mathbb{C}^{4} \to \mathbb{C}^{2}$ as
\begin{equation}
 \begin{array}{lcrcl}
    \hat{\phi}_{D_{4}}&:&\C^2&\to&\C^2\\
    & &\vb{u}=(u_1,u_2)&\mapsto&\left(\dfrac{(b_1 + u_1)u_2 + b_2}{u_1}, \dfrac{(b_4 + u_2)u_1 + b_1u_2 + b_2)((b_3 + u_2)y_1 + b_1u_2 + b_2)}{u_2u_1^2(b_1 + u_1)}\right) \;, 
  \end{array}
\end{equation}
which is intertwined with $\phi_{D_{4}}$ via $\pi$, i.e.  $\hat{\phi}_{D_{4}} \cdot \pi = \pi \cdot \phi_{D_{4}}$, and preserves the symplectic form \eqref{sympformD4}. To verify the Liouville integrability of the map, we first assume that corresponding first integral takes analogous form of the \eqref{firstintD4} such that  
\begin{equation}
\tilde{I}(\vb{u})=   \alpha_{0} u_1 + \alpha_1 u_2 + \frac{\alpha_2 u_1}{u_2} + \frac{\alpha_3 u_2}{u_1} + \frac{\alpha_4}{u_2} + \frac{\alpha_5}{u_1} + \frac{\alpha_6 u_2}{u_1^2} + \frac{\alpha_7}{u_2u_1} + \frac{\alpha_8}{u_1^2} + \frac{\alpha_9}{u_2u_1^2}
\end{equation} 
where we adjusted each monomial in the first integral \eqref{firstintD4} by attaching an arbitrary parameter $\alpha_{i}$. Without loss of generality, we can set $\alpha_{0} = 1$.  Imposing invariance of $\tilde{I}$ under the action of $\hat{\phi}_{D_{4}}$, i.e. (i.e.$\hat{\phi}_{D_{4}}^{*}(\tilde{I}(\vb{u})) = \tilde{I}(\vb{u})$), constrains parameters, yielding  
\begin{equation} 
\tilde{I}(\vb{u}) =  u_1 + u_2 + \frac{u_1}{u_2} + \frac{(b_1 + 1)u_2}{u_1} + \frac{b_3 + b_4 + 1}{u_2} + \frac{b_1 + b_2 + b_3 + b_4 + 1}{u_1} + \frac{b_1u_2}{u_1^2} + \frac{b_3b_4 + b_3 + b_4}{u_1u_2} + \frac{2b_2}{u_1^2} + \frac{b_3b_4}{u_1^2u_2},
\end{equation}
provided that one of the following sets of conditions is satisfied, 
\begin{align*}
(1) &\quad b_{2} = b_{4} = b_{1}b_{3} \\ 
(2) &\quad b_{1} = b_{2} = b_{3}b_{4} \\
(3) &\quad b_{2} = b_{3} = b_{1}b_{4} \\
\end{align*}
This result shows  that the cluster map $\varphi_{D_{4}}$ admits deformations into integrable maps in more than one distinct way. Notice that both the first integral and the deformed map are invariant under the exchange of parameters $b_{3} \leftrightarrow b_{4}$ and original variables $x_{3} \leftrightarrow x_{4}$. As a result, among the three 	cases considered,  integrable maps corresponding to cases (1) and (2) are genuinely distinct.  	
\begin{equation}\label{D41map}
\hat{\phi}_{1}: \quad  \mqty(u_{1} \\ u_{2}) \mapsto \qty(\frac{(b_1+u_{1})u_{2} + b_1b_3}{u_{1}}, \quad  
\frac{\qty[(b_1 + u_1)u_2 + b_1b_3(u_1+1)]\cdot \qty[u_2 + b_{3}]}{u_1^2u_2})
\end{equation}
\begin{equation}\label{D42map}
    \hat{\phi}_{2}: \quad  \mqty(u_{1} \\ u_{2}) \mapsto 
\qty(\frac{(b_3b_4+u_{1})u_{2} + b_3b_4}{u_{1}}, \quad  \frac{\qty[(b_{4} + u_{2})u_{1} + b_3b_4(u_{2} + 1)]\cdot \qty[(b_{3} + u_{2})u_{1} + 
b_3b_4 (u_{2}+1)]}{u_{1}^2u_{2}(b_3b_4+u_{1})})
\end{equation}
Thus, the deformed map preserves both integrability and symplectic structure after deformation, but in the process, it loses the Laurent property, which is an essential feature of cluster algebras. In our previous work \cite{hkm24}, we showed that each map admits \textit{Laurentification} that lifts each deformed map to cluster algebra of rank 8 (including 2 frozen variables), $\mathcal{A}(\tilde{\vb{x}}, Q_{D_{4}}^{(l)})$ endowed with cluster $\tilde{\vb{x}} = (\tx_{1}, \tx_{2},\dots, \tx_{10} )$ and the extended quiver $Q_{D_{4}}^{(l)}$, for $l= \qty{1,2}$  (shown in Figure \ref{DeformedQD4}).
\begin{figure}[h!]
\begin{center}
\resizebox{1\textwidth}{!}{%
 \begin{tikzpicture}[every circle node/.style={draw,scale=0.6,thick},node distance=15mm]
  \node [draw,circle,fill=red!50,"$4$"] (a7) at (0,0) {};
  \node [draw,circle,fill=red!50,"$5$"] (a2) [right=of a7] {};
  \node [draw,circle,fill=red!50,"$6$"] (a5) [right=of a2] {};
  \node [draw,circle,fill=red!50,"$2$" right] (a4) [below right=of a5] {};
  \node [draw,circle,fill=red!50,"$1$" below] (a3) [below left=of a4] {};
  \node [draw,circle,fill=red!50,"$8$" below] (a6) [left=of a3] {};
  \node [draw,circle,fill=red!50,"$3$" below] (a1) [left=of a6] {};
  \node [draw,circle,fill=red!50,"$7$" left]  (a8) [above left=of a1] {};
   \node [draw,circle,fill=blue!50,"$10$" below]  (10) at (2.7,-3.2) {};
   \node [draw,circle,fill=blue!50,"$9$" below]  (9) at (1,-3.2) {};
  
   \node (a) [below=of a6] {(a) $Q_{D_{4}}^{(1)}$};
  
  \begin{scope}[>=Latex]
  
  \draw[-> , thick]  (a2) edge (a1);
   \draw[-> , thick]  (a3) edge (a2);
  \draw[-> , thick]  (a1) edge (a6);
  \draw[-> , thick]  (a5) edge (a1);
  \draw[-> , thick]  (a2) edge (a8);
  \draw[-> , thick] (a7) edge (a2);
    \draw[-> , thick] (a2) edge (a5);
      \draw[-> , thick] (a4) edge (a2);
       \draw[-> , thick] (a4) edge (a2);
       \draw[-> , thick] (a6) edge (a3);
 \draw[-> , thick]  (a3) edge (a5);
 \draw[-> , thick]  (a6) edge (a4);
  \draw[-> , thick]   (a4) edge (a5);
   \draw[-> , thick]   (a7) edge[bend left=35] (a5);
    \draw[-> , thick]  (a5) edge(a6);
    \draw[-> , thick] (a5) edge(a8);
    \draw[-> , thick]  (a8) edge(a6);
    \draw[-> , thick] (a6) edge(a7);
     \draw[-> , thick] (a3) edge(a7);

     \draw[-> , thick] (a7) edge (9);
      \draw[-> , thick] (9) edge(a1);
      \draw[-> , thick] (a4) edge[bend left=40](10);
      \draw[-> , thick] (a7) edge(10);
      \draw[-> , thick] (10) edge(a6);

    \end{scope}
    
 \node [draw,circle,fill=red!50,"$4$"] (a7) at (8,0) {};
  \node [draw,circle,fill=red!50,"$5$"] (a2) [right=of a7] {};
  \node [draw,circle,fill=red!50,"$6$"] (a5) [right=of a2] {};
  \node [draw,circle,fill=red!50,"$2$" right] (a4) [below right=of a5] {};
  \node [draw,circle,fill=red!50,"$1$" below] (a3) [below left=of a4] {};
  \node [draw,circle,fill=red!50,"$8$" below] (a6) [left=of a3] {};
  \node [draw,circle,fill=red!50,"$3$" below] (a1) [left=of a6] {};
  \node [draw,circle,fill=red!50,"$7$" left]  (a8) [above left=of a1] {};
   \node [draw,circle,fill=blue!50,"$10$" below]  (10) at (10.7,-3.2) {};
   \node [draw,circle,fill=blue!50,"$9$" below]  (9) at (9,-3.2) {};
  
   \node (b) [below=of a6] {(b) $Q_{D_{4}}^{(2)}$};
  
  \begin{scope}[>=Latex]
  
  \draw[-> , thick]  (a2) edge (a1);
   \draw[-> , thick]  (a3) edge (a2);
  \draw[-> , thick]  (a1) edge (a6);
  \draw[-> , thick]  (a5) edge (a1);
  \draw[-> , thick]  (a2) edge (a8);
  \draw[-> , thick] (a7) edge (a2);
    \draw[-> , thick] (a2) edge (a5);
      \draw[-> , thick] (a4) edge (a2);
       \draw[-> , thick] (a4) edge (a2);
       \draw[-> , thick] (a6) edge (a3);
 \draw[-> , thick]  (a3) edge (a5);
 \draw[-> , thick]  (a6) edge (a4);
  \draw[-> , thick]   (a4) edge (a5);
   \draw[-> , thick]   (a7) edge[bend left=35] (a5);
    \draw[-> , thick]  (a5) edge(a6);
    \draw[-> , thick] (a5) edge(a8);
    \draw[-> , thick]  (a8) edge(a6);
    \draw[-> , thick] (a6) edge(a7);
     \draw[-> , thick] (a3) edge(a7);
     
      \draw[-> , thick] (9) edge(a1);
      \draw[-> , thick] (a4) edge(9);
      \draw[-> , thick] (10) edge(a1);
      \draw[-> , thick] (a3) edge(10);
      
    \end{scope}

\end{tikzpicture}
}
\end{center}
\caption{(a) corresponds to $\hat{\varphi}_{1}$, while (b) corresponds to $\hat{\varphi}_{2}$  }\label{DeformedQD4}
\end{figure}
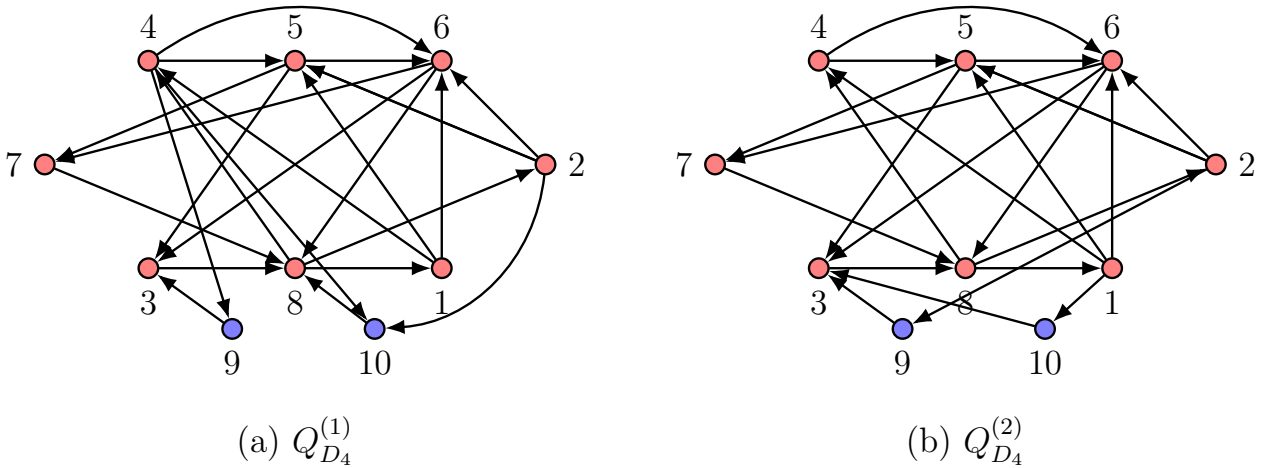
The process begins with defining the rational map $\pi_{1} : (\tx_{1}, \tx_{2},\dots, \tx_{8} ) \to (u_{1},u_{2})$, which is equivalent to
 \begin{equation}\label{vartransD4i}
    u_{1} =\frac{\tx_3}{\tx_{5}\tx_6}, \qquad 
u_2 = \frac{\tx_2\tx_4\tx_8}{\tx_1\tx_6\tx_{7}},
\end{equation}
and whose structure is derived from analysing singularity confinement pattern using an empirical version of $p$-adic analysis (see further details in \cite{hkm24}). Using this rational map $\pi_{1}$, one can pullback the symplectic form $\hat{\omega}$ to define new symplectic form  
\begin{equation}
 \tilde{\om} = \pi_{1}^{*}(\hat{\om}) = \sum_{i<j}\frac{\tilde{b}_{ij}}{\tilde{x}_{i}\tilde{x}_{j}} \dd \tilde{x}_{i} \w \dd \tilde{x}_{j}
\end{equation}
where the coefficients $\tilde{b}_{ij}$ denote the number of edges between mutable (non frozen) nodes in $ Q_{D_{4}}^{(l)}$, directed from node $i$ to node $j$. For each quiver, there exists a corresponding sequence of mutations that produces a cluster map. Each such cluster map arises from the deformed maps via Laurentification in the cluster algebra  $\mathcal{A}(\tilde{\vb{x}}, Q_{D_{4}}^{(l)})$. The map corresponding to $Q_{D_{4}}^{(1)}$ is given by 
\begin{equation}
\psi^{(1)}_{D_{4}} = \pi_{1}^{*}\hat{\phi}_{D_{4}}^{(1)}  = \hat{\rho}_{1}^{-1}\hat{\mu}_{1} \hat{\mu}_{3}\hat{\mu}_{8}, \qquad \rho_{1} = (3,2,8)(1,5,6,7,4)
\end{equation}
whose mutations are defined in cluster algebra  $\mathcal{A}(\tilde{\vb{x}}, Q_{D_{4}}^{(1)})$. Note that here we denote the mutations in this cluster algebra  $\hat{\mu}$ to distinguish them from the mutations in the cluster algebra of type $D_{4}$.
By setting the initial cluster $\tilde{\vb{x}}$ as 
$$(\tx_{1}, \tx_{2},\tx_{3}, \tx_{4},\tx_{5},\tx_{6},\tx_{7},\tx_{8},\tx_{9},\tx_{10} ) = (\tau_{0},\sigma_{1},\sigma_{0},\tau_{4}, \tau_{1},\tau_{2},\tau_{3},r_{0}, b_3, b_4)$$
the action of the map $\psi^{(1)}_{D_{4}}$ induces 
\begin{equation}\label{LaurentphihatmapD4i} 
 \begin{array}{lcrcl}
   \psi^{(1)}_{D_{4}}  &:&(\tx_{1}, \tx_{2},\tx_{3}, \tx_{4},\tx_{5},\tx_{6},\tx_{7},\tx_{8},\tx_{9},\tx_{10} )&\to& (\tx_{5}, \tx_{8}',\tx_{2},\tx_{1}',\tx_{6},\tx_{7},\tx_{4}, \tx_{3}' ,\tx_{9},\tx_{10})\\
    & &(\tau_{0},\sigma_{1},\sigma_{0},\tau_{4}, \tau_{1},\tau_{2},\tau_{3},r_{0}, b_3, b_4)&\mapsto& (\tau_{1},\sigma_{2},\sigma_{1},\tau_{5}, \tau_{2},\tau_{3},\tau_{4},r_{1},b_{3},b_{4}), 
  \end{array}
\end{equation}
Discrete evolution induced by map is described by the following system of recurrence relations, 
\begin{equation}\label{systmtauD4i}
    \begin{split}
        & \sigma_{n+2}r_{n} =b_{3}\sigma_{n}\tau_{n+2}\tau_{n+3} + \sigma_{n+1}\tau_{n}\tau_{n+4} \\ 
        &  r_{n+1}\sigma_{n} = \sigma_{n + 1}\tau_{n+4}\tau_{n} + b_{1}\sigma_{n+2}\tau_{n+2}\tau_{n} \\ 
        &\tau_{n+5}\tau_{n} = b_1b_3 \tau_{n+2}\tau_{n+3} + r_{n+1} . 
    \end{split}
\end{equation}
Laurentification of the second deformed map $\hat{\phi}_{2}$ gives rises to following cluster map, 
\begin{equation}
\psi^{(2)}_{D_{4}} = \pi_{2}^{*}\hat{\phi}_{D_{4}}^{(2)}  = \hat{\rho}_{2}^{-1}\hat{\mu}_{2}\hat{\mu}_{1}\hat{\mu}_{8}\hat{\mu}_{3} ,\qquad \hat{\rho}_{2} = (1,2)(3,4,5,6,7)
\end{equation}
Now if we identify the cluster variables in the initial cluster $\tilde{\vb{x}}$ as 
\begin{equation} 
(\tx_{1}, \tx_{2},\tx_{3}, \tx_{4},\tx_{5},\tx_{6},\tx_{7},\tx_{8},\tx_{9},\tx_{10} ) = (\hat{s}_{0},\hat{r}_{0},\hat{\tau}_{0}, \htau_{1},\htau_{2},\htau_{3},\htau_{4},\hat{\eta}_{0}, b_3, b_4)
\end{equation}
then we see that the iteration of the map, 
\begin{equation}\label{LaurentphihatmapD4ii} 
 \begin{array}{lcrcl}
   \psi^{(1)}_{D_{4}}  &:&(\tx_{1}, \tx_{2},\tx_{3}, \tx_{4},\tx_{5},\tx_{6},\tx_{7},\tx_{8},\tx_{9},\tx_{10} )&\to& (\tx_{2}',\tx_{1}', \tx_{3}',\tx_{4},\tx_{5},\tx_{6},\tx_{7},\tx_{3}, \tx_{8}' ,\tx_{9},\tx_{10})\\
    & &(\hat{s}_{0},\hat{r}_{0},\hat{\tau}_{0}, \htau_{1},\htau_{2},\htau_{3},\htau_{4},\hat{\eta}_{0}, b_3, b_4)&\mapsto& (\hat{s}_{1},\hat{r}_{1},\hat{\tau}_{1}, \htau_{2},\htau_{3},\htau_{4},\htau_{5},\hat{\eta}_{1}, b_{3},b_{4}), 
  \end{array}
\end{equation}
is equivalent to the following recursion relations,
\begin{equation}
\begin{split}
 &\htau_{n+5}\htau_{n} = b_3b_4  \htau_{n+3}\htau_{n+2} + \hat{\eta}_{n}, \\ 
    &\hat{\eta}_{n+1}\hat{\eta}_{n} = \hat{r}_{n}\hat{\sigma}_{n}\htau_{n+1}\htau_{n+5} + b_3b_4 \htau_{n+4}\htau_{n+3}^2\htau_{n+2}, \\ 
    &\hat{r}_{n+1}\hat{s}_{n} = b_{3}\htau_{n+3}\htau_{n+4} +\hat{\eta}_{n+1}, \\
    & \hat{s}_{n+1}\hat{r}_{n} = b_{4} \htau_{n+3}\htau_{n+4} + \hat{\eta}_{n+1}. 
\end{split}
\end{equation}

In summary, the deformation approach introduced by \cite{hk} generalizes the Liouville integrable type $D_{4}$ cluster map in two distinct ways. Moreover, this approach yields new cluster algebras via Laurentification, with quivers that induce cluster maps, $\psi^{(1)}$ and $\psi^{(2)}$, corresponding to each deformed map, $\hat{\varphi}_{1}$ and $\hat{\varphi}_{2}$ respectively.


\section{Deformed cluster map of type $D_{6}$}\label{deformedtypeD6}

The Cartan matrix for the type $\rD_6$ is 
\begin{equation}
    \mqty(2 & -1 & 0 & 0 & 0 & 0 \\-1 & 2 & -1 & 0 & 0 & 0  \\ 0 & -1 & 2 & -1 & 0 & 0  \\ 0 & 0 & -1 & 2 & -1 & -1 \\  0 & 0 & 0 & -1 & 2 & 0 \\ 0 & 0 & 0 & -1 & 0 & 2 \\ )
\end{equation}
which corresponds to the following exchange matrix, 
\begin{equation}\label{exchD6}
  B=  \mqty(0 & 1 & 0 & 0 & 0 & 0 \\-1 & 0 & 1 & 0 & 0 & 0  \\ 0 & -1 & 0 & 1 & 0 & 0  \\ 0 & 0 & -1 & 0 & 1 & 1 \\  0 & 0 & 0 & -1 & 0 & 0 \\ 0 & 0 & 0 & -1 & 0 & 0 \\ )
\end{equation}
The associated cluster map $\varphi_{D_{6}}$ is given by the composition of mutations $\mu_{6}\mu_{5}\mu_{4}\mu_{3}\mu_{2}\mu_{1}$, 
which preserves the type $D_{6}$ quiver. Let us consider the deformation of each mutations in the cluster map $\varphi_{D_{6}}$ which is shown as below. 
\begin{equation} \label{D6dmaps}\begin{array}{rcl}
\mu_{1} : (x_1,x_2,x_3,x_4, x_{5},x_{6}) \mapsto (x_1',x_2,x_3,x_4, x_{5},x_{6}), \qquad x_1 x_1' &=& b_1 + a_1x_2 \\
\mu_{2} :(x_1',x_2,x_3,x_4, x_{5},x_{6}) \mapsto (x_{1}',x_{2}',x_3,x_4, x_{5},x_{6}), \qquad x_2 x_2' &=& b_2 + a_2 x_3x_{1}'\\
\mu_{3} : (x_{1}',x_{2}',x_3,x_4, x_{5},x_{6}) \mapsto (x_{1}',x_{2}',x_{3}',x_4, x_{5},x_{6}), \qquad  x_3 x_{3}'& = & b_3 + a_3 x_{4}x_2'\\
\mu_{4} : (x_1',x_2',x_3',x_4, x_{5},x_{6}) \mapsto (x_{1}',x_{2}',x_{3}',x_{4}', x_{5},x_{6}), \qquad x_4 x_4' &= & b_{4} + a_4 x_{5}x_{6}x_{3}'\\ 
\mu_{5} : (x_1',x_2',x_3',x_4', x_{5},x_{6}) \mapsto (x_{1}',x_{2}',x_{3}',x_{4}', x_{5}',x_{6}), \qquad x_5 x_5' &= & b_{5} + a_5 x_{4}'\\ 
\mu_{6} : (x_1',x_2',x_3',x_4, x_{5}',x_{6}) \mapsto (x_{1}',x_{2}',x_{3}',x_{4}', x_{5}',x_{6}'), \qquad x_6 x_6' &= & b_{6} + a_6 x_{4}'\\ 
\end{array}
\end{equation} 
The deformed map $\tilde{\varphi}_{\rD_{6}} =   \tilde{\mu_{6}}\tilde{\mu_{5}}\tilde{\mu_{4}}\tilde{\mu_{3}}\tilde{\mu_{2}}\tilde{\mu_{1}}$ transforms back to original cluster map $\varphi_{\rD_{6}}= \mu_{6}\mu_{5}\mu_{4}\mu_{3}\mu_{2}\mu_{1}$ when we fix the parameters $a_i=1=b_i$ for all $i=1,\dots,4$. Furthermore, since Coxeter number $h$ for the type $\rD_6$ is 10,  the periodicity for the cluster map is period 6, 
\begin{equation}
    \varphi_{\rD_{6}}\cdot(\vb{x},B) = (\varphi_{\rD_{6}}(\vb{x}),B) \ \ \text{and}\ \  \varphi^6(\vb{x}) = \vb{x}
\end{equation}
The symplectic form $\om$, that is invariant under the deformed map, takes a following form, 
\begin{equation}
\begin{split}
\om = & \frac{1}{x_{1}x_{2}} \dd  x_{1} \wedge \dd x_{2} + \frac{1}{x_{2}x_{3}} \dd  x_{2} \wedge \dd x_{3} \\
&+ \frac{1}{x_{4}x_{5}} \dd  x_{4} \wedge \dd  x_{5} + \frac{1}{x_{4}x_{6}} \dd  x_{4} \wedge \dd x_{6}
\end{split}
\end{equation}
As in the previous example, we rescale each cluster variable (i.e. $x_{i} \to \lambda_{i} x_{i}$) so that the \eqref{D6dmaps} can be rewritten as 
\begin{equation}\label{D6recs} 
\begin{array}{rcl}
x_{1,n+1} x_{1,n} &=& x_{2,n} + b_1 , \\   
 x_{2,n+1} x_ {2,n} &=&  x_{3,n}x_{1,n+1} + b_{2}, \\  
 x_{3,n+1} x_{3,n}& = &  x_{4,n}x_{2,n+1} + b_3\\
x_{4,n+1} x_{4,n}&= &   x_{5,n}x_{6,n}x_{3,n+1}+b_{3} \\  
 x_{5,n+1} x_{5,n} &= &   x_{4,n+1} + b_{5} \\ 
 x_{6,n+1} x_{6,n} &= &  x_{4,n+1} + b_{6} \\ 
\end{array} 
\end{equation} 

From the exchange matrix $B_{\rD_{6}}$,  \eqref{exchD6}, one can see that it is degenerate and possess rank 4. Thus we consider the null space and column space of $B_{\rD_{6}}$ given by following
\begin{equation}
\begin{split}
    &\ker(B) = \qty{(1,0,1,0,0,1)^{T}, (1,0,1,0,1,0)^{T}}, \\
    &\im(B)=\qty{(0,1,0,0,0,0)^{T}, (-1,0,1,0,0,0)^{T}, (0,-1,0,1,0,0)^{T},  (0,0,-1,0,1,1)^{T}} 
\end{split}
\end{equation}
The vectors in $\im(B)$ leads to constructing new variables $y_{1}$, $y_{2}$, $y_{3}$, $y_{4}$ i.e. 
\begin{equation}
    \begin{split}
        y_{1} = x_{2}, \quad y_{2} = \frac{x_{3}}{x_{1}}, \quad y_{3} = \frac{x_{4}}{x_{2}}, \quad y_{4} = \frac{x_{5}x_{6}}{x_{3}}
    \end{split}
\end{equation}
which reduces the deformed map into the 4D symplectic map i.e. 
\begin{equation}
\hat{\varphi}: (y_{1},y_{2},y_{3},y_{4}) \to (y_{1}',y_{2}',y_{3}',y_{4}') 
\end{equation}
where the produced variables are written as 
\begin{equation}
\begin{split}
&y_{1}' = \frac{(b_1 + y_{1}) y_{2} + b_{2}}{y_{1}}\\
&y_{2}' = \frac{y_{3}(b_{1} + y_{1})y_{2} + b_{2}y_{3} + b_{3}}{y_{2}(b_{1}+y_{1})}\\
&y_{3}' = \frac{y_{3}(b_{1}+y_{1})y_{2}y_{4} + b_{2}y_{3}y_{4} + b_{3}y_{4} + b_{4}}{y_{3}(b_{2} + (b_{1}+y_{1}) y_{2})} \\
&y_{4}' = \frac{\qty((b_{1}y_{2} +y_{1}y_{2} + b_{2})y_{3}y_{4} +b_{6}y_{1}y_{3} + b_{3} y_{4}+ b_{4} ) \qty((b_{1}y_{2} + y_{1}y_{2} + b_{2})y_{3}y_{4} + b_{5}y_{1}y_{3} + b_{3}y_{4} + b_{4})}{y_{1}^2 y_{3}^2y_{4}((b_{1}y_{2}+y_{1}y_{2} + b_{2})y_{3} + b_{3} )}
\end{split}
\end{equation}
This reduced map preserves the nondegenerate symplectic form 
\begin{align*}
\hat{\om} = \sum_{ij}\hat{b}_{ij} \dl y_{i} \w \dl y_{j}
\end{align*}
whose coefficients are entries of the following $4\times 4$ skew-symmetric matrix 
$$\mqty(0 & 1 & 0 & 1 \\ -1 & 0 & 0 & 0 \\ 0 & 0 & 0 & 1 \\ -1 & 0 & -1 & 0 )$$

\subsection{Laurentification of deformed type $D_{6}$ map}


From the several successful cases shown in \cite{hk,grab,hkm24} (including the previous section), we observed that the degrees (or magnitude) of numerators and denominators in the iterates grows exponentially when arbitrary parameter values are used. However when the integrability condition is imposed, this growth is reduced significantly. This hints that choosing suitable conditions on the parameters leads to cancellation between the numerator and denominator, simplifying the rationals expressions and hence reducing complexity of the map. This reduced growth corresponds to the notion of degree growth used in \textit{algebraic entropy}, which is an algebraic test for integrability. Under these conditions, one was able to find the singularity confinement patterns which allow us to apply Laurentification and lift the map to a space where the Laurent property holds. By combining altogether, this indicates a close relation between slow growth (slower than exponential growth) and Laurentification.

We begin by considering the iteration of deformed type $D_6$ map, with the parameters chosen so that $b_{1} = b_{2} = b_{3} = b_{4} = b_{5}b_{6}$. We study the heights $H(\vb{y}_{n})$ of the iterates  $\vb{y}_{n} = (y_{j,n})_{1\leq j \leq4}$ of the deformed map $\hat{\varphi}_{D_{6}}$,  which measure the complexity growth of the map numerically. The height is defined by 
$$H(\vb{y}_{n}) = \max_{1 \leq j \leq 4}\max (\abs{u_{j,n}},\abs{v_{j,n}})$$
where $u_{j,n}$ and $v_{j,n}$ are numerator and denominator of the iterates $y_{j,n}$ and they are coprime. A standard way to test for subexponential is to compare the logarithmic height $ \log H(\vb{y}_{n})$ with $ \log n$: if $\log H(\vb{y}_{n})$ is asymptotically proportional to $\log n$, then $H(\vb{y}_{n})$ grows polynomially. Setting initial values $y_{1} = y_{2} = y_{3} = y_{4} =1$ and parameters $b_{5} = 2$, $b_{6} = 3$, we plot $ \log H(\vb{y}_{n})$ against $ \log n$ which is illustrated in the figure \ref{fig:D6graph}. From the graph, it is clear that the plot is asymptotically a straight line; the slope $a$ of this line indicates that 
\begin{align*}
\log H(\vb{y}_{n}) \sim a \log n \implies H(\vb{y}_{n})\sim n^{a}
\end{align*}
Thus the height exhibits polynomial growths $n^{a}$ for some  $a>0$. This growth is therefore slower than the exponential. 

\begin{figure}[!h]
\begin{center}
\includegraphics[width=0.5\textwidth]{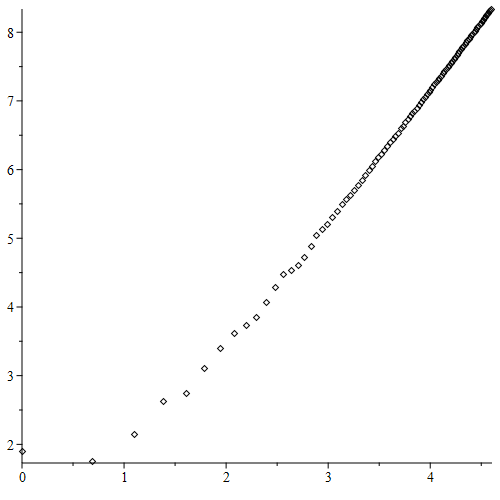}
\end{center}
\caption{$x$-axis represents $\log n$ and $y$-axis corresponds to $\log H(\vb{y}_{n})$}. 
\label{fig:D6graph}
\end{figure}
%





Next, using the same initial values and parameter values as mentioned earlier, we consider sequence of rational values, which are factorized into prime numbers as shown in the table.
{\renewcommand{\arraystretch}{1}
\begin{center}
\begin{tabular}{ |c | c | c | c | c |  } 
  \hline
  $n$ & $y_{1,n} $ & $ y_{2,n} $ & $y_{3,n} $ & $ y_{4,n} $ \\ 
  \hline
  $1 $ & $13 $ & $ \frac{19}{7} $ & $\frac{5^2}{13} $ & $\frac{2^2\cdot 3^3 \cdot 7 }{19}$ \\
  \hline
  $2$ & $\frac{31}{7}$ & $\frac{43}{19}$ & $2 \cdot 3 \cdot 7$ & $\frac{7\cdot 47}{43}$  \\ 
  \hline
  $3$ & $\frac{127}{19}$ & $\frac{2^3\cdot 3 \cdot 7 \cdot 23}{73}$ & $\frac{977}{127}$ & $\frac{5\cdot 11 \cdot 29}{2^2 \cdot 3 \cdot 19 \cdot 23}$  \\ 
  \hline
  $4$ & $\frac{2 \cdot 3^3 \cdot 137}{73}$ & $\frac{13109}{7 \cdot 241}$ & $ \frac{2 \cdot 23 \cdot 467}{3^3 \cdot 19 \cdot 137} $ & $\frac{3^3 \cdot 5^2 \cdot 7 \cdot 11 \cdot 13}{73 \cdot 13109}$  \\ 
  \hline
   $5$ & $\frac{5 \cdot 2797}{7 \cdot 241}$ & $\frac{47 \cdot 499}{2 \cdot 3 \cdot 19 \cdot 653}$ &  $ \frac{3 \cdot 7 \cdot 36263}{5 \cdot 73 \cdot 2797} $  & $\frac{2^4 \cdot 17 \cdot 19 \cdot 443}{47 \cdot 241 \cdot 499}$  \\ 
  \hline
    $6$ & $\frac{94307}{2 \cdot 3 \cdot 19 \cdot 653}$ & $\frac{3 \cdot 7 \cdot 257501}{73 \cdot 24107}$ &  $ \frac{2^2\cdot 19 \cdot 23 \cdot 19259}{241 \cdot 94307} $  & $\frac{73 \cdot 13457 \cdot 5197}{3^2 \cdot 653 \cdot 257501}$  \\ 
  \hline
    $7$ & $\frac{3 \cdot 13117691}{73 \cdot 24107}$ & $\frac{2^4 \cdot 19 \cdot 211 \cdot 30559}{7 \cdot 241 \cdot 269 \cdot 2011}$ &  $ \frac{5^2 \cdot 73 \cdot 2711 \cdot 62311}{3^2 \cdot 653 \cdot 13117691} $  & $\frac{2^5 \cdot 7 \cdot 23 \cdot 41 \cdot 181\cdot 241 \cdot  881}{211 \cdot 30559 \cdot 24107}$  \\ 
  \hline
    $8$ & $\frac{2 \cdot 199 \cdot 980249}{241 \cdot 269 \cdot 2011}$ & $ \frac{2^3 \cdot 13 \cdot 73 \cdot 149 \cdot 1053713}{3^3 \cdot 7 \cdot 11 \cdot 19 \cdot 653 \cdot 10289}$ &  $ \frac{83 \cdot 229 \cdot 241 \cdot 6317 \cdot 8681 }{199 \cdot 980249 \cdot 24107} $  & $\frac{3^2 \cdot 11 \cdot 179 \cdot 653 \cdot 98597 \cdot 14281}{2^3 \cdot 13 \cdot 149 \cdot 269 \cdot 1053713 \cdot 2011}$  \\ 
  \hline
  $9$ & $\frac{19 \cdot 252533 \cdot 45127}{3^3 \cdot 7 \cdot 11 \cdot 653 \cdot 10289}$ & $ \frac{59 \cdot 97 \cdot 241 \cdot 605239213}{73 \cdot 7712933 \cdot 24107}$ &  $ \frac{3^2 \cdot 11 \cdot 13 \cdot 109 \cdot 653 \cdot 74857 \cdot 31627 }{19 \cdot 269 \cdot 252533 \cdot 45127 \cdot 2011} $  & $\frac{5 \cdot 1543 \cdot 24107 \cdot 6089 \cdot 280759}{2 \cdot 3 \cdot 7 \cdot 59 \cdot 97 \cdot 605239213 \cdot 10289}$  \\ 
  \hline
\end{tabular}
\end{center}
}
To construct variable transformations which lifts the deformed type $D_{6}$ map to cluster algebra of higher rank, we study the singularity structure of the map via p-adic analysis, which is done by observing the patterns of prime factors appearing in the sequences.
One can see that for each prime numbers  $p_{1}= 73,241, 653$ (for instance) appears in all variables $y_{i}$ for $i=\qty{1,2,3,4}$. The corresponding $p_{1}$-adic norm exhibits the following pattern
\begin{equation}
\begin{split}
&\abs{y_{1,n}}_{p_{1}} = 1, 1 ,p_{1}, 1, 1, p_{1}, 1  \\
&\abs{y_{2,n}}_{p_{1}} = 1, p_{1}, 1 ,1 , p_{1}, 1, p_{1}^{-1}, p_{1}, 1  \\
&\abs{y_{3,n}}_{p_{1}} = 1, 1, 1,  p_{1}, 1 , p_{1}^{-1}, 1, 1 \\
&\abs{y_{4,n}}_{p_{1}} = 1, 1, p_{1}, 1, p_{1}^{-1}, 1, 1 \\
\end{split}
\end{equation}
There are particular values which emerge in two variables. For $p_{2} = 31,137$, one has $\abs{y_{1,n}}_{p_{2}} = p_{2}^{-1}$ and $\abs{y_{3,n}}_{p_{2}}= p_{2} $. For  $p_{3}= 19,23,43, 47,499,13109$, one has $\abs{y_{2,n}}_{p_{3}} = p_{3}^{-1}$ and $\abs{y_{4,n}}_{p_{3}}= p_{3} $. Each variable is associated with a set of primes that do does not appear in the other variables, for instance, the primes $p = 31,137$ arise exclusively in the $y_1$. Thus one can observe the singularity confinement patterns in the iterations in $(y_{1,n},y_{2,n},y_{3,n},y_{4,n})$
\begin{equation}\label{singD6}
\begin{split}
    \text{Pattern 1 :}  \dots &\to (R,\infty^{1}, R,R) \to (	\infty^{1},R, R , \infty^{1}) \to (R, R, \infty^{1}, R) \to (R, \infty^{1}, R , 0^{1} ) \\
& \to (\infty^{1}, R, 0^{1} , R)  \to (R,0^{1},R,R) \to(R,\infty^{1},R,R) \to \dots \\[0.5em]  
    \text{Pattern 2 :} \dots & \to (0^{1},R,\infty,R) \to \dots \\[0.5em]  
    \text{Pattern 3 :}\dots & \to (R,0^{1},R, \infty) \to \dots \\[0.5em]
    \text{Pattern 4 :}\dots & \to (R,R,0^{1}, R) \to \dots \\[0.5em]
     \text{Pattern 5 :}\dots & \to (R,R,R,0^{1}) \to \dots \\[0.5em]
\end{split}
\end{equation}
Given tau-functions specified by $\tau \equiv 0 \Mod{p_{1}}, \ \chi \equiv 0 \Mod{p_{2}}, \eta \equiv 0 \Mod{p_{3}}, \ \xi \equiv 0 \Mod{p_{4}}, r\equiv 0 \Mod{p_{5}}, \ s \equiv 0 \Mod{p_{6}} $, the  symplectic coordinates can b	e written as 
\begin{equation}\label{vartransformD6y}
    \begin{split}
        y_{1,n} = \frac{\chi_{n}}{\tau_{n+5}\tau_{n+2}}, \quad y_{2,n} = \frac{\tau_{n+1} \eta_{1,n}}{\tau_{n}\tau_{n+3}\tau_{n+6}}, \quad y_{3,n} = \frac{\tau_{n+2}\xi_{1,n}}{\tau_{n+4}\chi_{n}}, \quad y_{4,n} = \frac{\tau_{n+3}r_{n}s_{n}}{\tau_{n+5}\eta_{1,n}}
    \end{split}
\end{equation}
so that the confinement patterns are reproduced under the iteration of the map. Substituting these variables into the exchange relations yields new expressions but they are not in the form of an exchange relation. Instead, we consider the original system formulated with the  $x$-variables. By the initial cluster variables $x_{i} = 1$ for $i=1,2,3,4,5,6$ and taking the same parameter values as above, the  iteration then produces the sequences shown in the table below, 
{\renewcommand{\arraystretch}{1}
\begin{center}
\begin{tabular}{ |c | c | c | c | c | } 
  \hline
  $n$ & $x_{1,n} $ & $ x_{2,n} $ & $x_{3,n} $ & $ x_{4,n} $  \\ 
  \hline
  $1 $ & $7 $ & $ 13 $ & $  19 $ & $ 5^2 $  \\
  \hline
  $2$ & $\frac{19}{7}$ & $\frac{31}{7}$ & $\frac{43}{7} $ & $2\cdot 3 \cdot 31$ \\ 
  \hline
  $3$ & $\frac{73}{19}$ & $\frac{127}{19}$ & $\frac{2^3 \cdot 3 \cdot 7 \cdot 23}{19}$ & $\frac{977}{19 }$  \\ 
  \hline
  $4$ & $\frac{241}{73}$ & $ \frac{2 \cdot 3^3 \cdot 137}{73} $ & $\frac{13109}{7\cdot 73 }$ & $ \frac{2^2 \cdot 23 \cdot 467}{19 \cdot 73}$ \\ 
  \hline
   $5$ & $\frac{2^2 \cdot 3 \cdot 653}{241}$ & $\frac{5 \cdot 2797}{7 \cdot 241}$ &  $ \frac{2 \cdot 47 \cdot 499}{19 \cdot 241} $  & $\frac{3 \cdot 36263}{73 \cdot 241}$ \\ 
  \hline
    $6$ & $\frac{24107}{2^2 \cdot 3 \cdot 7 \cdot 653}$ & $\frac{94307}{2 \cdot 3 \cdot 19 \cdot 653}$ &  $ \frac{257501}{2^2 \cdot 73 \cdot 653} $  & $\frac{2 \cdot 23 \cdot 19259 }{3\cdot 241 \cdot 653}$ \\ 
  \hline
    $7$ & $\frac{2 \cdot 7 \cdot 269 \cdot 2011 }{19 \cdot 24107}$ & $\frac{3 \cdot 13117691}{73 \cdot 24107}$ &  $ \frac{2^5 \cdot 211 \cdot 30559}{241 \cdot 24107} $  & $\frac{5^2 \cdot 2711 \cdot 62311}{3 \cdot 653 \cdot 24107}$  \\ 
  \hline
    $8$ & $\frac{3^2 \cdot 7 \cdot 11 \cdot 19 \cdot 10289}{2 \cdot 73 \cdot 269 \cdot 2011}$ & $ \frac{2 \cdot 199 \cdot 980249}{241 \cdot 269 \cdot 2011}$ &  $ \frac{2^2 \cdot 13 \cdot 149 \cdot 1053713 }{3 \cdot 269 \cdot 653 \cdot 2011} $  & $\frac{2 \cdot 83 \cdot 229 \cdot 6317 \cdot 8681}{269 \cdot 24107 \cdot 2011}$  \\ 
  \hline
\end{tabular}
\end{center}
}
{\renewcommand{\arraystretch}{1}
\begin{center}
\begin{tabular}{| c | c| c |  } 
  \hline
  $n$ & $x_{5,n}$ & $x_{6,n}$ \\ 
  \hline
 $1$ & $3^3$ & $2^2 \cdot 7$ \\
  \hline
  $2$  & $\frac{2^2 \cdot 47}{3^3}$  & $ \frac{3^3}{2^2}$\\ 
  \hline
  $3$  & $ \frac{3^3 \cdot 5 \cdot 7 \cdot 29}{2^2 \cdot 19 \cdot 47}$ & $\frac{2^3\cdot 11 \cdot 47}{3^3 \cdot 19}$\\ 
  \hline
  $4$ & $ \frac{2^3 \cdot 11^2 \cdot 47}{5 \cdot 29 \cdot 73}$ & $\frac{3^3 \cdot 5^3 \cdot 13 \cdot 29}{2^3 \cdot 11 \cdot 47 \cdot 73}$ \\ 
  \hline
   $5$  & $\frac{5^3 \cdot 13 \cdot 29 \cdot 443}{2^3 \cdot 11^2 \cdot 47 \cdot 241}$ & $\frac{2^8 \cdot 11^2 \cdot 17 \cdot 47}{5^3 \cdot 13 \cdot 29 \cdot 241}$\\ 
  \hline
    $6$ & $\frac{2^6 \cdot 11^2 \cdot 17 \cdot  47 \cdot 13457}{3 \cdot 5^3 \cdot 13 \cdot 29 \cdot 443 \cdot 653}$ & $\frac{5^3 \cdot 13 \cdot 29 \cdot 443 \cdot 5197 }{2^8 \cdot 3 \cdot 11^2 \cdot 17 \cdot 47 \cdot 653} $\\ 
  \hline
    $7$ & $ \frac{5^3 \cdot 13 \cdot 23 \cdot 29 \cdot 41 \cdot 443 \cdot 881 \cdot 5197}{2^6 \cdot 11^2 \cdot 17 \cdot 47 \cdot 13457 \cdot 24107}$ & $\frac{2^16 \cdot 7 \cdot 11^2 \cdot 17 \cdot 47 \cdot 181 \cdot 13457}{5^3 \cdot 13 \cdot 29 \cdot 443 \cdot 24107 \cdot 5197}$\\ 
  \hline
    $8$ & $ \frac{2^{15} \cdot 3 \cdot 7 \cdot 11^3 \cdot 17 \cdot 47 \cdot 181 \cdot 98597\cdot 13457}{5^3 \cdot 13 \cdot 23 \cdot 29 \cdot 41 \cdot 269 \cdot 443 \cdot 881 \cdot 5197 \cdot 2011}$ & $ \frac{5^3 \cdot 13 \cdot 23 \cdot 29 \cdot 41 \cdot 179 \cdot 443 \cdot 881 \cdot 5197 \cdot 14281}{2^16 \cdot 7 \cdot 11^2 \cdot 17 \cdot 47 \cdot 181 \cdot 269 \cdot 13457 \cdot 2011}$ \\ 
  \hline
\end{tabular}
\end{center}
}
The prime numbers, appeared in the $y$ variables, once again can be spotted in the $x$-variables. e.g $p= 73,241, 653$  corresponds to $ \tau$, $p=31,137 $ associated with $P$ , $p= 19,23,43, 47,499$ corresponds to $Q$. $p= 467,977,36263$ corresponds to $R$.  There are specific primes in $x_{5}$ and $x_{6}$, which cancels out in $y_{4} = x_{5}x_{6} / x_{3}$, for instance $p=19,73$. Following the same approach as previously, we define the new tau function $\sigma_{n}$ which is associated such primes. Thus we can then construct explicit expression for the tau-functions, 
\begin{equation}\label{vartransformD6}
    \begin{split}
       & x_{1,n} = \frac{\tau_{n+6}\tau_{n}}{\tau_{n+5}\tau_{n+1}}, \quad  x_{2,n} = \frac{\chi_{n}}{\tau_{n+5}\tau_{n+2}}, \quad  x_{3,n} = \frac{\eta_{1,n}}{\tau_{n+5}\tau_{n+3}} \\ 
       &x_{4,n} = \frac{\xi_{1,n}}{\tau_{n+5}\tau_{n+4}}, \quad   x_{5,n} = \frac{r_{n}\sigma_{n}}{\tau_{n+5}}, \quad   x_{6,n} = \frac{s_{n}}{\tau_{n+5}\sigma_{n}}
    \end{split}
\end{equation}
where $\sigma_{n}\sigma_{n+1} = \frac{s_{n}}{r_{n}} $. Then by directly substituting thees variables into the deformed mutations, we obtain the following recurrence relations, 
\begin{equation}\label{eq:deformD6}
\begin{split}
&\tau_{n+7}\tau_{n} = b_{5}b_{6} \tau_{n+5}\tau_{n+2} + \chi_{n} \\ 
&\chi_{n+1}\chi_{n} = b_{5}b_{6}\tau_{n+2}\tau_{n+3}\tau_{n+5}\tau_{n+6} + \eta_{1,n}\tau_{n+1}\tau_{n+7} \\ 
&\eta_{1,n+1}\eta_{1,n} = b_{5}b_{6}\tau_{n+3}\tau_{n+4}\tau_{n+5} \tau_{n+6} + \xi_{1,n}\chi_{n+1} \\ 
&\xi_{1,n+1}\xi_{1,n} = b_{5}b_{6}\tau_{n+4}\tau_{n+5}^2 \tau_{n+6} + r_{n}s_{n}\eta_{1,n+1} \\ 
&r_{n+1}s_{n} = b_{5}\tau_{n+5}\tau_{n+6} + \xi_{1,n+1} \\ 
& s_{n+1}r_{n} = b_{6}\tau_{n+5}\tau_{n+6} + \xi_{1,n+1} \\ 
\end{split}
\end{equation}
that we denote as iterations of corresponding birational maps 
\begin{equation}
\psi_{\rD_{6}}: (\xi_{1,0},\tau_{0}, \tau_{1}, \tau_{2},\tau_{3},\tau_{4},\tau_{5},\tau_{6},\chi_{0}, \eta_{1,0},b_{5},b_{6}) \to   (\xi_{1,1}, \tau_{1}, \tau_{2},\tau_{3},\tau_{4},\tau_{5},\tau_{6},\tau{7},\chi_{1}, \eta_{1,1},b_{5},b_{6})
\end{equation}
on following initial data,
\begin{align*}
(\tilde{x}_{1},\tilde{x}_{2}, \tilde{x}_{4}, \tilde{x}_{5}, \tilde{x}_{6}, \tilde{x}_{7}, \tilde{x}_{8}, \tilde{x}_{9, }\tilde{x}_{10}, \tilde{x}_{11}, \tilde{x}_{12}) = (\xi_{1,0},\tau_{0}, \tau_{1}, \tau_{2},\tau_{3},\tau_{4},\tau_{5},\tau_{6},\chi_{0}, \eta_{1,0}).
\end{align*}
Let us denote the rational map $\pi : \mathbb{C}^{12} \to \mathbb{C}^{4}$ given by the variable transformations \eqref{vartransformD6}. Applying the pullback of the presymplectic form  $\om$ via the rational map $\pi$, yields a new presymplectic form whose coefficients give rise to new exchange matrix, which is essential in defining $\psi_{\rD_{6}}$. Furthermore, we inserts extra rows, $\qty(0,0,-1,1,0,0,0,0,0,0,0 )^{T}$ and $ \qty(0,-1,0,1,0,0,0,0,0,0,0 )^{T}$,  at the bottom of the matrix, thereby constructing the extended exchange matrix $\tilde{B}_{\rD_{6}}$ \eqref{exchmD6wfrozen}, which represents the quiver  illustrated in Figure \ref{quiverD6}.
\begin{figure}[!ht]
\centering
\resizebox{0.6\textwidth}{!}{%
 \begin{tikzpicture}[every circle node/.style={draw,scale=0.6,thick},node distance=15mm]

     \node [draw,circle,fill=red!50,"$5$"] (a7) at (9,0) {};
  \node [draw,circle,fill=red!50,"$6$"] (a2) [right=of a7] {};
  
   \node [draw,circle,fill=red!50,"$7$"] (aa1) [right=of a2] {};
     \node [draw,circle,fill=red!50,"$8$"] (aa2) [right=of aa1] {};

  \node [draw,circle,fill=red!50,"$9$"] (a5) [right=of aa2] {};
  \node [draw,circle,fill=red!50,"$3$" right] (a4) [below right=of a5] {};
  \node [draw,circle,fill=red!50,"$2$" below] (a3) [below left=of a4] {};
  
  \node [draw,circle,fill=red!50,"$1$" below] (bb1) [left=of a3] {};
  \node [draw,circle,fill=red!50,"$12$" below] (bb2) [left=of bb1] {};
  
  \node [draw,circle,fill=red!50,"$11$" below] (a6) [left=of bb2] {};
  \node [draw,circle,fill=red!50,"$4$" below] (a1) [left=of a6] {};
  \node [draw,circle,fill=red!50,"$10$" left]  (a8) [above left=of a1] {};
   \node [draw,circle,fill=blue!50,"$14$" below]  (10) at (12.7,-3.2) {};
   \node [draw,circle,fill=blue!50,"$13$" below]  (9) at (11,-3.2) {};

  \begin{scope}[>=Latex]
  
     \draw[-> , thick]  (a2) edge (a1);
   \draw[-> , thick]  (a3) edge (aa2);
  \draw[-> , thick]  (a1) edge (a6);
  \draw[-> , thick]  (a5) edge (a1);
  \draw[-> , thick]  (a2) edge (a8);
  \draw[-> , thick] (a7) edge (a2);
    \draw[-> , thick] (a2) edge (aa1);
      \draw[-> , thick] (a4) edge (aa2);
 
       \draw[-> , thick] (a6) edge (bb2);
 \draw[-> , thick]  (a3) edge (a5);
 \draw[-> , thick]  (bb1) edge (a4);
  \draw[-> , thick]   (a4) edge (a5);
   \draw[-> , thick]   (a7) edge[bend left=35] (a5);
    \draw[-> , thick]  (a5) edge(bb1);
    \draw[-> , thick] (a5) edge(a8);
    \draw[-> , thick]  (a8) edge(a6);
    \draw[-> , thick] (a6) edge(a7);
    
     \draw[-> , thick] (aa2) edge(a5);
     \draw[-> , thick] (aa1) edge (aa2);
     \draw[-> , thick] (bb2) edge(bb1);
     \draw[-> , thick] (bb1) edge(a3);
     \draw[-> , thick] (bb2) edge(a2);
     \draw[-> , thick] (bb1) edge(aa1);
     \draw[-> , thick] (aa1) edge(a6);
     \draw[-> , thick] (aa2) edge(bb2);
    
    \draw[-> , thick] (9) edge (a1);
     \draw[-> , thick] (a4) edge (9);
     
      \draw[-> , thick] (10) edge (a1);
     \draw[-> , thick] (a3) edge (10);

    \end{scope}

\end{tikzpicture}
}

\caption{Extended quiver $Q_{D_{6}}$ associated with the deformed $\rD_6$}
\label{quiverD6}
\end{figure}
As a result of this insertion, the parameters $b_{5}$ and $b_{6}$ appear in the cluster variables generated by mutations in cluster algebra determined by the pair consisting of matrix $\tilde{B}_{\rD_{6}}$ (or, equivalently, the quiver $Q_{D_{6}}$) and initial cluster 
\begin{align*}
\hat{\vb{x}}&= (\xi_{1,0},r_{0},s_{0},\tau_{0}, \tau_{1}, \tau_{2},\tau_{3},\tau_{4}, \tau_{5} ,\tau_{6},\chi_{0},\eta_{1,0}, b_{5},b_{6})  \\
		&= (\tilde{x}_{1},\tilde{x}_{2}, \tilde{x}_{4}, \tilde{x}_{5}, \tilde{x}_{6}, \tilde{x}_{7}, \tilde{x}_{8}, \tilde{x}_{9 }, \tilde{x}_{10}, \tilde{x}_{11}, \tilde{x}_{12},  \tilde{x}_{13},   \tilde{x}_{14} )
\end{align*}
%
%
Thus we reaches to the following result. 
\begin{thm}[Laurentification of the deformed map] Let $(\hat{\vb{x}},\hat{B}_{\rD_{6}})$ be initial seed which is composed of extended initial cluster
\begin{equation}\label{D6initialvar}
\begin{split}
\hat{\vb{x}}&= (\xi_{1,0},r_{0},s_{0},\tau_{0}, \tau_{1}, \tau_{2},\tau_{3},\tau_{4}, \tau_{5} ,\tau_{6},\chi_{0},\eta_{1,0}, b_{5},b_{6})  \\
		&= (\tilde{x}_{1},\tilde{x}_{2}, \tilde{x}_{4}, \tilde{x}_{5}, \tilde{x}_{6}, \tilde{x}_{7}, \tilde{x}_{8}, \tilde{x}_{9 }, \tilde{x}_{10}, \tilde{x}_{11}, \tilde{x}_{12},  \tilde{x}_{13},   \tilde{x}_{14} )
\end{split}
\end{equation}
together with the associated extended exchange matrix
\begin{equation}\label{exchmD6wfrozen}
\footnotesize	
    \tilde{B}_{D_{6}}= 
      \left(
\begin{array}{*{12}c}
0 & 1 & 1 & 0 & 0 & 0 & 1 & 0 & -1 & 0 & 0  & -1  \\ 
 -1 & 0 & 0 & 0 & 0 & 0 & 0 & 1 & 1 & 0 & 0 & 0 \\
  -1 & 0 & 0 & 0 & 0 & 0 & 0 & 1 & 1 & 0 & 0 & 0 \\ 
   0 &  0 &  0 & 0 & 0 & -1 & 0 & 0 & -1 & 0 & 1 & 0 \\
 0 & 0 & 0 & 0 & 0 & 1 & 0 & 0 & 1 & 0 & -1 & 0 \\
 0 & 0 & 0 & 1  & -1 & 0 & 1 & 0 & 0 & 1 & 0 & -1 \\ 
 -1 & 0 & 0 & 0  & 0 & -1 & 0 & 1 & 0 & 0 & 1 & 0 \\ 
 0 & -1 & -1 & 0 & 0 & 0 & -1 & 0 & 1 & 0 & 0 & 1  \\ 
 1 & -1& -1 & 1 & -1 & 0  & 0 & -1 & 0 & 1 & 0 & 0 \\ 
 0 & 0 & 0  & 0 & 0  & -1 & 0 & 0 & -1 & 0 & 1 & 0 \\ 
 0 & 0 & 0 & -1 & 1 & 0 & -1  & 0 & 0 & -1 & 0 & 1  \\ 
  0 & 0 & 0 & 0 & 0 & 1 & 0  & -1 & 0 & 0 & -1 & 0  \\ 
   0&0 &-1 & 1 & 0 & 0 & 0 & 0 & 0 & 0 & 0 & 0 \\
  0 & -1 & 0 & 1 & 0 & 0 & 0 & 0 & 0 & 0 & 0 & 0
    \end{array}
\right)
\end{equation}
and consider the permutation $\rho = (456789\underline{10})(23)$ . Then the iteration of  cluster map $\psi_{D_{6}} = \rho^{-1} \mu_{12}\mu_{11}\mu_{10}\mu_{9}\mu_{8}\mu_{1}$ is equivalent to the recurrence \eqref{eq:deformD6}, and for the tau functions  $\tau_{n}$, $\chi_{n}$,  $\eta_{1,n}$, $\xi_{1,n}$, $r_{n}$, $s_{n}$ are elements of the Laurent polynomial ring $$ \mathbb{Z}_{>0}\qty[b_{5},b_{6}, \chi_{0}^{\pm}, \eta_{1,0}^{\pm},\xi_{1,0}^{\pm},r_{0}^{\pm},s_{0}^{\pm}, \tau_{0}^{\pm}, \tau_{1}^{\pm},\tau_{2}^{\pm},\tau_{3}^{\pm},\tau_{4}^{\pm},\tau_{5}^{\pm},\tau_{6}^{\pm}].$$ 

\end{thm}

\subsection{Tropicalization and degree growth for deformed $\rD_{6}$ map} \label{ss:tropD6}

In the previous work \cite{hkm24}, we showed that the degree growth of iteration of cluster maps $\psi_{D_{4}}^{(1)}$ and $\psi_{D_{4}}^{(2)}$,  (induced from deformed type $D_{4}$ map introduced in Section \ref{DefmapD4} ) is quadratic, with a (max,+) relation corresponding to the exchange relations. This results in vanishing algebraic entropy, indicating that the maps are integrable, which is consistent with the integrability of the deformed map $\tilde{\varphi}_{D_{4}}$. In this section, we study the degree growth of cluster map $\psi_{D_{6}}$ and  determine associated algebraic entropy by following the same steps  as in the type $D_4$ case \cite{hkm24}.

As the Laurent property of the cluster map $\psi_{\rD_{6}}$ holds, we can write the sequence of tau-functions $\tau_{n}$, $\chi_{n}$, $\eta_{1,n}$, $\xi_{1,n}$, $r_{n}$ and $s_{n}$ as Laurent polynomial in initial cluster $\hat{\vb{x}}$  \eqref{D6initialvar}, shown below
\begin{equation}\label{eq:LaurentD6m}
\begin{split}
&\tau_{n} = \frac{N^{(1)}_{n}(\hat{\vb{x}})}{\hat{\vb{x}}^{\vb{d}_{n}}}, \quad \chi_{n} = \frac{N^{(2)}_{n}(\hat{\vb{x}})}{\hat{\vb{x}}^{\vb{p}_{n}}}, \quad \eta_{1,n} = \frac{N^{(3)}_{n}(\hat{\vb{x}})}{\hat{\vb{x}}^{\vb{q}_{n}}}, \\
&\xi_{1,n} = \frac{N^{(4)}_{n}(\hat{\vb{x}})}{\hat{\vb{x}}^{\vb{r}_{n}}}, \quad  r_{n} = \frac{N^{(4)}_{n}(\hat{\vb{x}})}{\hat{\vb{x}}^{\vb{v}_{n}}},\quad s_{n} = \frac{N^{(4)}_{n}(\hat{\vb{x}})}{\hat{\vb{x}}^{\vb{j}_{n}}} \\
\end{split}
\end{equation}
where d-vectors (denominator vectors) $\vb{d}_{n}, \vb{p}_{n}, \vb{q}_{n},\vb{r}_{n}, \vb{v}_{n}, \vb{j}_{n} $ is associated with the tau-functions assoicated with unfrozen variables $
(\xi_{1,0},r_{0},s_{0},\tau_{0}, \tau_{1}, \tau_{2},\tau_{3},\tau_{4}, \tau_{5} ,\tau_{6},\chi_{0},\eta_{1,0}, b_{5},b_{6})  $ 
and  have initial data given by a $12 \times 12$ identity matrix.
\begin{equation}\label{initialdvD6}
\qty( \vb{q}_{0} \ \vb{r}_{0} \ \vb{v}_{0} \ \vb{j}_{0} \ \vb{d}_{0} \ \vb{d}_{1} \ \vb{d}_{2} \ \vb{d}_{3} \ \vb{d}_{4} \ \vb{d}_{5} \  \vb{d}_{6} \  \vb{p}_{0}  ) = -I
\end{equation}
Direct substitution \eqref{eq:LaurentD6m} into \eqref{eq:deformD6} and comparison of the exponents in the denominators on both sides  yields the (max,+) relations for the d-vectors,
\begin{equation}\label{bilinearD6deg}
\begin{array}{rcl}
\vb{d}_{n+7} + \vb{d}_{n} & = & \max(\vb{d}_{n+5} + \vb{d}_{n+2}, \vb{p}_{n} ), \\
 \vb{p}_{n+1} + \vb{p}_{n} & = & \max(\vb{d}_{n+2} + \vb{d}_{n+3} +\vb{d}_{n+5} + \vb{d}_{n+6}, \vb{q}_{n} + \vb{d}_{n+1} + \vb{d}_{n+7} ), \\
  \vb{q}_{n+1} + \vb{q}_{n} & = & \max(\vb{d}_{n+3} + \vb{d}_{n+4} +\vb{d}_{n+5} +\vb{d}_{n+6}, \vb{r}_{n}+\vb{p}_{n+1} ), \\
 \vb{r}_{n+1} + \vb{r}_{n} & = & \max(\vb{d}_{n+4} + 2\vb{d}_{n+5}+ \vb{d}_{n+6}, \vb{v}_{n} +  \vb{j}_{n}+  \vb{q}_{n+1}), \\
  \vb{v}_{n+1} + \vb{j}_{n} & = & \max(\vb{d}_{n+5} + \vb{d}_{n+6}, \vb{r}_{n+1}), \\
   \vb{j}_{n+1} + \vb{v}_{n} & = & \max(\vb{d}_{n+5} + \vb{d}_{n+6}, \vb{r}_{n+1}), \\
\end{array}  
\end{equation}

Next we introduce quantities which is analogous to the tropical version of \eqref{vartransformD6} as following,
\begin{equation}\label{eq:tropvarD6x}
\begin{split}
\vb{X}_{1,n} = \vb{d}_{n} + \vb{d}_{n+6} - \vb{d}_{n+1} - \vb{d}_{n+5}, & \quad \vb{X}_{2,n} = \vb{p}_{n} - \vb{d}_{n+2} - \vb{d}_{n+5},\\
\vb{X}_{3,n} = \vb{q}_{n} - \vb{d}_{n+3} - \vb{d}_{n+5}, &  \quad  \vb{X}_{4,n} = \vb{r}_{n} - \vb{d}_{n+4} - \vb{d}_{n+5}, \\ 
\vb{X}_{5,n} = \vb{v}_{n} + \vb{s}_{n}  - \vb{d}_{n+5}, & \quad \vb{X}_{6,n} = \vb{j}_{n} - \vb{d}_{n+5} - \vb{s}_{n},
\end{split}
\end{equation}
along with quantities corresponding to symplectic coordinates   $y_{j}$ \eqref{vartransformD6y}, shown below. 
\begin{equation}\label{eq:tropvarD6y}
\begin{array}{rcl}
\vb{Y}_{1,n} = \vb{p}_{n} - \vb{d}_{n+2} - \vb{d}_{n+5},& \quad &\vb{Y}_{2,n} = \vb{d}_{n+1} + \vb{q}_{n} - \vb{d}_{n} - \vb{d}_{n+3} - \vb{d}_{n+6},\\
  \vb{Y}_{3,n} = \vb{d}_{n+2} + \vb{r}_{n} - \vb{d}_{n+4} - \vb{p}_{n}, & \quad &\vb{Y}_{4,n} = \vb{d}_{n+3} +  \vb{v}_{n} + \vb{j}_{n} - \vb{d}_{n+5} - \vb{q}_{n} \\ 
\end{array}
\end{equation}
It follows that  $\vb{X}_{i,n}$ satisfies ultradiscretized expression of original type $\rD_{6}$ (\eqref{D6recs} with all $b_{i}=1$), which can be obtained from the	(max,+) equations \eqref{bilinearD6deg}, as shown below.
\begin{lm}\label{periodtropD6}
The quantities $\vb{X}_{j,n}$ in \eqref{eq:tropvarD6x} satisfy following system of (max,+) equations: 
\begin{equation} \label{maxD6}
\begin{array}{rcl}
\vb{X}_{1,n+1} + \vb{X}_{1,n} & = & \qty[\vb{X}_{2,n}]_{+}, \\
\vb{X}_{2,n+1} + \vb{X}_{2,n} & = & \qty[\vb{X}_{1,n+1} +\vb{X}_{3,n} ]_{+} ,\\
\vb{X}_{3,n+1} + \vb{X}_{3,n} & = & \qty[\vb{X}_{2,n+1} +\vb{X}_{4,n} ]_{+} ,\\
\vb{X}_{4,n+1} + \vb{X}_{4,n} & = & \qty[\vb{X}_{5,n}+ \vb{X}_{6,n} + \vb{X}_{3,n+1} ]_{+}. \\
\vb{X}_{5,n+1} + \vb{X}_{5,n} & = & \qty[\vb{X}_{4,n+1} ]_{+}. \\
\vb{X}_{6,n+1} + \vb{X}_{6,n} & = & \qty[\vb{X}_{4,n+1} ]_{+}. \\
\end{array}
\end{equation}
where $ \qty[a]_{+} = \max(a,0)$.

Given arbitrary initial values $(\vb{X}_{1,0},\vb{X}_{2,0}, \vb{X}_{3,0}, \vb{X}_{4,0}, \vb{X}_{5,0},\vb{X}_{6,0})$, the quantities $\vb{X}_{j,n}$ ( and $\vb{Y}_{i,n}$) are periodic with period 6 for $1\leq i \leq 4$ and $1 \leq j \leq 6$. 
\end{lm}
\begin{proof}
The (max,+) equations in \eqref{maxD6} arise from \eqref{bilinearD6deg} by rearranging the d-vectors so that they can be expressed in terms of quantities $\vb{X}_{j,n}$ in \eqref{eq:tropvarD6x}. Alternatively, these equations can be derived from the exchange relations of type $D_{6}$ cluster map, $\varphi_{D_{6}}$, \eqref{typeD4:clmu}, by considering the d- vectors of cluster variables $x_j$. Since  $\varphi_{D_{6}}$ is periodic with period 6, it follows that each quantity $\vb{X}_{j,n}$ inherits this periodicity. 
\end{proof}
Note that for $1\leq j \leq 4$, each $\vb{Y}_{j,n}$ satisfies the same periodicity since it can be expressed in terms of $\vb{X}_{i,n}$ for all $1\leq i \leq 6$. The periodicity of $\vb{X}_{i,n}$ (or $\vb{Y}_{j,n}$) is essential in determining the degree growth of the d-vectors of tau functions as demonstrated in the proof of the following statement. 
\begin{thm} Let $\cT$ be linear operator which shifts $n \to n+1$.
The  d-vectors $\vb{e}_{n}$,$\vb{d}_{n}$,$ \vb{f}_{n}$ and $\vb{g}_{n}$, which solve the system of equations \eqref{bilinearD6deg}, satisfy the following linear difference equations
\begin{equation}
\cL\vb{r}_{n}=(\cT^{6} - 1)(\cT^5 - 1)(\cT-1)\vb{r}_{n} = 0 
\end{equation}
where $\cT$ is shift operator corresponding to $n \to n+1$ and $\vb{r}_{n} = \vb{e}_{n}$,$\vb{d}_{n}$,$ \vb{f}_{n}, \vb{g}_{n}$.  For the generated tau functions, the leading order of degree growth of their denominators is given by 
\begin{equation}\label{degreegrowthD6}
\begin{split}
\vb{d}_{n} = \frac{n^2}{60} \vb{a} + O(n),& \quad \vb{v}_{n} = \frac{n^2}{60} \vb{a} + O(n), \quad \vb{j}_{n} = \frac{n^2}{60} \vb{a} + O(n)   \\
\vb{p}_{n} = \frac{n^2}{30}\vb{a} + O(n),& \quad \vb{q}_{n} = \frac{n^2}{30}\vb{a} + O(n), \quad \vb{r}_{n}=\frac{n^2}{30}\vb{a} + O(n),
\end{split}
\end{equation}
where $\vb{a} = (2,1,1,1,1,1,1,1,1,1,2,2)^{T}$. 
\end{thm}
\begin{proof} The quantity $\vb{X}_{1,n}$ can be written in terms of the shift  operator as following 
\begin{equation}\label{quantityX1}
\vb{X}_{1,n} = (\cT^6 - \cT^5 - \cT + 1)\vb{d}_{n}
\end{equation}
Since $\vb{X}_{1,n}$ exhibit periodicity of period $6$, shown in Lemma \ref{periodtropD6}, it follows that the quantity \eqref{quantityX1} yields a characteristic equation in $\cT$,
\begin{equation}\label{relX1}
 (\cT^6 -1 )\vb{X}_{1,n} = (\cT-1)^{3} (\sum_{i=0}^{5}\cT^i) (\sum_{i=0}^{4}\cT^i)\vb{d}_{n}  = 0
\end{equation}
Solving this equation provides the general solution for the $\vb{d}_{n}$ whose leading order term is $n^2$ with some constant vector $\vb{a}$, 
\begin{equation}
\vb{d}_{n} = \vb{a}n^{2} + O(n) 
\end{equation}
We now consider an relation \eqref{relX1} once again to determine the coefficient constant of the leading term of $\vb{d}_{n}$. The expression can be rewritten in the form $(\cT^6 -1 )(\cT^5 - 1)(\cT-1)\vb{d}_{n}$ which yields following relation 
\begin{equation}\label{dnrelation}
\begin{split}
&(\cT^6 -1 )(\cT^5 - 1)(\cT-1)\vb{d}_{n}=  0 \\
& \implies (\cT^6 -1 )(\cT^5 - 1)\vb{d}_{n+1} = (\cT^6 -1 )(\cT^5 - 1)\vb{d}_{n}.
\end{split}
\end{equation}
Substituting the general solution into this relation gives
\begin{equation}\label{D6constantrel}
(\cT^6 -1 )(\cT^5 - 1)\vb{d}_{n} = 60 \vb{a}
\end{equation}
Evaluating the sequence of $\vb{d}_{n}$, obtained from the (max,+) relations \eqref{bilinearD6deg} with initial d-vectors \eqref{initialdvD6}, and then substituting  it into \eqref{D6constantrel}, we obtain  
$$60 \vb{a} = (2,1,1,1,1,1,1,1,1,2,2)^{T} $$
Thus the leading order term of  $\vb{d}_{n}$ is:
$$\vb{d}_{n} = \frac{n^2}{60} \vb{a} + O(n)$$
which agrees with one of the results in \eqref{degreegrowthD6}. 

In the case of d-vector $\vb{p}_{n}$, by using periodicity of quantity $\vb{X}_{2,n}$, $(\cT^6 - 1) \vb{X}_{2,n} = 0$ and applying  $(\cT^5-1)(\cT-1)$ to both sides, one finds 
\begin{align*}
(\cT^6 - 1)(\cT^5-1)(\cT-1)\vb{p}_{n} =0
\end{align*} 
where $\vb{d}_{n+2}$ and $\vb{d}_{n+5}$ vanish since the $\vb{d}_{n}$ satisfy the relation \eqref{dnrelation}. By using a similar argument above, one can show that the d-vectors  $\vb{q}_{n}$ and $\vb{r}_{n}$ satisfy the same linear relation of $\cT$ in \eqref{dnrelation}. 

As for the rest of d-vectors $\vb{v}_{n}$ and $\vb{j}_{n}$, we first consider the last two (max,+) relations in \eqref{bilinearD6deg}.  Subtracting these relations then yields the following expression 
\begin{equation}\label{eq:D6vj}
(\cT - 1)\vb{v}_{n} =  (\cT - 1)\vb{j}_{n}.
\end{equation}
From the periodicity $(\cT^6 -1) \vb{Y}_{4,n} = 0 $, the same procedure used in the case for $\vb{p}_{n}$ gives   
 \begin{equation}
 (\cT^6 -1)(\cT^5 -1)(\vb{v}_{n+1} + \vb{j}_{n+1} - \vb{v}_{n} - \vb{j}_{n}) = 0
 \end{equation}
 Then substituting the identity \eqref{eq:D6vj}, we find that both $\vb{v}_{n}$ and $\vb{j}_{n}$ are solutions of the recurrence in \eqref{dnrelation}.
 
For the rest of proof (determining the coefficients of leading order terms of d-vectors), we can take the same approach as for the case $\vb{d}_{n}$ and the required results follows. 
\end{proof}

Since the degree growth of each variable is quadratic, associated algebraic entropy vanishes, which leads to the conjecture that the deformed type $\rD_{6}$ map is a Liouville integrable map.

\newpage

\section{Local expansion} \label{S:localexpand}


The main result of the previous work \cite{grab} is that the quiver $Q_{A_{2N}}$, arising from the deformed type $A_{2N}$ map via Laurentification, can be obtained by inserting a specific subquiver (referred to as \textit{local expansion}) into the quiver $Q_{A_2}$, which itself arises from the Laurentification of the deformed type $A_{2}$ map, as shown in the Figure \ref{Q4toQ6}. Here we take similar approach to the type $D_{4}$.

\begin{figure}[h!]
\begin{center}
\resizebox{1 \textwidth}{!}{%
 \begin{tikzpicture}[every circle node/.style={draw,scale=0.6,thick},node distance=15mm]

  \node [draw,circle,fill=blue!50,"$12$"] (12) at (0,0) {};
  
     \node [draw,circle,fill=red!50,"$6$"] (6) [right= of 12] {};
      \node [draw,circle,fill=red!50,"$7$"] (7) [right=of 6] {};
      \node [draw,circle,fill=red!50,"$8$"] (8) [right=of 7] {};
      \node [draw,circle,fill=red!50,"$9$"] (9) [right=of 8] {};
      \node [draw,circle,fill=red!50,"$10$"right] (10) [below right=of 9] {};
      \node [draw,circle,fill=red!50,"$5$" left] (5) [below left=of 6] {};

      \node [draw,circle,fill=blue!50,"$13$"] (13) [right=of 9] {};

       \node [draw,circle,fill=red!50,"$4$"below] (4) [below right=of 5] {};
       \node [draw,circle,fill=red!50,"$11$"below] (11) [right=of 4] {};
       \node [draw,circle,fill=red!50,"$1$"below] (1) [right=of 11] {};
       \node [draw,circle,fill=red!50,"$2$"below] (2) [right=of 1] {};
       
        \node [draw,circle,fill=red!50,"$3$" below] (3) at (4.4,-3.5) {};
       
       \node (a) at (4.4,-5.5) {\Large (a) $Q_{A_{4}}$};

  \begin{scope}[>=Latex]
  
  \draw[-> , thick]  (1) edge (2); 
  \draw[-> , thick]  (1) edge (7); 
  \draw[-> , thick]  (9) edge (1);
  \draw[-> , thick]  (11) edge (1);
   \draw[-> , thick]  (1) edge (10);

 \draw[-> , thick]  (2) edge (3);
 \draw[-> , thick]  (2) edge (8);
 \draw[-> , thick]  (2) edge (13);

 \draw[-> , thick]  (4) edge (11);
  \draw[-> , thick]  (12) edge (4);
   \draw[-> , thick]  (7) edge (4);
    \draw[-> , thick]  (3) edge (4);
    
 \draw[-> , thick]  (5) edge (12);
  \draw[-> , thick]  (5) edge (11);
   \draw[-> , thick]  (3) edge (5);
    \draw[-> , thick]  (7) edge (5);
   
   \draw[-> , thick]  (6) edge (12);
    \draw[-> , thick]  (6) edge (7);
     \draw[-> , thick]  (11) edge (6);
      \draw[-> , thick]  (6) edge[bend left= 15] (3);
      
    \draw[-> , thick]  (7) edge (8);
    
     \draw[-> , thick]  (8) edge (11);
      \draw[-> , thick]  (8) edge (9);
       \draw[-> , thick]  (10) edge (8);
       
        \draw[-> , thick]  (13) edge (9);
         \draw[-> , thick]  (3) edge[bend left=15] (9);
         
          \draw[-> , thick]  (13) edge (10);
           \draw[-> , thick]  (10) edge (3);
           
            \draw[-> , thick]  (3) edge (13);
            
             \draw[-> , thick]  (12) edge (3);

    \end{scope}

\draw [-{Latex[length=3mm]}] (10,-1.2) -- (11,-1.2) node[midway,sloped,above] {Expansion};

  \node [draw,circle,fill=blue!50,"$12$"] (12) at (12,0) {};
  
     \node [draw,circle,fill=red!50,"$7$"] (6) [right= of 12] {};
      \node [draw,circle,fill=red!50,"$8$"] (7) [right=of 6] {};
       \node [draw,circle,fill=green!50,"$9$"] (a1) [right=of 7] {};
        \node [draw,circle,fill=green!50,"$10$"] (a2) [right=of a1] {};
      
      \node [draw,circle,fill=red!50,"$11$"] (8) [right=of a2] {};
      \node [draw,circle,fill=red!50,"$12$"] (9) [right=of 8] {};
      \node [draw,circle,fill=red!50,"$13$"right] (10) [below right=of 9] {};
      \node [draw,circle,fill=red!50,"$6$" left] (5) [below left=of 6] {};

      \node [draw,circle,fill=blue!50,"$17$"] (13) [right=of 9] {};

       \node [draw,circle,fill=red!50,"$5$"below] (4) [below right=of 5] {};
       \node [draw,circle,fill=red!50,"$14$"below] (11) [right=of 4] {};
        \node [draw,circle,fill=green!50,"$15$"below] (b1) [right=of 11] {};
         \node [draw,circle,fill=green!50,"$1$"below] (b2) [right=of b1] {};
       
       \node [draw,circle,fill=red!50,"$2$"below] (1) [right=of b2] {};
       \node [draw,circle,fill=red!50,"$3$"below] (2) [right=of 1] {};
       
        \node [draw,circle,fill=red!50,"$4$" below] (3) at (18,-3.5) {};
       
       \node (a) at (18,-5.5) {\Large (b) $Q_{A_{6}}$};

  \begin{scope}[>=Latex]
  
  \draw[-> , thick]  (1) edge (2); 
 
  \draw[-> , thick]  (9) edge (1);

   \draw[-> , thick]  (1) edge (10);

 \draw[-> , thick]  (2) edge (3);
 \draw[-> , thick]  (2) edge (8);
 \draw[-> , thick]  (2) edge (13);

 \draw[-> , thick]  (4) edge (11);
  \draw[-> , thick]  (12) edge (4);
   \draw[-> , thick]  (7) edge (4);
    \draw[-> , thick]  (3) edge (4);
    
 \draw[-> , thick]  (5) edge (12);
  \draw[-> , thick]  (5) edge (11);
   \draw[-> , thick]  (3) edge[bend left=35] (5);
    \draw[-> , thick]  (7) edge (5);
   
   \draw[-> , thick]  (6) edge (12);
    \draw[-> , thick]  (6) edge (7);
     \draw[-> , thick]  (11) edge (6);
      \draw[-> , thick]  (6) edge[bend left= 25] (3);

      \draw[-> , thick]  (8) edge (9);
       \draw[-> , thick]  (10) edge (8);
       
        \draw[-> , thick]  (13) edge (9);
         \draw[-> , thick]  (3) edge[bend left=25] (9);
         
          \draw[-> , thick]  (13) edge (10);
           \draw[-> , thick]  (10) edge[bend left= 35] (3);
           
            \draw[-> , thick]  (3) edge (13);
            
             \draw[-> , thick]  (12) edge (3);
        

 \draw[-> , thick]  (7) edge (a1);
  \draw[-> , thick]  (a1) edge (a2);
   \draw[-> , thick]  (a2) edge (8);
    \draw[-> , thick]  (11) edge (b1);
     \draw[-> , thick]  (b1) edge (b2);
      \draw[-> , thick]  (b2) edge (1);
      
      \draw[-> , thick]  (b1) edge (7);
     \draw[-> , thick]  (b2) edge (a1);
      \draw[-> , thick]  (1) edge (a2);

      \draw[-> , thick]  (a1) edge (11);
	\draw[-> , thick]  (a2) edge (b1);  
	\draw[-> , thick]  (8) edge (b2);

    \end{scope}

 \node [draw,circle,fill=red!50,"$7$"] (5) at (6,-7) {};
   \node [draw,circle,fill=red!50,"$8$"] (6)[right=of 5]{};
  
  \node [draw,circle,fill=red!50,"$11$"below] (4) at (6,-9.2) {};
   \node [draw,circle,fill=red!50,"$1$"below] (7) [right=of 4] {};
    
  \begin{scope}[>=Latex]
            
       \draw[-> , thick]  (5) edge (6);
        \draw[-> , thick]  (6) edge (4);
         \draw[-> , thick]  (4) edge (7);
          \draw[-> , thick]  (7) edge (5);

    \end{scope}
\draw [-{Latex[length=3mm]}] (10,-8.2) -- (11,-8.2) node[midway,sloped,above] {Expansion };
    \node [draw,circle,fill=red!50,"$8$"] (6) at (13,-7) {};
  \node [draw,circle,fill=red!50,"$14$" below] (4) at (13,-9.2) {};
   \node [draw,circle,fill=green!50,"$9$"] (7) [right=of 6] {};
  \node [draw,circle,fill=green!50,"$15$"below] (11) [right=of 4 ]  {};
   \node [draw,circle,fill=green!50,"$10$"] (8) [right=of 7] {};
  \node [draw,circle,fill=green!50,"$1$" below] (1) [right=of 11 ]  {};
    \node [draw,circle,fill=red!50,"$11$"] (9) [right=of 8] {};
  \node [draw,circle,fill=red!50,"$2$" below] (2) [right=of 1 ]  {};

  \begin{scope}[>=Latex]
   
    \draw[-> , thick]  (7) edge (8);
     \draw[-> , thick]  (8) edge (11);
      \draw[-> , thick]  (11) edge (1);
       \draw[-> , thick]  (1) edge (7);
       
        \draw[-> , thick]  (8) edge (9);
         \draw[-> , thick]  (9) edge (1);
          \draw[-> , thick]  (1) edge (2);
           \draw[-> , thick]  (2) edge (8);
       
        \draw[-> , thick]  (6) edge (7);
         \draw[-> , thick]  (7) edge (4);
          \draw[-> , thick]  (4) edge (11);
           \draw[-> , thick]  (11) edge (6);

    \end{scope}

\end{tikzpicture}
}
\end{center}
\caption{Extension from $Q_{A_{4}}$ to $Q_{A_{6}}$. Green nodes are new vertices which are inserted in $Q_{A_{4}}$ in the form of local expansion.}
\label{Q4toQ6}
\end{figure}
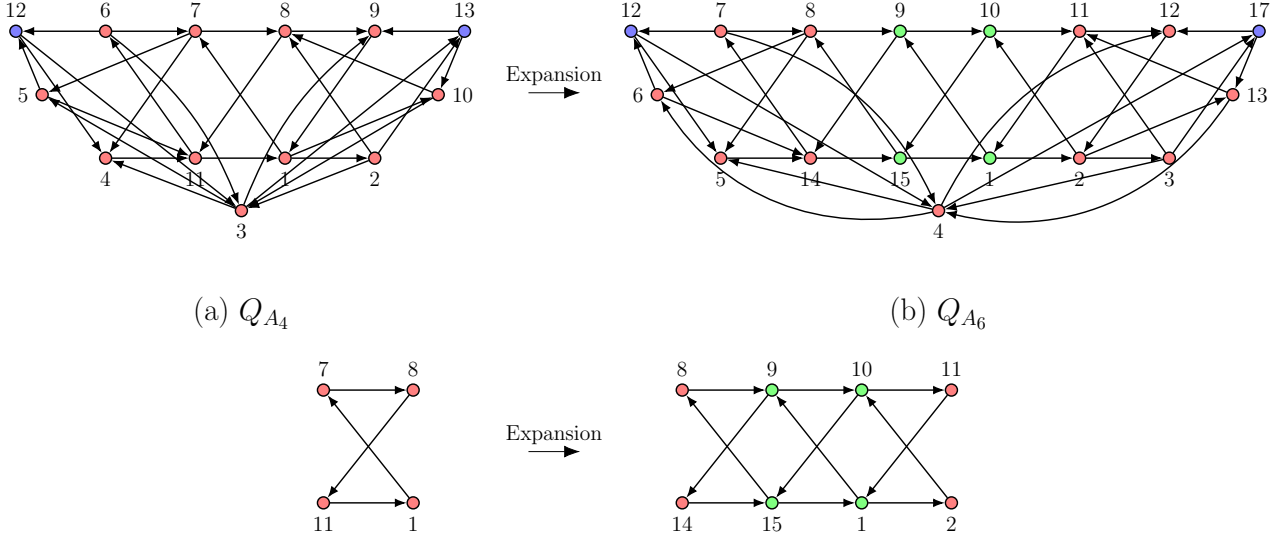


Comparing the quivers, $Q_{D_{4}}^{(2)}$ (Figure \ref{DeformedQD4} (b))  and $Q_{D_{6}}$ (Figure \ref{quiverD6}), one can see that the  $Q_{D_{4}}^{(2)}$ can be extended to  $Q_{D_{6}}$ via a transformation that splits the node 3 and 5 into two nodes each, then insert a subquiver that has the same structure as illustrated in Figure \ref{extensionD4toD6}, and relabels the nodes accordingly.  
\begin{figure}[h!]
\begin{center}
\resizebox{1\textwidth}{!}{%
 \begin{tikzpicture}[every circle node/.style={draw,scale=0.6,thick},node distance=15mm]
 
  \node [draw,circle,fill=red!50,"$4$"] (a7) at (0,0) {};
  \node [draw,circle,fill=red!50,"$5$"] (a2) [right=of a7] {};
  \node [draw,circle,fill=red!50,"$6$"] (a5) [right=of a2] {};
  \node [draw,circle,fill=red!50,"$2$" right] (a4) [below right=of a5] {};
  \node [draw,circle,fill=red!50,"$1$" below] (a3) [below left=of a4] {};
  \node [draw,circle,fill=red!50,"$8$" below] (a6) [left=of a3] {};
  \node [draw,circle,fill=red!50,"$3$" below] (a1) [left=of a6] {};
  \node [draw,circle,fill=red!50,"$7$" left]  (a8) [above left=of a1] {};
   \node [draw,circle,fill=blue!50,"$10$" below]  (10) at (2.7,-3.2) {};
   \node [draw,circle,fill=blue!50,"$9$" below]  (9) at (1,-3.2) {};
   
   \node (b) [below=of a6] {\Large (a) $Q_{D_{4}}$};
 
  \begin{scope}[>=Latex]
  
  \draw[-> , thick]  (a2) edge (a1);
   \draw[-> , thick]  (a3) edge (a2);
  \draw[-> , thick]  (a1) edge (a6);
  \draw[-> , thick]  (a5) edge (a1);
  \draw[-> , thick]  (a2) edge (a8);
  \draw[-> , thick] (a7) edge (a2);
    \draw[-> , thick] (a2) edge (a5);
      \draw[-> , thick] (a4) edge (a2);
       \draw[-> , thick] (a4) edge (a2);
       \draw[-> , thick] (a6) edge (a3);
 \draw[-> , thick]  (a3) edge (a5);
 \draw[-> , thick]  (a6) edge (a4);
  \draw[-> , thick]   (a4) edge (a5);
   \draw[-> , thick]   (a7) edge[bend left=35] (a5);
    \draw[-> , thick]  (a5) edge(a6);
    \draw[-> , thick] (a5) edge(a8);
    \draw[-> , thick]  (a8) edge(a6);
    \draw[-> , thick] (a6) edge(a7);
    
    \draw[-> , thick] (9) edge(a1);
     \draw[-> , thick] (a4) edge(9);
     \draw[-> , thick] (10) edge(a1);
     \draw[-> , thick] (a3) edge(10);

    \end{scope}
\draw [-{Latex[length=3mm]}] (5.5,-1.2) -- (6.5,-1.2) node[midway,sloped,above] {Expansion};

     \node [draw,circle,fill=red!50,"$5$"] (a7) at (9,0) {};
  \node [draw,circle,fill=red!50,"$6$"] (a2) [right=of a7] {};
  
   \node [draw,circle,fill=green!50,"$7$"] (aa1) [right=of a2] {};
     \node [draw,circle,fill=red!50,"$8$"] (aa2) [right=of aa1] {};

  \node [draw,circle,fill=red!50,"$9$"] (a5) [right=of aa2] {};
  \node [draw,circle,fill=red!50,"$3$" right] (a4) [below right=of a5] {};
  \node [draw,circle,fill=red!50,"$2$" below] (a3) [below left=of a4] {};
  
  \node [draw,circle,fill=red!50,"$1$" below] (bb1) [left=of a3] {};
  \node [draw,circle,fill=green!50,"$12$" below] (bb2) [left=of bb1] {};
  
  \node [draw,circle,fill=red!50,"$11$" below] (a6) [left=of bb2] {};
  \node [draw,circle,fill=red!50,"$4$" below] (a1) [left=of a6] {};
  \node [draw,circle,fill=red!50,"$10$" left]  (a8) [above left=of a1] {};
   \node [draw,circle,fill=blue!50,"$14$" below]  (10) at (12.7,-3.2) {};
   \node [draw,circle,fill=blue!50,"$13$" below]  (9) at (11,-3.2) {};

   \node (b) [below=of bb2] {\Large (b) $Q_{D_{6}}$};

  \begin{scope}[>=Latex]
  
     \draw[-> , thick]  (a2) edge (a1);
   \draw[-> , thick]  (a3) edge (aa2);
  \draw[-> , thick]  (a1) edge (a6);
  \draw[-> , thick]  (a5) edge (a1);
  \draw[-> , thick]  (a2) edge (a8);
  \draw[-> , thick] (a7) edge (a2);
    \draw[-> , thick] (a2) edge (aa1);
      \draw[-> , thick] (a4) edge (aa2);
 
       \draw[-> , thick] (a6) edge (bb2);
 \draw[-> , thick]  (a3) edge (a5);
 \draw[-> , thick]  (bb1) edge (a4);
  \draw[-> , thick]   (a4) edge (a5);
   \draw[-> , thick]   (a7) edge[bend left=35] (a5);
    \draw[-> , thick]  (a5) edge(bb1);
    \draw[-> , thick] (a5) edge(a8);
    \draw[-> , thick]  (a8) edge(a6);
    \draw[-> , thick] (a6) edge(a7);
    
     \draw[-> , thick] (aa2) edge(a5);
     \draw[-> , thick] (aa1) edge (aa2);
     \draw[-> , thick] (bb2) edge(bb1);
     \draw[-> , thick] (bb1) edge(a3);
     \draw[-> , thick] (bb2) edge(a2);
     \draw[-> , thick] (bb1) edge(aa1);
     \draw[-> , thick] (aa1) edge(a6);
     \draw[-> , thick] (aa2) edge(bb2);
    
    \draw[-> , thick] (9) edge (a1);
     \draw[-> , thick] (a4) edge (9);
     
      \draw[-> , thick] (10) edge (a1);
     \draw[-> , thick] (a3) edge (10);

    \end{scope}

   \node [draw,circle,fill=red!50,"$5$"] (5) at (3.5,-7) {};
  \node [draw,circle,fill=red!50,"$8$"below] (8) at (3.5,-9.2) {};
  
 \node (4) at (2.5,-7) {}; 
  \node (3) [below=of 4] {}; 
  \node (7) at (2.5,-7.5) {}; 
  \node (3a) at (2.5,-9.2){};
    \node (7a) at (2.5,-8.8) {}; 
    \node (4a) at (2.5,-7.7) {};

   \node (6) at (4.5,-7) {}; 
    \node (1) [below=of 6] {}; 
  \node (2) at (4.5,-7.5){};
  \node (1a) at (4.5,-9.2){};
   \node (2a) at (4.5,-8.8) {}; 
   \node (6a) at (4.5,-7.7) {};

  \begin{scope}[>=Latex]
      \draw[-> , thick]  (4) edge (5);
      \draw[- , thick]  (5) edge (7);
      \draw[- , thick]  (5) edge (3);
      \draw[-> , thick]  (3a) edge (8);
      
       \draw[- , thick]  (5) edge (6);
       \draw[-> , thick]  (2) edge (5);
        \draw[-> , thick]  (1) edge (5);
        \draw[- , thick]  (8) edge (1a);
         \draw[-> , thick]  (7a) edge (8);
          \draw[- , thick]  (8) edge (4a);
          
           \draw[->, thick]  (6a) edge (8);
            \draw[- , thick]  (8) edge (2a) ;

    \end{scope}
\draw [-{Latex[length=3mm]}] (5.5,-8.2) -- (6.5,-8.2) node[midway,sloped,above] {Expansion };
    \node [draw,circle,fill=red!50,"$6$"] (6) at (8.5,-7) {};
  \node [draw,circle,fill=red!50,"$11$" below] (11) at (8.5,-9.2) {};
   \node [draw,circle,fill=green!50,"$7$"] (7) [right=of 6] {};
  \node [draw,circle,fill=green!50,"$12$"below] (12) [right=of 11 ]  {};
   \node [draw,circle,fill=red!50,"$8$"] (8) [right=of 7] {};
  \node [draw,circle,fill=red!50,"$1$" below] (1) [right=of 12 ]  {};

    \node (5) at (7.5,-7) {}; 
  \node (4) [below=of 5] {}; 
  \node (10) at (7.5,-7.5) {}; 
  \node (4a) at (7.5,-9.2){};
    \node (10a) at (7.5,-8.8) {}; 
    \node (5a) at (7.5,-7.7) {};

   \node (9) at (12.9,-7) {}; 
    \node (2) [below=of 9] {}; 
  \node (3) at (12.9,-7.5){};
  \node (2a) at (12.9,-9.2){};
   \node (3a) at (12.9,-8.8) {}; 
   \node (9a) at (12.9,-7.7) {};

  \begin{scope}[>=Latex]
  
    \draw[-> , thick]  (5) edge (6);
    \draw[-, thick]  (6) edge (10);
     \draw[- , thick]  (6) edge (4);
      \draw[-> , thick]  (4a) edge (11);
       \draw[-> , thick]  (10a) edge (11);
       \draw[- , thick]  (11) edge (5a);
       
       \draw[- , thick]  (8) edge (9);
       \draw[-> , thick]  (3) edge (8);
       \draw[-> , thick]  (2) edge (8);
       \draw[-> , thick]  (9a) edge (1);
       \draw[- , thick]  (1) edge (3a);
       \draw[- , thick]  (1) edge (2a);
       
       \draw[-> , thick]  (6) edge (7);
        \draw[-> , thick]  (7) edge (8);
         \draw[-> , thick]  (11) edge (12);
          \draw[-> , thick]  (12) edge (1);
           \draw[-> , thick]  (12) edge (6);
            \draw[-> , thick]  (7) edge (11);
             \draw[-> , thick]  (1) edge (7);
              \draw[-> , thick]  (8) edge (12);

    \end{scope}

\end{tikzpicture}
}
\end{center}
\caption{Extension from $Q_{D_{4}}$ to $Q_{D_{6}}$.}\label{extensionD4toD6}
\end{figure}
This expansion suggests that the cluster $\tilde{\vb{x}}$, associated with $Q_{D_{4}}^{(2)}$, extends to the type $D_6$ case by shifting the subindices $1\leq i \leq 5 $ of the $\tx_i$ by 1 ($\tx_i \to \tx_{i+1}$) and while the remaining indices are shifted by 3  ($\tx_i \to \tx_{i+3}$). The new cluster variables $\tx_{1}$, $\tx_{7}$, $\tx_8$, $\tx_{12}$, corresponding to additional nodes, are then inserted into the cluster $\tilde{\vb{x}}$, resulting in the cluster $\hat{\vb{x}}$ \eqref{D6initialvar}, which is associated with $Q_{D_{6}}$. From these examples, one can infer that applying this expansion recursively constructs a family of quivers with $4N$ nodes for $N\geq 3$ that arises from the Laurentification of deformed type $D_{2N}$ map,$\tilde{\varphi}_{D_{2N}}$ whose iteration is given by  
\begin{equation}\label{D2Ndmap}
\begin{split}
x_{1,n+1}x_{1,n} &= b_{2N-1}b_{2N} + x_{2,n}\\
x_{2,n+1}x_{2,n} &= b_{2N-1}b_{2N} + x_{1,n+1}x_{3,n}\\
&\vdots \\ 
x_{N-2,n+1}x_{N-2,n} &= b_{2N-1}b_{2N} + x_{N-3,n+1}x_{N-1,n}x_{N,n}\\
x_{N-1,n+1}x_{N-1,n} &= b_{2N-1} + x_{N-2,n+1}\\
x_{N,n+1}x_{N,n} &= b_{2N} + x_{N-2,n+1}\\
\end{split}
\end{equation}
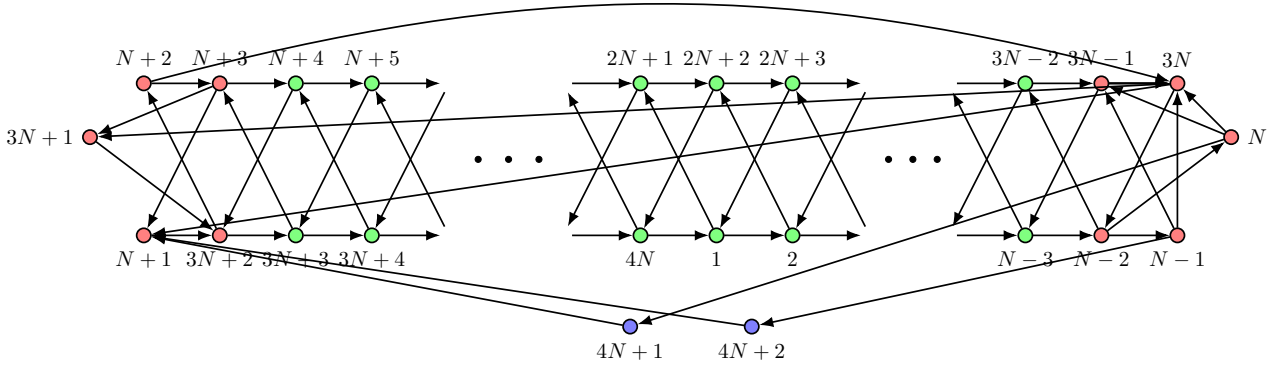
\begin{figure}[h!]

\begin{center}
\resizebox{1\textwidth}{!}{%
 \begin{tikzpicture}[every circle node/.style={draw,scale=0.6,thick},node distance=10mm]

  \node [draw,circle,fill=red!50,"\footnotesize{$N+2$}"] (bb1) at (1,0) {};
  \node [draw,circle,fill=red!50,"\footnotesize{$N+3$}"] (aa3) [right=of bb1] {};
  \node [draw,circle,fill=green!50,"\footnotesize{$N+4$}"] (bb3) [right=of aa3] {};
  \node [draw,circle,fill=green!50,"\footnotesize{$N+5$}"] (aa5) [right=of bb3] {};
  \node [draw,circle,fill=red!50,"\footnotesize{$3N+1$}" left]  (cc1) [below left=of bb1] {};
  \node  (bb5) [right=of aa5] {};

  \node [draw,circle,fill=red!50,"\footnotesize{$N+1$}" below] (aa2) at (1,-2.5) {} ; 
  \node [draw,circle,fill=red!50,"\footnotesize{$3N+2$}" below] (bb2) [right=of aa2] {};
  \node [draw,circle,fill=green!50,"\footnotesize{$3N+3$}" below] (aa4) [right=of bb2] {};
   \node [draw,circle,fill=green!50,"\footnotesize{$3N+4$}" below] (bb4) [right=of aa4] {};
   \node (aa6) [right=of bb4] {};

  \begin{scope}[>=Latex]

\draw[-> , thick]  (cc1) edge (bb2);

\draw[-> , thick]  (aa3) edge (cc1);

\draw[-> , thick]  (bb2) edge (bb1); 
\draw[-> , thick]  (bb1) edge (aa3);

\draw[-> , thick]  (aa3) edge (aa2); 
\draw[-> , thick]   (aa2) edge (bb2); 

\draw[-> , thick]  (aa4) edge (aa3);
\draw[-> , thick]  (aa3) edge (bb3);

\draw[-> , thick]   (bb2) edge (aa4); 
\draw[-> , thick]    (aa5) edge (aa4); 
\draw[-> , thick]    (aa4) edge (bb4);

\draw[-> , thick]   (bb3) edge (aa5);
\draw[-> , thick]  (aa6) edge (aa5);

\draw[-> , thick]   (bb3) edge (bb2);
\draw[-> , thick]  (bb4) edge (bb3);
\draw[-> , thick]  (bb5) edge (bb4);

\draw[-> , thick]  (bb4) edge (aa6);
\draw[-> , thick]  (aa5) edge (bb5);

\end{scope}

 \node [draw,circle,fill=green!50,"\footnotesize{$3N-2$}"] (b2) at (15.5,0) {};
  \node [draw,circle,fill=red!50,"\footnotesize{$3N-1$}"] (a4) [right=of b2] {};
  \node [draw,circle,fill=red!50,"\footnotesize{$3N$}"] (b4) [right=of a4] {};
  \node [draw,circle,fill=red!50,"\footnotesize{$N$}" right] (c2) [below right=of b4] {};
  \node (a1) [left=of b2] {};
   \node (a0) [left=of a1] {};
   \node (B0) [draw,circle,fill=green!50,"\footnotesize{$2N+3$}"] [left=of a0] {};
   \node (A1) [draw,circle,fill=green!50,"\footnotesize{$2N+2$}"][left=of B0] {};
   \node (B2) [draw,circle,fill=green!50,"\footnotesize{$2N+1$}"][left=of A1] {};
   \node (A3)[left=of B2] {};

  \node [draw,circle,fill=green!50,"\footnotesize{$N-3$}" below](a3) at (15.5,-2.5) {};
  \node (b1) [left=of a3] {};
  \node (b0) [left=of b1] {};
  \node (A0)[draw,circle,fill=green!50,"\footnotesize{$2$}" below] [left=of b0] {};
  \node (B1)[draw,circle,fill=green!50,"\footnotesize{$1$}" below] [left=of A0] {};
  \node (A2)[draw,circle,fill=green!50,"\footnotesize{$4N$}" below] [left=of B1] {};
  \node (B3) [left=of A2] {};
  \node [draw,circle,fill=red!50,"\footnotesize{$N-2$}" below] (b3) [right=of a3] {};
   \node [draw,circle,fill=red!50,"\footnotesize{$N-1$}" below] (a5) [right=of b3] {};

\draw[decorate sep={1mm}{5mm},fill] (6.5,-1.25) -- (7.5,-1.25);
\draw[decorate sep={1mm}{4mm},fill] (13.25,-1.25) -- (14.25,-1.25);

 \node [draw,circle,fill=blue!50,"\footnotesize{$4N+1$}" below] (4N+1) at (9,-4) {} ; 
 
 \node [draw,circle,fill=blue!50,"\footnotesize{$4N+2$}" below] (4N+2) at (11,-4) {} ;

  \begin{scope}[>=Latex]

  \draw[-> , thick]  (4N+1) edge (aa2);
   \draw[-> , thick]  (4N+2) edge (aa2); 
    \draw[-> , thick]  (c2) edge (4N+1);
   \draw[-> , thick]  (a5) edge (4N+2);

    \draw[-> , thick]  (a3) edge  (a1);
   \draw[-> , thick]   (a4) edge  (a3);
    \draw[-> , thick]  (a5) edge  (a4);

     \draw[-> , thick]   (b2) edge  (b1);
   \draw[-> , thick]  (b3) edge  (b2);
   \draw[-> , thick]  (b4) edge  (b3);
   
   \draw[-> , thick]   (b4) edge  (cc1);
     \draw[-> , thick]   (b4) edge  (aa2);
     \draw[-> , thick]   (bb1) edge[bend left=15]  (b4);

 \draw[-> , thick]   (a1) edge  (b2);
   \draw[-> , thick]  (b2) edge  (a4);

 \draw[-> , thick]   (a3) edge  (b3);
   \draw[-> , thick]  (b3) edge  (c2);
    \draw[-> , thick]    (b3) edge  (a5);

   \draw[-> , thick]   (a4) edge  (b4);

\draw[-> , thick]   (c2) edge  (a4);
\draw[-> , thick]   (c2) edge  (b4);
\draw[-> , thick]    (a5) edge  (b4);

\draw[-> , thick]     (a4) edge  (b4);

   \draw[-> , thick]   (a4) edge  (b4);

\draw[-> , thick]  (b1) edge  (a3); 

\draw[-> , thick]   (b0) edge  (B0); 
\draw[-> , thick]  (B0) edge  (B1);
\draw[-> , thick]   (B1) edge  (B2); 
\draw[-> , thick]  (B2) edge  (B3);

\draw[-> , thick]  (A0) edge  (b0); 
\draw[-> , thick]   (B0) edge  (a0);
\draw[-> , thick]    (A1) edge  (B0); 
\draw[-> , thick]    (B1) edge  (A0);
\draw[-> , thick]     (A2) edge  (B1); 
\draw[-> , thick]   (B2) edge  (A1);
\draw[-> , thick]   (A3) edge  (B2); 
\draw[-> , thick]  (B3) edge  (A2);

\draw[-> , thick]   (A2) edge  (A3); 
\draw[-> , thick]   (A1) edge  (A2); 
\draw[-> , thick]   (A0) edge  (A1); 
\draw[-> , thick]  (a0) edge  (A0);
    \end{scope}

\end{tikzpicture}
}
\end{center}
\caption{(potential) quiver $Q_{D_{2N}}$ emerged via Laurentification of the deformed type $D_{2N}$ map \eqref{D2Ndmap}. } \label{D2Nquiver}
\end{figure}

Let us denote the initial cluster variables, corresponding to the nodes of the quiver, by following tau-functions, 
\begin{equation}\label{D2Ninitialcluster}
\vb{\tx}_{2N} = (\tx_{1},\tx_{2}, \dots, \tx_{4N}) = (\xi_{N-2,0}, \dots, \xi_{1,0}, r_{0},s_{0}, \tau_{0},\tau_{1}, \dots, \tau_{2N},\rho_{0}, \eta_{1,0}, \dots , \eta_{N-2,0} )
\end{equation}
Recall the variable transformations for $x$ variables in type $D_{4}$ case are given by 
\begin{align*}
x_{1} = \frac{\tx_{7}\tx_{3} }{\tx_{6}\tx_{4}}, \quad x_{2}= \frac{\tx_{8}}{\tx_{6}\tx_{5}}, \quad x_{3} = \frac{\tx_{1}\sigma_{0}}{\tx_{6}}, \quad x_{4} = \frac{\tx_{2}}{\tx_{6}\sigma_{0}}
\end{align*}
We then shift the indices of $x_{j}$ for $j=3,4$ by 1 $(j \to j+1)$, while at the same time shifting the indices of the $\tx_{i}$ variables, according to the cluster extension described above. After this shift, we define new variables
\begin{align*}
x_{3} = \frac{\tx_{12}}{\tx_{9} \tx_{7}}, \quad x_{4} = \frac{\tx_{1}}{\tx_{9} \tx_{8}}
\end{align*}
By including this, we have following variable transformation,
\begin{align*}
x_{1} = \frac{\tx_{10}\tx_{4} }{\tx_{9}\tx_{5}}, \quad x_{2}= \frac{\tx_{11}}{\tx_{9}\tx_{6}},\quad x_{3}= \frac{\tx_{12}}{\tx_{9}\tx_{7}},\quad x_{4} = \frac{\tx_{1}}{\tx_{9}\tx_{8}}, \quad  x_{5} = \frac{\tx_{2}\sigma_{0}}{\tx_{9}}, \quad x_{6} = \frac{\tx_{3}}{\tx_{9}\sigma_{0}}
\end{align*}
which matches with \eqref{vartransformD6}. Then same as above, one can deduce the structure of the variable transformation (shown in \eqref{D2N: xvartrans}) that brings the deformed type $D_{2N}$ map \eqref{D2Ndmap}  into the form of cluster map defined by the cluster mutations in cluster algebra, equipped with the cluster $\vb{\tx}_{2N}$ and the quiver $Q_{D_{2N}}$. 
\begin{equation}\label{D2N: xvartrans}
\begin{split}
&x_{1} = \frac{\tx_{3N+1}\tx_{N+1}}{\tx_{3N} \tx_{N+2}}, \quad x_{2} = \frac{\tx_{3N+2}}{\tx_{3N} \tx_{N+3}} ,\quad \dots, \quad  x_{N} = \frac{\tx_{4N}}{\tx_{3N} \tx_{2N+1}},\\
& x_{N+1} = \frac{\tx_{1}}{\tx_{3N} \tx_{2N+2}}, \quad \dots, x_{2N-2} = \frac{\tx_{N-2}}{\tx_{3N} \tx_{3N-1}}, \quad x_{2N-1} = \frac{\tx_{N-1}\sigma_{0}}{\tx_{3N}}, \quad x_{2N} = \frac{\tx_{N}}{\tx_{3N}\sigma_{0}}  
\end{split} 
\end{equation}
As for the associated symplectic coordinates, we consider the symplectic structure in the cluster algebra of Dynkin type $D_{2N}$. Let  $e_{m} = (0, \dots, 0 ,\underbrace{1}_{\text{ m\ts{th} position} },0,\dots, 0)$ be standard basis vectors for $1 \leq m \leq 2N $. The exchange matrix corresponding to Dynkin type $D_{2N}$, written in the form, 
\begin{center}
\footnotesize
 \begin{tikzpicture}
 \node at (-3.5,0) {$B_{D_{2N}}=$};
        \matrix [inner sep=0pt, 
      nodes={inner sep=.3333em}, 
      matrix of math nodes,
      left delimiter=(,
      right delimiter=),
    ] (m)
        {
      0 & 1 &  &  &  &  & \\ 
        -1 & \phantom{0} & \phantom{1} &  &  \\
        &  \phantom{-1} & 0 & 1 &  &  &   \\ 
         &   &  -1 & 0 & 1 & 0  & 0 &    \\
          &  &  & -1& 0 & 1 &  1& \\
   &  &  & 0 & -1  & 0 & 0 &    \\ 
  &  &  & 0 & -1  & 0 & 0 &    \\ 
        };
        \draw[loosely dotted,thick] (m-2-1)-- (m-4-3);
\draw[loosely dotted,thick] (m-1-1)-- (m-3-3);
\draw[loosely dotted,thick] (m-1-2)-- (m-3-4);
 \end{tikzpicture}
\end{center}
has rank $2N-2 < 2N$, with kernel $\ker B_{D_{2N}}$ consisting of two linearly independent vectors, $\vb{u}_{j} = (u^{(j)}_{1}, u^{(j)}_{2}, \dots , u^{(j)}_{2N})$ for $ j= 1 ,  2$, namely, 
\begin{equation}\label{kernD2N}
\vb{u}_{1} =  e_{ 2N}+  \sum_ {k=1}^{N-1} e_{2k-1} , \quad \vb{u}_{2} =    \sum_ {k=1}^{N}e_{2k-1}.
\end{equation}
and image $\im B_{D_{2N}}$ spanned by $2N-2$ linearly independent vectors   $\vb{v}_{j} =  (v^{(j)}_{1}, v^{(j)}_{2}, \dots , v^{(j)}_{2N})$ of $\im B_{D_{2N}}$, for $1 \leq j \leq 2N-2$, specifically, 
\begin{equation}\label{ImD2N}
\vb{v}_{1}  =  e_{2}, \quad \vb{v}_{l}=  e_{ l+1} -  e_{l-1} \ (2\leq l \leq 2N-3), \quad \vb{v}_{2N-2}= e_{ 2N-1} + e_{ 2N} - e_{ 2N-3}.
\end{equation}
Then null distribution of associated presymplectic form, $\om$, 
\begin{align*}
\om = \sum^{2N}_{i<j}\frac{b_{ij}}{x_{i}x_{j}}\dd x_{i} \w \dd x_{j}
\end{align*}
 is generated by the two vector fields  $\vb{z}_{1} = x_{2N}\partial_{x_{2N}} + \sum_{k=1}^{N-1}x_{2k-1}\partial_{x_{2k-1}}$ and $\vb{z}_{2} = \sum_{k=1}^{N}x_{2k-1}\partial_{x_{2k-1}} $. Therefore  corresponding symplectic coordinates is given by
\begin{align*}
y_{1}&=x_{2},\\
y_{i} &= x_{i+1}/x_{i-1} \quad \text{for} \ 2 \leq i \leq 2N-3 \\ 
y_{2N-2} &= x_{2N-1}x_{2N}/x_{2N-3}.
\end{align*}
and they are invariant under the flow generated by the vector fields $\vb{z}_{1}$ and $\vb{z}_{2}$. By substituting the $x_{i}$ \eqref{D2N: xvartrans} into the reduced coordinates, we have 
\begin{equation}\label{y-vartransform}
\begin{split}
&y_{1} = \frac{\tx_{3N+2}}{\tx_{3N} \tx_{N+3}}, \quad y_{2} = \frac{\tx_{3N+3} \tx_{N+2}}{\tx_{N+4} \tx_{3N+1}\tx_{N+1}}, \quad y_{3} = \frac{\tx_{3N+4}\tx_{N+3}}{\tx_{3N+2} \tx_{N+5}}, \dots , y_{N-1} = \frac{\tx_{4N}\tx_{2N-1}}{ \tx_{4N-2} \tx_{2N+1}}, \\
&y_{N} = \frac{\tx_{1} \tx_{2N}}{\tx_{4N-1} \tx_{2N+2}},\quad  y_{N+1} = \frac{\tx_{2} \tx_{2N+1}}{\tx_{4N} \tx_{2N+3}}, \quad y_{N+2} = \frac{\tx_{3} \tx_{2N+2}}{\tx_{1} \tx_{2N+4}}, \dots, y_{2N-3} = \frac{\tx_{N-2} \tx_{3N-3}}{\tx_{N-4} \tx_{3N-1}}, \\ 
& y_{2N-2} = \frac{\tx_{N-1}\tx_{N} \tx_{3N-2}}{\tx_{3N}\tx_{N-3}}
\end{split}
\end{equation}
for $N \geq 4$. Through these variable transformation, one can verify that $Q_{D_{2N}}$ is indeed arises from the cluster algebra of type $D_{2N}$. Firstly we consider the $(2N-2) \times (2N-2)$  coefficient matrix $\tilde{B} = (\tilde{b}_{ij})$ of the symplectic form $\hat{\om}$,
\begin{align*}
\hat{\omega} = \sum\limits_{i<j}\frac{\hat{b}_{ij}}{y_{i}y_{j}}\dd y_{i} \wedge \dd y_{j}
\end{align*}
defined on the reduced space $\vb{y} = (y_{1}, \dots, y_{2N})$, that satisfies the pull-back relation $\pi^* \hat{\omega} = \omega$.   Following the statement in Theorem \ref{symreduction}, for $N\geq 4$, we define the $(2N-2) \times 4N$ matrix $M$, constructed from the vectors in $\im B_{D_{2N}}$  \eqref{ImD2N} and in $\ker B_{D_{2N}}$\eqref{kernD2N}, as shown below
\begin{equation}
M = \mqty(\vb{v}_{1} \\ \vb{v}_{2} \\ \vdots \\ \vb{v}_{N-2} \\ \vb{u}_{1} \\ \vb{u}_{2}). 
\end{equation}
The reduced matrix is then obtained via the relation \eqref{sympmatrix},
\begin{align*}
M^{-T} B_{D_{2N}}M^{-1} = \mqty(\hat{B} & 0 \\ 0& 0),
\end{align*}
yielding the $2N-2 \times 2N-2$ skew-symmetric matrix whose entries are given by 
\begin{equation}
    \hb_{i,j} = \begin{cases}
        1 & i<j, \ i= 2m-1, \ j=2n, \quad 0<n,m <N \\
        0 & j=2n+1, \ n\geq 1 \\
        -1 & j<i, \ j= 2m-1, \ i=2n, \quad 0<n,m <N
    \end{cases}
\end{equation}

Secondly, we write the variable transformations \eqref{y-vartransform} into the form $y_{j} = \tilde{\vb{x}}^{\vb{w}_{j}}$. Setting that the standard basis $\te_{i}$ for $1 \leq i \leq 4N$, the vector $\vb{w}_{j}$ can be written as 
\begin{equation}
\begin{split}
\vb{w}_{1} &= \te_{3N+2} - \te_{3N} - \te_{N+3}, \\ 
\vb{w}_{2} &= \te_{3N+3} + \te_{N+2} - \te_{3N+1} - \te_{N+4} - \te_{N+1}, \\ 
\vb{w}_{3} &= \te_{3N+4} + \te_{N+3} - \te_{3N+2} - \te_{N+5}, \\ 
\vdots \\ 
\vb{w}_{N-1} &= \te_{4N} + \te_{2N-1} - \te_{4N-2} - \te_{2N+1} \\
\vb{w}_{N} &= \te_{1} + \te_{2N} - \te_{4N-1} - \te_{2N+2} \\
\vb{w}_{N+1} &= \te_{2} + \te_{2N+1} -\te_{4N} - \te_{2N+3} \\
\vb{w}_{N+2} &= \te_{3} + \te_{2N+2} - \te_{1} - \te_{2N+4} \\
\vdots \\
\vb{w}_{2N-3} &=\te_{N-2} + \te_{3N-3} - \te_{N-4} - \te_{3N-1} \\
\vb{w}_{2N-2} &= \te_{N-1} + \te_{N} + \te_{3N-2} - \te_{3N} - \te_{N-3}
\end{split}
\end{equation}
Let $\tilde{M}$ be $(2N-2) \times 4N$ matrix given by 
\begin{equation}
\tilde{M} = \mqty(\vb{w}_{1} \\ \vb{w}_{2} \\ \vdots \\ \vb{w}_{2N-2})
\end{equation} 
The pull-back of the symplectic form 
\begin{equation}
\hat{\pi}^{*}\hat{\om} =\sum_{i<j}^{4N} \tilde{b}_{ij} \dl \tx_{i} \w \dl \tx_{j} \iff \tilde{B} = \tilde{M}^{T}\hat{B} \tilde{M}
\end{equation}
gives rise to the $4N \times 4N$ skew-symmetric matrix $\tilde{B}$. Furthermore, appending  the row corresponding frozen nodes at the bottom of the matrix $\tilde{B}$ yields the  extended exchange matrix 
\begin{equation}
B_{D_{2N}} = 
{\tiny
\vcenter{\hbox{ \begin{tikzpicture}
        \matrix [inner sep=0pt, 
      nodes={inner sep=.3333em}, 
      matrix of math nodes,
      left delimiter=(,
      right delimiter=),
    ] (m)
        {
      0 & 1 &  &  &  &  &  &  &  &  & 0 & 1 & 0 & -1 &  &  &  & &   &  &  & -1 \\ 
        -1 & \phantom{0} & \phantom{1} &  &  &  &  &  &  &  &  & \phantom{1} & \phantom{0} & \phantom{-1} &  &  &  &  &  &  &  & \\
        &  \phantom{-1} & 0 & 1 &  &  &  &  &  &  &  &  & \phantom{0} & 1 & 0 & -1 &  &  &  &  &  &  \\ 
          &   &  -1 & 0 & 1 & 1  & 0 & 0 &0  &  &  &  &  & 0 & 1 & 0 & -1& 0 &0  &  &  & \\
          &  &  & -1 & 0 & 0 & 0 & 0 & 0 &  &  &  &  & 0 & 0 & 1 & 1 & 0 &0  & &  & \\
   &  &  & -1 & 0  & 0 & 0 & 0 & 0  &  &  &  &  & 0 & 0 & 1 & 1 & 0 &  0&  &  &  \\ 
  &  &  & 0 & 0  & 0 & 0 & 0 & -1  &  &  &  &  & 0 & 0 & 0 & -1 & 0 & 1  & &  & \\ 
 &  &  & 0 & 0 & 0 & 0 & 0 & 1 &  &  &  &  & 0 & 0 & 0 & 1 & 0 & -1 &  & \phantom{0} & \phantom{0} \\ 
  &  &  &0  &  0 & 0 & 1 & -1 & 0 & 1 &  &  &  & 0 & 0 & 0 &0 & 1 & 0 & -1  & \phantom{0} & \phantom{0} \\ 
  &  &  & & &  &  &  & -1  & \phantom{0} & \phantom{1} &  &  &  &  &  & \phantom{0} & \phantom{0} & \phantom{1} & \phantom{0} & \phantom{-1} & \phantom{0} \\ 
 0 &  &  &  &  &  &  &  &  & \phantom{-1} & \phantom{0} & \phantom{1} &  &  &  &  &  & \phantom{0} & \phantom{0} & \phantom{1} & 0 & -1 \\ 
 -1 & \phantom{-1} &   &  &  &  &  &  &  &  &  \phantom{-1} & \phantom{0}  & \phantom{1} &  &  &  &  & \phantom{0} & \phantom{0} & \phantom{0} & 1 & 0\\
 0 & \phantom{0} & \phantom{-1} &  &   &  &  &  &  &  &   & \phantom{-1} & \phantom{0} & 1 &  &  &  &  & \phantom{0} & \phantom{0} & \phantom{0} & 1\\ 
 1 & \phantom{1} & -1 & 0 &  &  &  &  &  &  &  &  & -1 & 0  & 1 &  &  &  & &  &  &  \\
   &  & 0 & -1 & 0 & 0 & 0 & 0 & 0  &  &  &  &  & -1 & 0 & 1 & 0 & 0 &  &  &  &  \\
   &  &  1& 0 & -1 & -1 & 0 & 0 & 0 &  &  &  &  &  & -1 & 0 & 1 & 0 &  &  &   &  \\
   &  &  &1 & -1 & -1 & 1 & -1 &  0&  &  & &  &  & 0 & -1 & 0 & 1 &  &  &  &  \\
    &  &  & 0 & 0 & 0 & 0 & 0 & -1 &  &  &  &  &  & 0 & 0 & -1 & 0 &  &  &  & \\
    &  &  & 0 & 0 & 0 & -1 & 1 & 0 & \phantom{-1} &  &  &  &  &  &  &  &  &  &  &  &  \\
    &  &  &  &  &  &  &  & \phantom{1} & \phantom{0} & \phantom{-1} &  &  &  &  &  &  &  &  &  & 0 & 1 \\
    1 &  &  &  &  &  &  &  &  & 1 &0 & -1 &  &  &  &  &  &  &  &  &-1  & 0 \\
     &  &  & 0 & 0 & -1 & 1 & 0 & &  & &  &  &  & 0 & 0 & 0 & 0 &  &  &  & \\
      &  &  & 0 & -1 & 0 & 1 & 0 & &  & &  &  &  & 0 & 0 & 0 & 0 &  &  &  & \\
        };

 \draw[loosely dotted,thick] (m-2-1)-- (m-4-3);
\draw[loosely dotted,thick] (m-1-1)-- (m-3-3);
\draw[loosely dotted,thick] (m-1-2)-- (m-3-4);

\draw[loosely dotted,thick] (m-9-10)-- (m-13-14);
\draw[loosely dotted,thick] (m-9-9)-- (m-14-14);
\draw[loosely dotted,thick] (m-10-9)-- (m-14-13);

\draw[loosely dotted,thick] (m-11-1)-- (m-14-4);
\draw[loosely dotted,thick] (m-12-1)-- (m-14-3);
\draw[loosely dotted,thick] (m-13-1)-- (m-15-3);
\draw[loosely dotted,thick] (m-14-1)-- (m-16-3);

\draw[loosely dotted,thick] (m-18-9)-- (m-21-12);
\draw[loosely dotted,thick] (m-19-9)-- (m-21-11);
\draw[loosely dotted,thick] (m-19-8)-- (m-21-10);

\draw[loosely dotted,thick] (m-1-12)-- (m-3-14);
\draw[loosely dotted,thick] (m-1-13)-- (m-3-15);
\draw[loosely dotted,thick] (m-1-14)-- (m-3-16);

\draw[loosely dotted,thick] (m-8-19)-- (m-11-22);
\draw[loosely dotted,thick] (m-9-19)-- (m-11-21);
\draw[loosely dotted,thick] (m-9-18)-- (m-12-21);

\draw[loosely dotted,thick] (m-17-18)-- (m-20-22);
\draw[loosely dotted,thick] (m-18-18)-- (m-20-21);
\draw[loosely dotted,thick] (m-18-17)-- (m-21-21);

    \end{tikzpicture}}}}
\end{equation}

which represents the quiver $Q_{D_{2N}}$.

\newpage

\section{Type $D_{2N}$ deformed cluster map} \label{S: D2Nmap}

In the previous sections, we demonstrated that the deformed maps of type $D_{4}$ and $D_{6}$ can be described as the cluster map on the new cluster algebra of higher rank case. Here we verify that the quiver $Q_{D_{2N}}$ (illustrated in Figure \ref{D2Nquiver}) is mutation periodic under specific sequence of mutations, which in turn gives rise to the cluster maps associated with deformed type $D_{2N}$ maps via Laurentification.

\begin{prop} For each $N \geq 2$, the quiver $Q_{D_{2N}}$ is mutation periodic,
\begin{equation}
\mu_{N}\mu_{N-1} \cdots \mu_{1}\mu_{4N}\mu_{4N-1} \cdots \mu_{3N+2}\mu_{N+1} (Q) = \rho(Q),
\end{equation}
\end{prop}
up to permutation $\rho= (N-1,N)(N+1,N+2,\dots, 3N,3N+1)$.

\begin{proof}

Applying the mutation $\mu_{N+1}$ followed by the mutation $\mu_{3N+2}$ at the nodes of the quiver $Q_{D_{2N}}$ gives 

\begin{center}
\resizebox{0.85\textwidth}{!}{%
 \begin{tikzpicture}[every circle node/.style={draw,scale=0.6,thick},node distance=10mm]
  \node [draw,circle,fill=red!50,"{\footnotesize{$N+2$}}"] (N+2) at (0,0) {};
  \node [draw,circle,fill=red!50,"{\footnotesize{$N+3$}}"] (N+3) [right=of N+2] {};
  \node [draw,circle,fill=red!50,"{\footnotesize{$N+4$}}"] (N+4) [right=of N+3] {};
  \node [draw,circle,fill=red!50,"{\footnotesize{$N+5$}}"] (N+5) [right=of N+4] {};
  \node [draw,circle,fill=red!50,"{\footnotesize{$N+6$}}"] (N+6) [right=of N+5] {};
  \node [draw,circle,fill=red!50,"{\footnotesize{$3N+1$}}" left]  (3N+1) [below left=of N+2] {};
     \node (a6)[right=of N+6]{};

  \node [draw,circle,fill=red!50,"{\footnotesize{$N+1$}}" below] (N+1) at (0,-2.5) {} ;
     \node [draw,circle,fill=blue!50,"{\footnotesize{$4N+1$}}" left] (4N+1) [above left=of N+1] {} ;
  \node [draw,circle,fill=red!50,"{\footnotesize{$3N+2$}}" below] (3N+2) [right=of N+1] {} ; 
  \node [draw,circle,fill=red!50,"{\footnotesize{$3N+3$}}" below] (3N+3) [right=of 3N+2] {};
  \node [draw,circle,fill=red!50,"{\footnotesize{$3N+4$}}" below] (3N+4) [right=of 3N+3] {};
   \node [draw,circle,fill=red!50,"{\footnotesize{$3N+5$}}" below] (3N+5) [right=of 3N+4] {};
   \node (b5)[right=of 3N+4]{};
    \node (b6) at (5.5,-3) {};
    \node (b7) at (5.5,-3.2) {};
    \node (b8) at (5.5,-3.6) {};
    \node (b9) at (5.5,-3.8) {};
    
    \node (h1) at (5.4,-1) {};
        \node (h2) at (5.5, -5) {};

    \node [draw,circle,fill=red!50,"{\footnotesize{$3N$}}" below] (3N) at (2.5,-4.2) {} ;
    
    \node [draw,circle,fill=blue!50,"{\footnotesize{$4N+2$}}" below] (4N+2) at (2.5,-5.2) {} ;

\draw [-{Latex[length=3mm]}] (-3.5,-1.2) -- (-2.5,-1.2) node[midway,sloped,above] {$\mu_{N+1}$};

  \begin{scope}[>=Latex]
  
  \draw[-> , thick] (3N+2) edge (N+1);
   \draw[-> , thick] (N+1) edge (3N);
     \draw[-> , thick] (N+1) edge (4N+2);
       \draw[-> , thick] (N+1) edge (N+3);
        \draw[-> , thick] (N+1) edge (4N+1);
        
      \draw[-> , thick] (3N+2) edge (N+2);
      \draw[-> , thick] (3N+2) edge (3N+3);
      \draw[-> , thick] (4N+1) edge (3N+2);
       \draw[-> , thick] (3N+1) edge (3N+2);
        \draw[-> , thick] (N+3) edge (3N+2);
         \draw[-> , thick] (N+4) edge (3N+2);
          \draw[-> , thick] (3N) edge (3N+2);
           \draw[-> , thick] (4N+2) edge (3N+2);
           
            \draw[-> , thick] (h1) edge (4N+1);
            
            \draw[-> , thick] (N+3) edge (3N+1); 
            \draw[-> , thick] (3N) edge (3N+1);    
            
            \draw[-> , thick] (N+2) edge (3N); 
             \draw[-> , thick] (N+2) edge (N+3); 
            
             \draw[-> , thick] (3N+3) edge (N+3); 
             \draw[-> , thick] (N+3) edge (N+4); 
             
             \draw[-> , thick] (N+4) edge (N+5); 
             \draw[-> , thick] (3N+4) edge (N+4); 
             
             \draw[-> , thick] (N+5) edge (3N+3); 
              \draw[-> , thick] (3N+5) edge (N+5); 
              
               \draw[-> , thick] (3N+3) edge (3N+4); 
              \draw[-> , thick] (3N+4) edge (3N+5); 
              
               \draw[-> , thick] (N+5) edge (N+6); 
                \draw[-> , thick] (N+6) edge (3N+4); 
                
                 \draw[-> , thick] (b6) edge (3N); 
                  \draw[-> , thick] (b7) edge (3N); 
                   \draw[- , thick] (3N) edge (b8); 
                    \draw[-> , thick] (b9) edge (3N); 
             \draw[-> , thick] (h2) edge (4N+2);

    \end{scope}
\draw [-{Latex[length=3mm]}] (6,-1.2) -- (7,-1.2) node[midway,sloped,above] {$\mu_{3N+2}$};

\node [draw,circle,fill=red!50,"{\footnotesize{$N+2$}}"] (N+2) at (10,0) {};
  \node [draw,circle,fill=red!50,"{\footnotesize{$N+3$}}"] (N+3) [right=of N+2] {};
  \node [draw,circle,fill=red!50,"{\footnotesize{$N+4$}}"] (N+4) [right=of N+3] {};
  \node [draw,circle,fill=red!50,"{\footnotesize{$N+5$}}"] (N+5) [right=of N+4] {};
  \node [draw,circle,fill=red!50,"{\footnotesize{$N+6$}}"] (N+6) [right=of N+5] {};
  \node [draw,circle,fill=red!50,"{\footnotesize{$3N+1$}}" left]  (3N+1) [below left=of N+2] {};
     \node (a6)[right=of N+6]{};

  \node [draw,circle,fill=red!50,"{\footnotesize{$N+1$}}" below] (N+1) at (10,-2.5) {} ;
     \node [draw,circle,fill=blue!50,"{\footnotesize{$4N+1$}}" left] (4N+1) [above left=of N+1] {} ;
  \node [draw,circle,fill=red!50,"{\footnotesize{$3N+2$}}" below] (3N+2) [right=of N+1] {} ; 
  \node [draw,circle,fill=red!50,"{\footnotesize{$3N+3$}}" below] (3N+3) [right=of 3N+2] {};
  \node [draw,circle,fill=red!50,"{\footnotesize{$3N+4$}}" below] (3N+4) [right=of 3N+3] {};
   \node [draw,circle,fill=red!50,"{\footnotesize{$3N+5$}}" below] (3N+5) [right=of 3N+4] {};
   \node (b5)[right=of 3N+4]{};
    \node (b6) at (15.5,-3) {};
    \node (b7) at (15.5,-3.2) {};
    \node (b8) at (15.5,-3.6) {};
    \node (b9) at (15.5,-3.8) {};
    
    \node (h1) at (15.4,-1) {};
         \node (h2) at (15.5, -5) {};

  \node [draw,circle,fill=red!50,"{\footnotesize{$3N$}}" below] (3N) at (12.5,-4.2) {} ;
    
    \node [draw,circle,fill=blue!50,"{\footnotesize{$4N+2$}}" below] (4N+2) at (11.5,-5.2) {} ;

  \begin{scope}[>=Latex]
  \draw[-> , thick] (N+1) edge (3N+2);

      \draw[-> , thick] (N+2) edge (3N+2);
      \draw[-> , thick] (3N+3) edge (3N+2);
      \draw[-> , thick] (3N+2) edge (4N+1);
       \draw[-> , thick] (3N+2) edge (3N+1);
        \draw[-> , thick] (3N+2) edge (N+3);
         \draw[-> , thick] (3N+2) edge (N+4);
          \draw[-> , thick] (3N+2) edge (3N);
           \draw[-> , thick] (3N+2) edge (4N+2);
           
            \draw[-> , thick] (h1) edge (4N+1);
            
            \draw[-> , thick] (N+3) edge (3N+1); 
            \draw[-> , thick] (3N) edge (3N+1);

             \draw[-> , thick] (N+3) edge (N+4); 
             
             \draw[-> , thick] (N+4) edge (N+5); 
             \draw[-> , thick] (3N+4) edge (N+4); 
             
             \draw[-> , thick] (N+5) edge (3N+3); 
              \draw[-> , thick] (3N+5) edge (N+5); 
              
               \draw[-> , thick] (3N+3) edge (3N+4); 
              \draw[-> , thick] (3N+4) edge (3N+5); 
              
               \draw[-> , thick] (N+5) edge (N+6); 
                \draw[-> , thick] (N+6) edge (3N+4); 
                
                 \draw[-> , thick] (b6) edge (3N); 
                  \draw[-> , thick] (b7) edge (3N); 
                   \draw[- , thick] (3N) edge (b8); 
                    \draw[-> , thick] (b9) edge (3N); 
                    
              \draw[-> , thick] (h2) edge (4N+2);

 \draw[-> , thick] (4N+2) edge(N+2); 
  \draw[-> , thick] (4N+1) edge (N+2); 
  \draw[-> , thick] (4N+1) edge (3N+3); 
    \draw[-> , thick] (3N+1) edge (N+1); 
   \draw[-> , thick] (N+4) edge (N+1); 
    \draw[-> , thick] (3N+3) edge (3N+2); 
     \draw[-> , thick] (3N+3) edge (3N+4); 
     \draw[-> , thick] (3N+1) edge (3N+3); 
      \draw[-> , thick] (N+4) edge (3N+3); 
       \draw[-> , thick] (3N) edge (3N+3); 
        \draw[-> , thick] (4N+2) edge (3N+3); 
        \draw[-> , thick] (3N+1) edge (N+2); 
         \draw[-> , thick] (N+4) edge[bend right=15] (N+2);

    \end{scope}

\end{tikzpicture}
}
\end{center}

Let us consider the subquiver in mutated quiver $\mu_{3N+2}\mu_{1}(Q_{D_{2N}})$, consisting of next mutating node $3N+3$ along with all incident nodes that are connected to it. Upon applying mutation $\mu_{3N+3}$, we see that the node $3N+4$ is connected in a similar way as previous mutating node. This resembles a rightward shift of the mutating node within the quiver as shown below.
\begin{center}
\resizebox{0.85\textwidth}{!}{%
 \begin{tikzpicture}[every circle node/.style={draw,scale=0.6,thick},node distance=10mm]
  \node [draw,circle,fill=red!50,"{\footnotesize{$N+2$}}"] (N+2) at (0,0) {};
  \node [draw,circle,fill=red!50,"{\footnotesize{$N+3$}}"] (N+3) [right=of N+2] {};
  \node [draw,circle,fill=red!50,"{\footnotesize{$N+4$}}"] (N+4) [right=of N+3] {};
  \node [draw,circle,fill=red!50,"{\footnotesize{$N+5$}}"] (N+5) [right=of N+4] {};
  \node [draw,circle,fill=red!50,"{\footnotesize{$N+6$}}"] (N+6) [right=of N+5] {};
  \node [draw,circle,fill=red!50,"{\footnotesize{$3N+1$}}" left]  (3N+1) [below left=of N+2] {};
     \node (a6)[right=of N+6]{};

  \node [draw,circle,fill=red!50,"{\footnotesize{$N+1$}}" below] (N+1) at (0,-2.5) {} ;
     \node [draw,circle,fill=blue!50,"{\footnotesize{$4N+1$}}" left] (4N+1) [above left=of N+1] {} ;
  \node [draw,circle,fill=red!50,"{\footnotesize{$3N+2$}}" below] (3N+2) [right=of N+1] {} ; 
  \node [draw,circle,fill=red!50,"{\footnotesize{$3N+3$}}" below] (3N+3) [right=of 3N+2] {};
  \node [draw,circle,fill=red!50,"{\footnotesize{$3N+4$}}" below] (3N+4) [right=of 3N+3] {};
   \node [draw,circle,fill=red!50,"{\footnotesize{$3N+5$}}" below] (3N+5) [right=of 3N+4] {};
   \node (b5)[right=of 3N+4]{};
    \node (b6) at (5.5,-3) {};
    \node (b7) at (5.5,-3.2) {};
    \node (b8) at (5.5,-3.6) {};
    \node (b9) at (5.5,-3.8) {};
    
    \node (h1) at (5.4,-1) {};
        \node (h2) at (5.5, -5) {};

    \node [draw,circle,fill=red!50,"{\footnotesize{$3N$}}" below] (3N) at (2.5,-4.2) {} ;
    
    \   \node [draw,circle,fill=blue!50,"{\footnotesize{$4N+2$}}" below] (4N+2)[left=of 3N] {} ;

\draw [-{Latex[length=3mm]}] (-3.5,-1.2) -- (-2.5,-1.2) node[midway,sloped,above] {$\mu_{3N+2}$};

  \begin{scope}[>=Latex]
   \draw[-> , thick] (4N+1) edge (3N+3); 
   \draw[-> , thick] (3N+2) edge (4N+1); 
   \draw[-> , thick] (3N+2) edge (3N+1); 
  
   \draw[-> , thick] (N+4) edge (N+1); 
   \draw[-> , thick] (N+1) edge (3N+2);
    \draw[-> , thick] (3N+3) edge (3N+2); 
     \draw[-> , thick] (3N+3) edge (3N+4); 
     \draw[-> , thick] (3N+1) edge (3N+3); 
      \draw[-> , thick] (N+4) edge (3N+3); 
       \draw[-> , thick] (3N) edge (3N+3); 
        \draw[-> , thick] (4N+2) edge (3N+3); 
        
         \draw[-> , thick] (3N+2) edge (N+4); 
           \draw[-> , thick] (3N+4) edge (N+4);
             \draw[-> , thick] (3N+2) edge (4N+2);
  
           \draw[-> , thick] (3N+2) edge (3N);
           
           \draw[-> , thick] (N+4) edge (N+5);
           \draw[-> , thick] (N+5) edge (N+6);
           \draw[-> , thick] (N+5) edge (3N+3);
           \draw[-> , thick] (N+6) edge (3N+4);
           \draw[-> , thick] (3N+4) edge (3N+5);
             \draw[-> , thick] (3N+5) edge (N+5);

    \end{scope}
\draw [-{Latex[length=3mm]}] (6,-1.2) -- (7,-1.2) node[midway,sloped,above] {$\mu_{3N+3}$};


\node [draw,circle,fill=red!50,"{\footnotesize{$N+2$}}"] (N+2) at (10,0) {};
  \node [draw,circle,fill=red!50,"{\footnotesize{$N+3$}}"] (N+3) [right=of N+2] {};
  \node [draw,circle,fill=red!50,"{\footnotesize{$N+4$}}"] (N+4) [right=of N+3] {};
  \node [draw,circle,fill=red!50,"{\footnotesize{$N+5$}}"] (N+5) [right=of N+4] {};
  \node [draw,circle,fill=red!50,"{\footnotesize{$N+6$}}"] (N+6) [right=of N+5] {};
  \node [draw,circle,fill=red!50,"{\footnotesize{$3N+1$}}" left]  (3N+1) [below left=of N+2] {};
     \node (a6)[right=of N+6]{};

  \node [draw,circle,fill=red!50,"{\footnotesize{$N+1$}}" below] (N+1) at (10,-2.5) {} ;
     \node [draw,circle,fill=blue!50,"{\footnotesize{$4N+1$}}" left] (4N+1) [above left=of N+1] {} ;
  \node [draw,circle,fill=red!50,"{\footnotesize{$3N+2$}}" below] (3N+2) [right=of N+1] {} ; 
  \node [draw,circle,fill=red!50,"{\footnotesize{$3N+3$}}" below] (3N+3) [right=of 3N+2] {};
  \node [draw,circle,fill=red!50,"{\footnotesize{$3N+4$}}" below] (3N+4) [right=of 3N+3] {};
   \node [draw,circle,fill=red!50,"{\footnotesize{$3N+5$}}" below] (3N+5) [right=of 3N+4] {};
   \node (b5)[right=of 3N+4]{};
    \node (b6) at (15.5,-3) {};
    \node (b7) at (15.5,-3.2) {};
    \node (b8) at (15.5,-3.6) {};
    \node (b9) at (15.5,-3.8) {};
    
    \node (h1) at (15.4,-1) {};
         \node (h2) at (15.5, -5) {};

  \node [draw,circle,fill=red!50,"{\footnotesize{$3N$}}" below] (3N) at (12.5,-4.2) {} ;
    
    \node [draw,circle,fill=blue!50,"{\footnotesize{$4N+2$}}" below] (4N+2) [left=of 3N] {} ;

  \begin{scope}[>=Latex]

  \draw[-> , thick] (3N+3) edge (4N+1);

   \draw[-> , thick] (N+4) edge (N+1); 
   \draw[-> , thick] (N+1) edge (3N+2);
    \draw[-> , thick] (3N+2) edge (3N+3); 
     \draw[-> , thick] (3N+4) edge (3N+3); 
     \draw[-> , thick] (3N+3) edge (3N+1); 
      \draw[-> , thick] (3N+3) edge (N+4); 
       \draw[-> , thick] (3N+3) edge (3N); 
        \draw[-> , thick] (3N+3) edge (4N+2);

           \draw[-> , thick] (N+4) edge (N+5);
           \draw[-> , thick] (N+5) edge (N+6);
           \draw[-> , thick] (3N+3) edge (N+5);
           \draw[-> , thick] (N+6) edge (3N+4);
           \draw[-> , thick] (3N+4) edge (3N+5);
             \draw[-> , thick] (3N+5) edge (N+5);
     
      \draw[-> , thick,red] (3N) edge (3N+4);
       \draw[-> , thick, red ] (4N+1) edge (3N+4);
        \draw[-> , thick, red] (3N+1) edge (3N+4);
       \draw[-> , thick, red] (N+5) edge (3N+2);
    \draw[-> , thick,red] (N+5) edge (3N+4);
      \draw[-> , thick,red] (4N+2) edge (3N+4);
    \end{scope}

\end{tikzpicture}
}
\end{center}
Since the "ladder" shaped structural pattern of the quiver is extended toward right end of the quiver, this transformation is repeated throughout the quiver. Thus by applying the sequence of mutations $\mu_{N}\mu_{N-1} \cdots \mu_{1}\mu_{4N}\mu_{4N-1} \cdots \mu_{3N+4} $ sequentially, we see the following transformations within quiver.

\begin{center}
\resizebox{1.1\textwidth}{!}{%
 \begin{tikzpicture}[every circle node/.style={draw,scale=0.6,thick},node distance=10mm]
 
 
  \node [draw,circle,fill=red!50,"{\footnotesize{$2N-1$}}"] (2N-1) at (0,0) {};
  \node [draw,circle,fill=red!50,"{\footnotesize{$2N$}}"] (2N) [right=of 2N-1] {};
  \node [draw,circle,fill=red!50,"{\footnotesize{$2N+1$}}"] (2N+1) [right=of 2N] {};
  \node [draw,circle,fill=red!50,"{\footnotesize{$2N+2$}}"] (2N+2) [right=of 2N+1] {};
  \node [draw,circle,fill=red!50,"{\footnotesize{$2N+3$}}"] (2N+3) [right=of 2N+2] {};
  \node (a1)  [below left=of 2N-1] {};
  \node [draw,circle,fill=red!50,"{\footnotesize{$3N+1$}}"  above]  (3N+1) [left=of a1] {};
     \node (a6)[right=of 2N+3]{};

  \node [draw,circle,fill=red!50,"{\footnotesize{$4N-2$}}" below] (4N-2) at (0,-2.5) {} ;
  \node (b1)  [above left=of 4N-2] {} ;
     \node [draw,circle,fill=blue!50,"{\footnotesize{$4N+1$}}" below] (4N+1) [ left=of b1] {} ;
  \node [draw,circle,fill=red!50,"{\footnotesize{$4N-1$}}" below] (4N-1) [right=of 4N-2] {} ; 
  \node [draw,circle,fill=red!50,"{\footnotesize{$4N$}}" below] (4N) [right=of 4N-1] {};
  \node [draw,circle,fill=red!50,"{\footnotesize{$1$}}" below] (1) [right=of 4N] {};
   \node [draw,circle,fill=red!50,"{\footnotesize{$2$}}" below] (2) [right=of 1] {};
   
   \node (r1)[left=of 4N-2]{};
   
    \node(l1) at (5.7,0){};
     \node(l2) at (5.7,-2.5){};
   
     \node(c1) at (-0.7,0){};
     \node(c2) at (-0.7,-2.5){};
  
    \node (b6) at (5.5,-3) {};
    \node (b7) at (5.5,-3.2) {};
     \node (m1) at (5.5,-3.4) {};
    \node (b8) at (5.5,-3.6) {};
    \node (b9) at (5.5,-3.8) {};
    
    \node (h1) at (5.4,-1) {};
        \node (h2) at (5.5, -5) {};
    
     \node(r2) at (-2,0) {};

    \node [draw,circle,fill=red!50,"{\footnotesize{$3N$}}" below] (3N) at (2.5,-4.2) {} ;
    
    \   \node [draw,circle,fill=blue!50,"{\footnotesize{$4N+2$}}" below] (4N+2)[left=of 3N] {} ;

\draw [-{Latex[length=3mm]}] (-4,-1.2) -- (-3,-1.2) node[midway,sloped,above] {$\mu_{4N-2}$};

  \begin{scope}[>=Latex]
  
  \draw[- , thick] (4N+2) edge (r2);
  
   \draw[- , thick] (2N-1) edge (c1);
   \draw[- , thick] (4N-2) edge (c2);
  
    \draw[- , thick] (2N+3) edge (l1);
        \draw[- , thick] (2) edge (l2);
  
           \draw[-> , thick] (2N) edge (4N-1);
           \draw[->, thick] (4N) edge (2N);
             \draw[-> , thick] (2N+1) edge (4N-1);
               \draw[-> , thick] (2N) edge (2N+1);
                 \draw[-> , thick] (2N+1) edge (2N+2);
                 \draw[-> , thick] (2N+2) edge (2N+3);
                   \draw[-> , thick] (4N-1) edge (4N);
                     \draw[-> , thick] (4N) edge (1);
                       \draw[-> , thick] (1) edge (2);
                         \draw[-> , thick] (2N+2) edge (4N);
                         \draw[-> , thick] (2N+3) edge (1);
                          \draw[-> , thick] (1) edge (2N+1);
                            \draw[-> , thick] (2) edge (2N+2);
                            \draw[-> , thick] (2N) edge (r1);
                            \draw[-> , thick] (4N-2) edge (2N);
              \draw[-> , thick] (4N-1) edge (4N-2);
            \draw[-> , thick] (2N-1) edge (2N);
            \draw[-> , thick] (3N+1) edge (4N-1);
            \draw[-> , thick] (4N-2) edge (3N+1);
            \draw[-> , thick] (4N+1) edge (4N-1);
            \draw[-> , thick] (4N-2) edge (4N+1);
            \draw[-> , thick] (3N+1) edge (4N-1);
            \draw[-> , thick] (4N-2) edge (4N+2);
            
            \draw[-> , thick] (3N) edge (4N-1);
            \draw[-> , thick] (3N+1) edge (4N-1);
            \draw[-> , thick] (4N+2) edge (4N-1);
            \draw[-> , thick] (4N-2) edge (3N);
            \draw[-> , thick] (3N) edge[bend right=30] (3N+1);
            
            \draw[-> , thick] (m1) edge (4N+2);
             \draw[-> , thick] (b6) edge (3N);
            \draw[-> , thick] (b7) edge (3N);
            \draw[- , thick] (3N) edge (b8);
            \draw[-> , thick] (b9) edge (3N);

    \end{scope}
\draw [-{Latex[length=3mm]}] (6,-1.2) -- (7,-1.2) node[midway,sloped,above] {$\mu_{4N-1}$};


 \node [draw,circle,fill=red!50,"{\footnotesize{$2N-1$}}"] (2N-1) at (10,0) {};
  \node [draw,circle,fill=red!50,"{\footnotesize{$2N$}}"] (2N) [right=of 2N-1] {};
  \node [draw,circle,fill=red!50,"{\footnotesize{$2N+1$}}"] (2N+1) [right=of 2N] {};
  \node [draw,circle,fill=red!50,"{\footnotesize{$2N+2$}}"] (2N+2) [right=of 2N+1] {};
  \node [draw,circle,fill=red!50,"{\footnotesize{$2N+3$}}"] (2N+3) [right=of 2N+2] {};
  \node (a1) [below left=of 2N-1] {};
  \node [draw,circle,fill=red!50,"{\footnotesize{$3N+1$}}" above]  (3N+1) [left=of a1] {};
     \node (a6)[right=of 2N+3]{};

  \node [draw,circle,fill=red!50,"{\footnotesize{$4N-2$}}" below] (4N-2) at (10,-2.5) {} ;
  \node (b1) [above left=of 4N-2] {} ;
     \node [draw,circle,fill=blue!50,"{\footnotesize{$4N+1$}}" below] (4N+1) [left=of b1] {} ;
  \node [draw,circle,fill=red!50,"{\footnotesize{$4N-1$}}" below] (4N-1) [right=of 4N-2] {} ; 
  \node [draw,circle,fill=red!50,"{\footnotesize{$4N$}}" below] (4N) [right=of 4N-1] {};
  \node [draw,circle,fill=red!50,"{\footnotesize{$1$}}" below] (1) [right=of 4N] {};
   \node [draw,circle,fill=red!50,"{\footnotesize{$2$}}" below] (2) [right=of 1] {};
   
   \node (r1)[left=of 4N-2]{};
   
   \node(l1) at (15.7,0){};
     \node(l2) at (15.7,-2.5){};
   
     \node(c1) at (9.3,0){};
     \node(c2) at (9.3,-2.5){};
  
    \node (b6) at (15.5,-3) {};
    \node (b7) at (15.5,-3.2) {};
     \node (m1) at (15.5,-3.4) {};
    \node (b8) at (15.5,-3.6) {};
    \node (b9) at (15.5,-3.8) {};
    
    \node (h1) at (15.4,-1) {};
        \node (h2) at (15.5, -5) {};
    
    \node(r2) at (8,0) {};

    \node [draw,circle,fill=red!50,"{\footnotesize{$3N$}}" below] (3N) at (12.5,-4.2) {} ;
    
    \   \node [draw,circle,fill=blue!50,"{\footnotesize{$4N+2$}}" below] (4N+2)[left=of 3N] {} ;

  \begin{scope}[>=Latex]

\draw[- , thick] (4N+2) edge (r2);

  \draw[-> , thick] (4N-1) edge (2N);
  
  \draw[- , thick] (2N-1) edge (c1);
   \draw[- , thick] (4N-2) edge (c2);
  
    \draw[- , thick] (2N+3) edge (l1);
        \draw[- , thick] (2) edge (l2);

             \draw[-> , thick] (4N-1) edge (2N+1);
               \draw[-> , thick] (2N) edge (2N+1);
                 \draw[-> , thick] (2N+1) edge (2N+2);
                 \draw[-> , thick] (2N+2) edge (2N+3);
                   \draw[-> , thick] (4N) edge (4N-1);
                     \draw[-> , thick] (4N) edge (1);
                       \draw[-> , thick] (1) edge (2);
                         \draw[-> , thick] (2N+2) edge (4N);
                         \draw[-> , thick] (2N+3) edge (1);
                          \draw[-> , thick] (1) edge (2N+1);
                            \draw[-> , thick] (2) edge (2N+2);
                            \draw[-> , thick] (2N) edge (r1);
                           
              \draw[-> , thick] (4N-2) edge (4N-1);
            \draw[-> , thick] (2N-1) edge (2N);
            \draw[-> , thick] (4N-1) edge (3N+1);
            \draw[-> , thick] (3N+1) edge (4N);
            \draw[-> , thick] (4N-1) edge (4N+1);
            \draw[-> , thick] (4N+1) edge (4N);
            \draw[-> , thick] (2N+1) edge (4N);
            \draw[-> , thick] (2N+1) edge (4N-2);
            
\draw[-> , thick] (3N) edge (4N);
             \draw[-> , thick] (4N+2) edge (4N);
            \draw[-> , thick] (4N-1) edge (3N);
           
            \draw[-> , thick] (4N-1) edge (4N+2);
          
            \draw[-> , thick] (3N) edge[bend right=30] (3N+1);
            
            \draw[-> , thick] (m1) edge (4N+2);
             \draw[-> , thick] (b6) edge (3N);
            \draw[-> , thick] (b7) edge (3N);
            \draw[- , thick] (3N) edge (b8);
            \draw[-> , thick] (b9) edge (3N);

    \end{scope}

\draw [-{Latex[length=3mm]}] (16,-1.2) -- (17,-1.2) node[midway,sloped,above] {$\mu_{4N}$};
 
 
 \node [draw,circle,fill=red!50,"{\footnotesize{$2N-1$}}"] (2N-1) at (20,0) {};
  \node [draw,circle,fill=red!50,"{\footnotesize{$2N$}}"] (2N) [right=of 2N-1] {};
  \node [draw,circle,fill=red!50,"{\footnotesize{$2N+1$}}"] (2N+1) [right=of 2N] {};
  \node [draw,circle,fill=red!50,"{\footnotesize{$2N+2$}}"] (2N+2) [right=of 2N+1] {};
  \node [draw,circle,fill=red!50,"{\footnotesize{$2N+3$}}"] (2N+3) [right=of 2N+2] {};
   \node (a1) [below left=of 2N-1] {};
  \node [draw,circle,fill=red!50,"{\footnotesize{$3N+1$}}" above ] (3N+1) [left=of a1] {};
     \node (a6)[right=of 2N+3]{};

  \node [draw,circle,fill=red!50,"{\footnotesize{$4N-2$}}" below] (4N-2) at (20,-2.5) {} ;
  \node (b1) [above left=of 4N-2] {} ;
     \node [draw,circle,fill=blue!50,"{\footnotesize{$4N+1$}}" below] (4N+1) [left=of b1] {} ;
  \node [draw,circle,fill=red!50,"{\footnotesize{$4N-1$}}" below] (4N-1) [right=of 4N-2] {} ; 
  \node [draw,circle,fill=red!50,"{\footnotesize{$4N$}}" below] (4N) [right=of 4N-1] {};
  \node [draw,circle,fill=red!50,"{\footnotesize{$1$}}" below] (1) [right=of 4N] {};
   \node [draw,circle,fill=red!50,"{\footnotesize{$2$}}" below] (2) [right=of 1] {};

   \node (r1)[left=of 4N-2]{};
   
   \node(l1) at (25.7,0){};
     \node(l2) at (25.7,-2.5){};
     
       \node(c1) at (19.3,0){};
     \node(c2) at (19.3,-2.5){};

    \node (b6) at (25.5,-3) {};
    \node (b7) at (25.5,-3.2) {};
     \node (m1) at (25.5,-3.4) {};
    \node (b8) at (25.5,-3.6) {};
    \node (b9) at (25.5,-3.8) {};
    
    \node (h1) at (25.4,-1) {};
        \node (h2) at (25.5, -5) {};
        
        \node(r2) at (18,0) {};

    \node [draw,circle,fill=red!50,"{\footnotesize{$3N$}}" below] (3N) at (22.5,-4.2) {} ;
    
    \   \node [draw,circle,fill=blue!50,"{\footnotesize{$4N+2$}}" below] (4N+2)[left=of 3N] {} ;

  \begin{scope}[>=Latex]

 \draw[- , thick] (4N+2) edge (r2);

 \draw[-> , thick] (4N-1) edge (2N);
     
     \draw[- , thick] (2N+3) edge (l1);
 \draw[- , thick] (2) edge (l2);
     \draw[- , thick] (2N-1) edge (c1);
 \draw[- , thick] (4N-2) edge (c2);

               \draw[-> , thick] (2N) edge (2N+1);
                 \draw[-> , thick] (2N+1) edge (2N+2);
                 \draw[-> , thick] (2N+2) edge (2N+3);
                   \draw[-> , thick] (4N-1) edge (4N);
                     \draw[-> , thick] (1) edge (4N);
                       \draw[-> , thick] (1) edge (2);
                         \draw[-> , thick] (4N) edge (2N+2);
                         \draw[-> , thick] (2N+3) edge (1);
                         
                            \draw[-> , thick] (2) edge (2N+2);
                            \draw[-> , thick] (2N) edge (r1);
                           
              \draw[-> , thick] (4N-2) edge (4N-1);
            \draw[-> , thick] (2N-1) edge (2N);
         
          \draw[-> , thick] (2N+2) edge (1);
           \draw[-> , thick] (3N+1) edge (1);
            \draw[-> , thick] (4N) edge (3N+1);
             \draw[-> , thick] (4N+1) edge (1);
            \draw[-> , thick] (4N) edge (4N+1);
            \draw[-> , thick] (4N) edge (2N+1);
            \draw[-> , thick] (2N+1) edge (4N-2);
            
\draw[-> , thick] (4N) edge (3N);
             \draw[-> , thick] (4N) edge (4N+2);
            \draw[-> , thick] (3N) edge (1);
            \draw[-> , thick] (4N+2) edge (1);

            \draw[-> , thick] (3N) edge[bend right=50] (3N+1);
            
            \draw[-> , thick] (m1) edge (4N+2);
             \draw[-> , thick] (b6) edge (3N);
            \draw[-> , thick] (b7) edge (3N);
            \draw[- , thick] (3N) edge (b8);
            \draw[-> , thick] (b9) edge (3N);    
            
            \end{scope}
 
 \draw [-{Latex[length=3mm]}] (-4,-9.2) -- (-3,-9.2) node[midway,sloped,above] {$\mu_{1}$};
 
 \draw[decorate sep={1mm}{5mm},fill] (0,-9.2) -- (,-9.2);
 
 
  \node [draw,circle,fill=red!50,"{\footnotesize{$3N-4$}}"] (3N-4) at (10,-8) {};
  \node [draw,circle,fill=red!50,"{\footnotesize{$3N-3$}}"] (3N-3) [right=of 3N-4] {};
  \node [draw,circle,fill=red!50,"{\footnotesize{$3N-2$}}"] (3N-2) [right=of 3N-3] {};
  \node [draw,circle,fill=red!50,"{\footnotesize{$3N-1$}}"] (3N-1) [right=of 3N-2] {};
  \node [draw,circle,fill=red!50,"{\footnotesize{$N$}}"] (N) [right=of 3N-1] {};
  \node (a1) [below left=of 3N-4] {};
  \node [draw,circle,fill=red!50,"{\footnotesize{$3N+1$}}"]  (3N+1) [left=of a1] {};
     \node (a6)[right=of N]{};

  \node [draw,circle,fill=red!50,"{\footnotesize{$N-5$}}" below] (N-5) at (10,-10.5) {} ;
    \node (r1)[left=of N-5]{};
  \node [draw,circle,fill=red!50,"{\footnotesize{$N-4$}}" below] (N-4) [right=of N-5] {} ; 
  
   \node [draw,circle,fill=blue!50,"{\footnotesize{$4N+1$}}" below] (4N+1) [above left=of r1] {} ;
  \node [draw,circle,fill=red!50,"{\footnotesize{$N-3$}}" below] (N-3) [right=of N-4] {};
  \node [draw,circle,fill=red!50,"{\footnotesize{$N-2$}}" below] (N-2) [right=of N-3] {};
   \node [draw,circle,fill=red!50,"{\footnotesize{$N-1$}}" below] (N-1) [right=of N-2] {};
   
   \node (r1)[left=of N-5]{};
  
    \node (b6) at (15.5,-11) {};
    \node (b7) at (15.5,-11.2) {};
     \node (m1) at (15.5,-11.4) {};
    \node (b8) at (15.5,-11.6) {};
    \node (b9) at (15.5,-11.8) {};
    
     \node (c1) at (9.3,-8){}; 
    \node (c2) at (9.3,-10.5){};
    
    \node (h1) at (15.4,-9) {};
        \node (h2) at (15.5, -13) {};

    \node(r2) at (8,-8) {};

    \node [draw,circle,fill=red!50,"{\footnotesize{$3N$}}" below] (3N) at (12.5,-12.2) {} ;
    
    \   \node [draw,circle,fill=blue!50,"{\footnotesize{$4N+2$}}" below] (4N+2)[left=of 3N] {} ;

  \draw [-{Latex[length=3mm]}] (6,-9.2) -- (7,-9.2) node[midway,sloped,above] {$\mu_{N-4}$};
  
  \begin{scope}[>=Latex]
  
\draw[- , thick] (3N-4) edge (c1);
\draw[- , thick] (N-5) edge (c2);
  
  \draw[- , thick] (4N+2) edge (r2);
    \draw[-> , thick] (3N-2) edge (N-3);
    \draw[-> , thick] (3N-1) edge (N-3);
    \draw[-> , thick] (N-3) edge (N-2);
    \draw[-> , thick] (N-3) edge (N-4);
    \draw[-> , thick] (3N+1) edge (N-3);
    \draw[-> , thick] (4N+1) edge (N-3);
    \draw[-> , thick] (N-4) edge (4N+1);
    \draw[-> , thick] (N-4) edge (3N+1);
    
    \draw[-> , thick] (N-4) edge (3N-2);
    \draw[-> , thick] (N-2) edge (3N-2);
    \draw[-> , thick] (3N-2) edge (N-5);
    \draw[-> , thick] (3N-3) edge (r1);
    
    \draw[-> , thick] (3N-4) edge (3N-3);
    \draw[-> , thick] (3N-3) edge (3N-2);
    \draw[-> , thick] (3N-2) edge (3N-1);
    \draw[-> , thick] (N) edge (3N-1);
    
    \draw[-> , thick] (N-2) edge (N);
    \draw[-> , thick] (N-1) edge (3N-1);
    \draw[-> , thick] (N-5) edge (N-4);
    \draw[-> , thick] (N-2) edge (N-1);

    \draw[-> , thick] (4N+2) edge (N-3);
    \draw[-> , thick] (N-4) edge (4N+2);
    \draw[-> , thick] (3N) edge (N-3);
    \draw[-> , thick] (N-4) edge (3N);
    
    \draw[-> , thick] (3N) edge[bend right=50] (3N+1);
            
            \draw[-> , thick] (N-1) edge (4N+2);
             \draw[-> , thick] (3N-1) edge (3N);
            \draw[-> , thick] (N) edge (3N);
            \draw[-> , thick] (3N) edge (N-2);
            \draw[-> , thick] (N-1) edge (3N);    
    
  \end{scope}
  
\draw [-{Latex[length=3mm]}] (16,-9.2) -- (17,-9.2) node[midway,sloped,above] {$\mu_{N-3}$};  
  
  
  \node [draw,circle,fill=red!50,"{\footnotesize{$3N-4$}}"] (3N-4) at (20,-8) {};
  \node [draw,circle,fill=red!50,"{\footnotesize{$3N-3$}}"] (3N-3) [right=of 3N-4] {};
  \node [draw,circle,fill=red!50,"{\footnotesize{$3N-2$}}"] (3N-2) [right=of 3N-3] {};
  \node [draw,circle,fill=red!50,"{\footnotesize{$3N-1$}}"] (3N-1) [right=of 3N-2] {};
  \node [draw,circle,fill=red!50,"{\footnotesize{$N$}}"] (N) [right=of 3N-1] {};
  \node (a1) [below left=of 3N-4] {};
  \node [draw,circle,fill=red!50,"{\footnotesize{$3N+1$}}" above]  (3N+1) [left=of a1] {};
     \node (a6)[right=of N]{};

  \node [draw,circle,fill=red!50,"{\footnotesize{$N-5$}}" below] (N-5) at (20,-10.5) {} ;
  \node (b1) [above left=of N-5] {} ;
     \node [draw,circle,fill=blue!50,"{\footnotesize{$4N+1$}}" below] (4N+1) [left=of b1] {} ;
  \node [draw,circle,fill=red!50,"{\footnotesize{$N-4$}}" below] (N-4) [right=of N-5] {} ; 
  \node [draw,circle,fill=red!50,"{\footnotesize{$N-3$}}" below] (N-3) [right=of N-4] {};
  \node [draw,circle,fill=red!50,"{\footnotesize{$N-2$}}" below] (N-2) [right=of N-3] {};
   \node [draw,circle,fill=red!50,"{\footnotesize{$N-1$}}" below] (N-1) [right=of N-2] {};
   
   \node (r1)[left=of N-5]{};
  
    \node (b6) at (25.5,-11) {};
    \node (b7) at (25.5,-11.2) {};
     \node (m1) at (25.5,-11.4) {};
    \node (b8) at (25.5,-11.6) {};
    \node (b9) at (25.5,-11.8) {};
    
    \node (h1) at (25.4,-9) {};
        \node (h2) at (25.5, -13) {};
        
         \node (c1) at (19.3,-8){}; 
    \node (c2) at (19.3,-10.5){};
        
     \node(r2) at (18,-8) {};

    \node [draw,circle,fill=red!50,"{\footnotesize{$3N$}}" below] (3N) at (22.5,-12.2) {} ;
    
    \   \node [draw,circle,fill=blue!50,"{\footnotesize{$4N+2$}}" below] (4N+2)[left=of 3N] {} ;

  \begin{scope}[>=Latex]
   \draw[- , thick] (3N-4) edge (c1);
    \draw[- , thick] (N-5) edge (c2);

  \draw[- , thick] (4N+2) edge (r2);
    \draw[-> , thick] (N-3) edge (3N-2);
    \draw[-> , thick] (N-3) edge (3N-1);
    \draw[-> , thick] (N-2) edge (N-3);
    \draw[-> , thick] (N-4) edge (N-3);
    \draw[-> , thick] (N-3) edge (3N+1);
    \draw[-> , thick] (N-3) edge (4N+1);
    \draw[-> , thick] (3N+1) edge (N-2);
    \draw[-> , thick] (4N+1) edge (N-2);

    \draw[-> , thick] (3N-2) edge (N-5);
    \draw[-> , thick] (3N-3) edge (r1);
    
    \draw[-> , thick] (3N-4) edge (3N-3);
    \draw[-> , thick] (3N-3) edge (3N-2);
    \draw[-> , thick] (3N-2) edge (3N-1);
    \draw[-> , thick] (N) edge (3N-1);
    
    \draw[-> , thick] (N-5) edge (3N-4);
    \draw[-> , thick] (N-4) edge (3N-3);
    \draw[-> , thick] (3N-1) edge (N-4);
    \draw[-> , thick] (3N-1) edge (N-2);
    
    \draw[-> , thick] (N-2) edge (N);
    \draw[-> , thick] (N-1) edge (3N-1);
    \draw[-> , thick] (N-5) edge (N-4);
    \draw[-> , thick] (N-2) edge (N-1);

    \draw[-> , thick] (N-3) edge (4N+2);
    \draw[-> , thick] (4N+2) edge (N-2);
    \draw[-> , thick] (N-3) edge (3N);

    \draw[-> , thick] (3N) edge[bend right=50] (3N+1);
            
            \draw[-> , thick] (N-1) edge (4N+2);
             \draw[-> , thick] (3N-1) edge (3N);
            \draw[-> , thick] (N) edge (3N);
            \draw[->> , thick] (3N) edge[bend right=40] (N-2);
            \draw[-> , thick] (N-1) edge (3N);    
    
  \end{scope}

  \draw [-{Latex[length=3mm]}] (-4,-17.2) -- (-3,-17.2) node[midway,sloped,above] {$\mu_{N-2}$};
  
  
  \node [draw,circle,fill=red!50,"{\footnotesize{$3N-4$}}"] (3N-4) at (0,-16) {};
  \node [draw,circle,fill=red!50,"{\footnotesize{$3N-3$}}"] (3N-3) [right=of 3N-4] {};
  \node [draw,circle,fill=red!50,"{\footnotesize{$3N-2$}}"] (3N-2) [right=of 3N-3] {};
  \node [draw,circle,fill=red!50,"{\footnotesize{$3N-1$}}"] (3N-1) [right=of 3N-2] {};
  \node [draw,circle,fill=red!50,"{\footnotesize{$N$}}"] (N) [right=of 3N-1] {};
  \node (a1) [below left=of 3N-4] {};
  \node [draw,circle,fill=red!50,"{\footnotesize{$3N+1$}}" above ]  (3N+1) [left=of a1] {};
     \node (a6)[right=of N]{};

  \node [draw,circle,fill=red!50,"{\footnotesize{$N-5$}}" below] (N-5) at (0,-18.5) {} ;
  \node (b1) [above left=of N-5] {} ;' 
     \node [draw,circle,fill=blue!50,"{\footnotesize{$4N+1$}}" below] (4N+1) [left=of b1] {} ;
  \node [draw,circle,fill=red!50,"{\footnotesize{$N-4$}}" below] (N-4) [right=of N-5] {} ; 
  \node [draw,circle,fill=red!50,"{\footnotesize{$N-3$}}" below] (N-3) [right=of N-4] {};
  \node [draw,circle,fill=red!50,"{\footnotesize{$N-2$}}" below] (N-2) [right=of N-3] {};
   \node [draw,circle,fill=red!50,"{\footnotesize{$N-1$}}" below] (N-1) [right=of N-2] {};
   
   \node (r1)[left=of N-5]{};
  
    \node (b6) at (5.5,-19) {};
    \node (b7) at (5.5,-19.2) {};
     \node (m1) at (5.5,-19.4) {};
    \node (b8) at (5.5,-19.6) {};
    \node (b9) at (5.5,-19.8) {};
    
    \node (c1) at (-0.7,-16){}; 
    \node (c2) at (-0.7,-18.5){};
    
    \node (h1) at (5.4,-17) {};
        \node (h2) at (5.5, -21) {};
        
        \node(r2) at (-2,-16) {};

    \node [draw,circle,fill=red!50,"{\footnotesize{$3N$}}" below] (3N) at (2.5,-20.2) {} ;
    
    \   \node [draw,circle,fill=blue!50,"{\footnotesize{$4N+2$}}" below] (4N+2)[left=of 3N] {} ;

  \begin{scope}[>=Latex]
  
  \draw[-, thick] (3N-4) edge (c1);
  \draw[-, thick] (N-5) edge (c2);
  
  \draw[-, thick] (4N+2) edge (r2);
  
    \draw[-> , thick] (N-3) edge (3N-2);

    \draw[-> , thick] (N-4) edge (N-3);
 
    \draw[-> , thick] (N-2) edge (3N+1);
    \draw[-> , thick] (N-2) edge (4N+1);

    \draw[-> , thick] (3N-2) edge (N-5);
    \draw[-> , thick] (3N-3) edge (r1);
    
    \draw[-> , thick] (3N-4) edge (3N-3);
    \draw[-> , thick] (3N-3) edge (3N-2);
    \draw[-> , thick] (3N-2) edge (3N-1);

    \draw[-> , thick] (N-5) edge (3N-4);
    \draw[-> , thick] (N-4) edge (3N-3);
    \draw[-> , thick] (3N-1) edge (N-4);
    \draw[-> , thick] (N-2) edge (3N-1);
    
    \draw[-> , thick] (N) edge (N-2);
  
    \draw[-> , thick] (N-5) edge (N-4);
    \draw[-> , thick] (N-1) edge (N-2);
    
    \draw[-> , thick] (N-2) edge (4N+2);
    \draw[-> , thick] (3N) edge (N-3);
     
     \draw[->> , thick] (N-2) edge[bend left=40] (3N);
    \draw[-> , thick] (3N) edge[bend right=50] (3N+1);
      
             \draw[-> , thick] (3N-1) edge (3N);
            \draw[-> , thick] (3N) edge (N);
      
      \draw[-> , thick] (3N+1) edge (N);   
            
  \draw[-> , thick] (3N+1) edge (N-1);    
  \draw[-> , thick] (3N) edge (N-1);     
  \draw[-> , thick] (4N+1) edge (N-1);  
  \draw[-> , thick] (4N+2) edge (N);    
    \draw[-> , thick] (N-3) edge (N-2);

  \end{scope}
  
  
 \draw [-{Latex[length=3mm]}] (6,-17.2) -- (7,-17.2) node[midway,sloped,above] {$\mu_{N-1}$};

 \node [draw,circle,fill=red!50,"{\footnotesize{$3N-4$}}"] (3N-4) at (10,-16) {};
  \node [draw,circle,fill=red!50,"{\footnotesize{$3N-3$}}"] (3N-3) [right=of 3N-4] {};
  \node [draw,circle,fill=red!50,"{\footnotesize{$3N-2$}}"] (3N-2) [right=of 3N-3] {};
  \node [draw,circle,fill=red!50,"{\footnotesize{$3N-1$}}"] (3N-1) [right=of 3N-2] {};
  \node [draw,circle,fill=red!50,"{\footnotesize{$N$}}"] (N) [right=of 3N-1] {};
   \node (a1) [below left=of 3N-4] {};
  \node [draw,circle,fill=red!50,"{\footnotesize{$3N+1$}}" above]  (3N+1) [left=of a1] {};
     \node (a6)[right=of N]{};

  \node [draw,circle,fill=red!50,"{\footnotesize{$N-5$}}" below] (N-5) at (10,-18.5) {} ;
  \node (b1)[above left=of N-5] {} ;
     \node [draw,circle,fill=blue!50,"{\footnotesize{$4N+1$}}" below] (4N+1) [left=of b1] {} ;
  \node [draw,circle,fill=red!50,"{\footnotesize{$N-4$}}" below] (N-4) [right=of N-5] {} ; 
  \node [draw,circle,fill=red!50,"{\footnotesize{$N-3$}}" below] (N-3) [right=of N-4] {};
  \node [draw,circle,fill=red!50,"{\footnotesize{$N-2$}}" below] (N-2) [right=of N-3] {};
   \node [draw,circle,fill=red!50,"{\footnotesize{$N-1$}}" below] (N-1) [right=of N-2] {};
   
   \node (r1)[left=of N-5]{};
  
    \node (b6) at (15.5,-19) {};
    \node (b7) at (15.5,-19.2) {};
     \node (m1) at (15.5,-19.4) {};
    \node (b8) at (15.5,-19.6) {};
    \node (b9) at (15.5,-19.8) {};
    
    \node (c1) at (9.3,-16){}; 
    \node (c2) at (9.3,-18.5){};
    
    \node (h1) at (15.4,-17) {};
        \node (h2) at (15.5, -21) {};
        
          \node(r2) at (8,-16) {};

    \node [draw,circle,fill=red!50,"{\footnotesize{$3N$}}" below] (3N) at (12.5,-20.2) {} ;
    
    \   \node [draw,circle,fill=blue!50,"{\footnotesize{$4N+2$}}" below] (4N+2)[left=of 3N] {} ;

  \begin{scope}[>=Latex]
  
  \draw[- , thick] (3N-4) edge (c1);
  \draw[- , thick] (N-5) edge (c2);
  
   \draw[- , thick] (4N+2) edge (r2);
  
 \draw[-> , thick] (N-3) edge (3N-2);

    \draw[-> , thick] (N-4) edge (N-3);

    \draw[-> , thick] (3N-2) edge (N-5);
    \draw[-> , thick] (3N-3) edge (r1);
    
    \draw[-> , thick] (3N-4) edge (3N-3);
    \draw[-> , thick] (3N-3) edge (3N-2);
    \draw[-> , thick] (3N-2) edge (3N-1);

    \draw[-> , thick] (N-5) edge (3N-4);
    \draw[-> , thick] (N-4) edge (3N-3);
    \draw[-> , thick] (3N-1) edge (N-4);
    \draw[-> , thick] (N-2) edge (3N-1);
    
    \draw[-> , thick] (N) edge (N-2);
  
    \draw[-> , thick] (N-5) edge (N-4);
    \draw[-> , thick] (N-2) edge (N-1);
    
    \draw[-> , thick] (N-2) edge (4N+2);
    \draw[-> , thick] (3N) edge (N-3);
     
     \draw[-> , thick] (N-2) edge[bend left=40] (3N);
    \draw[-> , thick] (3N) edge[bend right=50] (3N+1);
      
             \draw[-> , thick] (3N-1) edge (3N);
            \draw[-> , thick] (3N) edge (N);
      
            \draw[-> , thick] (N-1) edge (3N);    
            
  \draw[-> , thick] (N-1) edge (3N+1);    
   
  \draw[-> , thick] (N-1) edge (4N+1);  
  \draw[-> , thick] (4N+2) edge (N);    
    \draw[-> , thick] (N-3) edge (N-2);    
 \draw[-> , thick] (3N+1) edge (N);
  \end{scope}

   \draw [-{Latex[length=3mm]}] (16,-17.2) -- (17,-17.2) node[midway,sloped,above] {$\mu_{N}$};  
 

 \node [draw,circle,fill=red!50,"{\footnotesize{$3N-4$}}"] (3N-4) at (20,-16) {};
  \node [draw,circle,fill=red!50,"{\footnotesize{$3N-3$}}"] (3N-3) [right=of 3N-4] {};
  \node [draw,circle,fill=red!50,"{\footnotesize{$3N-2$}}"] (3N-2) [right=of 3N-3] {};
  \node [draw,circle,fill=red!50,"{\footnotesize{$3N-1$}}"] (3N-1) [right=of 3N-2] {};
  \node [draw,circle,fill=red!50,"{\footnotesize{$N$}}"] (N) [right=of 3N-1] {};
  \node (a1) [below left=of 3N-4] {};
  \node [draw,circle,fill=red!50,"{\footnotesize{$3N+1$}}" above]  (3N+1) [left=of a1] {};
     \node (a6)[right=of N]{};

  \node [draw,circle,fill=red!50,"{\footnotesize{$N-5$}}" below] (N-5) at (20,-18.5) {} ;
  \node (b1) [above left=of N-5]{};
     \node [draw,circle,fill=blue!50,"{\footnotesize{$4N+1$}}"  below] (4N+1) [left=of b1] {} ;
  \node [draw,circle,fill=red!50,"{\footnotesize{$N-4$}}" below] (N-4) [right=of N-5] {} ; 
  \node [draw,circle,fill=red!50,"{\footnotesize{$N-3$}}" below] (N-3) [right=of N-4] {};
  \node [draw,circle,fill=red!50,"{\footnotesize{$N-2$}}" below] (N-2) [right=of N-3] {};
   \node [draw,circle,fill=red!50,"{\footnotesize{$N-1$}}" below] (N-1) [right=of N-2] {};
   
   \node (r1)[left=of N-5]{};
  
    \node (b6) at (25.5,-19) {};
    \node (b7) at (25.5,-19.2) {};
     \node (m1) at (25.5,-19.4) {};
    \node (b8) at (25.5,-19.6) {};
    \node (b9) at (25.5,-19.8) {};
    
    \node (h1) at (25.4,-17) {};
        \node (h2) at (25.5, -21) {};
        
    \node (c1) at (19.3,-16){}; 
    \node (c2) at (19.3,-18.5){};
        
      \node(r2) at (19,-16) {};

    \node [draw,circle,fill=red!50,"{\footnotesize{$3N$}}" below] (3N) at (22.5,-20.2) {} ;
    
    \   \node [draw,circle,fill=blue!50,"{\footnotesize{$4N+2$}}" below] (4N+2)[left=of 3N] {} ;

  \begin{scope}[>=Latex]

  \draw[- , thick] (3N-4) edge (c1);
  \draw[- , thick] (N-5) edge (c2);
  
  \draw[- , thick] (4N+2) edge (r2);
  
 \draw[-> , thick] (N-3) edge (3N-2);

    \draw[-> , thick] (N-4) edge (N-3);

    \draw[-> , thick] (3N-2) edge (N-5);
    \draw[-> , thick] (3N-3) edge (r1);
    
    \draw[-> , thick] (3N-4) edge (3N-3);
    \draw[-> , thick] (3N-3) edge (3N-2);
    \draw[-> , thick] (3N-2) edge (3N-1);

    \draw[-> , thick] (N-5) edge (3N-4);
    \draw[-> , thick] (N-4) edge (3N-3);
    \draw[-> , thick] (3N-1) edge (N-4);
    \draw[-> , thick] (N-2) edge (3N-1);
    
    \draw[-> , thick] (N-2) edge (N);
  
    \draw[-> , thick] (N-5) edge (N-4);
    \draw[-> , thick] (N-2) edge (N-1);

    \draw[-> , thick] (3N) edge (N-3);

    \draw[-> , thick] (3N) edge[bend right=50] (3N+1);
      
             \draw[-> , thick] (3N-1) edge (3N);
            \draw[-> , thick] (N) edge (3N);
      
            \draw[-> , thick] (N-1) edge (3N);    
            
  \draw[-> , thick] (N-1) edge (3N+1);    
   
  \draw[-> , thick] (N-1) edge (4N+1);  
  \draw[-> , thick] (N) edge (4N+2);    
    \draw[-> , thick] (N-3) edge (N-2);    
      \draw[-> , thick] (3N+1) edge (N-2); 
      \draw[-> , thick] (N) edge (3N+1); 
 
  \end{scope}
\end{tikzpicture}
}
\end{center}
Thus after the sequence of mutations, the quiver yields following form. 
\begin{center}
\resizebox{0.9\textwidth}{!}{%
 \begin{tikzpicture}[every circle node/.style={draw,scale=0.6,thick},node distance=10mm]
 
 
\node [draw,circle,fill=red!50,"{\footnotesize{$N+2$}}"] (N+2) at (0,0) {};
  \node [draw,circle,fill=red!50,"{\footnotesize{$N+3$}}"] (N+3) [right=of N+2] {};
  \node [draw,circle,fill=red!50,"{\footnotesize{$N+4$}}"] (N+4) [right=of N+3] {};
  \node [draw,circle,fill=red!50,"{\footnotesize{$N+5$}}"] (N+5) [right=of N+4] {};
  \node [draw,circle,fill=red!50,"{\footnotesize{$N+6$}}"] (N+6) [right=of N+5] {};
  \node [draw,circle,fill=red!50,"{\footnotesize{$3N+1$}}" left]  (3N+1) [below left=of N+2] {};
     \node (a6)[right=of N+6]{};

  \node [draw,circle,fill=red!50,"{\footnotesize{$N+1$}}" below] (N+1) at (0,-2.5) {} ;
     \node [draw,circle,fill=blue!50,"{\footnotesize{$4N+1$}}" left] (4N+1) [above left=of N+1] {} ;
  \node [draw,circle,fill=red!50,"{\footnotesize{$3N+2$}}" below] (3N+2) [right=of N+1] {} ; 
  \node [draw,circle,fill=red!50,"{\footnotesize{$3N+3$}}" below] (3N+3) [right=of 3N+2] {};
  \node [draw,circle,fill=red!50,"{\footnotesize{$3N+4$}}" below] (3N+4) [right=of 3N+3] {};
   \node [draw,circle,fill=red!50,"{\footnotesize{$3N+5$}}" below] (3N+5) [right=of 3N+4] {};
   \node (b5)[right=of 3N+4]{};
    
  \node [draw,circle,fill=red!50,"{\footnotesize{$3N$}}" below] (3N) at (11,-4.2) {} ;
  
    \node [draw,circle,fill=blue!50,"{\footnotesize{$4N+2$}}" below] (4N+2) at (21.9,-1.7) {} ;

 \node(l1) at (5.7,0){};
     \node(l2) at (5.7,-2.5){};
     \node(l3) at (5.7,-0.7){};
     \node(l4) at (5.7,-1.8){};

    \draw[decorate sep={1mm}{5mm},fill] (5.5,-1.25) -- (6.5,-1.25);
\draw[decorate sep={1mm}{4mm},fill] (14,-1.25) -- (15,-1.25);

  \begin{scope}[>=Latex]
  
   \draw[-> , thick] (l3) edge (3N+4);
   \draw[- , thick] (N+6) edge (l1);
   \draw[- , thick] (3N+5) edge (l2);
   \draw[-> , thick] (l4) edge (3N+5);

  \draw[-> , thick] (N+1) edge (3N+2);
   
      \draw[-> , thick] (N+2) edge (3N+2);
      \draw[-> , thick] (3N+2) edge (3N+3);
        \draw[-> , thick] (3N+2) edge (N+3);
                   
            \draw[-> , thick] (N+3) edge (3N+1); 
            \draw[-> , thick] (3N) edge (3N+1);

             \draw[-> , thick] (N+3) edge (N+4); 
             
             \draw[-> , thick] (N+4) edge (N+5); 
       
               \draw[-> , thick] (3N+3) edge (3N+4); 
              \draw[-> , thick] (3N+4) edge (3N+5); 
              
               \draw[-> , thick] (N+5) edge (N+6); 
                \draw[-> , thick] (N+5) edge (3N+2); 
                \draw[-> , thick] (N+6) edge (3N+3); 
                
                 \draw[-> , thick] (3N+4) edge (N+5); 
                \draw[-> , thick] (3N+5) edge (N+6);

 \draw[-> , thick] (4N+2) edge(N+2); 
  \draw[-> , thick] (4N+1) edge (N+2); 
 
    \draw[-> , thick] (3N+1) edge (N+1); 
   \draw[-> , thick] (N+4) edge (N+1); 
     \draw[-> , thick] (3N+3) edge (3N+4); 
      \draw[-> , thick] (3N+3) edge (N+4);  
        \draw[-> , thick] (3N+1) edge (N+2); 
         \draw[-> , thick] (N+4) edge[bend right=15] (N+2); 
   
    \end{scope}

 \node [draw,circle,fill=red!50,"{\footnotesize{$2N-1$}}"] (2N-1) at (8,0) {};
  \node [draw,circle,fill=red!50,"{\footnotesize{$2N$}}"] (2N) [right=of 2N-1] {};
  \node [draw,circle,fill=red!50,"{\footnotesize{$2N+1$}}"] (2N+1) [right=of 2N] {};
  \node [draw,circle,fill=red!50,"{\footnotesize{$2N+2$}}"] (2N+2) [right=of 2N+1] {};
  \node [draw,circle,fill=red!50,"{\footnotesize{$2N+3$}}"] (2N+3) [right=of 2N+2] {};
      \node (a6)[right=of 2N+3]{};

  \node [draw,circle,fill=red!50,"{\footnotesize{$4N-2$}}" below] (4N-2) at (8,-2.5) {} ;
  \node [draw,circle,fill=red!50,"{\footnotesize{$4N-1$}}" below] (4N-1) [right=of 4N-2] {} ; 
  \node [draw,circle,fill=red!50,"{\footnotesize{$4N$}}" below] (4N) [right=of 4N-1] {};
  \node [draw,circle,fill=red!50,"{\footnotesize{$1$}}" below] (1) [right=of 4N] {};
   \node [draw,circle,fill=red!50,"{\footnotesize{$2$}}" below] (2) [right=of 1] {};

   \node (r1) at (7.3,-2){};
    \node (r2) at (7.3,-0.7){};
   
   \node(l1) at (13.7,0){};
     \node(l2) at (13.7,-2.5){};
     \node(l3) at (13.7,-0.7){};
     
       \node(c1) at (7.3,0){};
     \node(c2) at (7.3,-2.5){};

  \begin{scope}[>=Latex]

\draw[-> , thick] (4N-1) edge (2N);

\draw[- , thick] (2N) edge (r1);
\draw[- , thick] (2N-1) edge (r2);
     
     \draw[- , thick] (2N+3) edge (l1);
 \draw[- , thick] (2) edge (l2);
     \draw[- , thick] (2N-1) edge (c1);
 \draw[- , thick] (4N-2) edge (c2);
  \draw[-> , thick] (l3) edge (1);

               \draw[-> , thick] (2N) edge (2N+1);
                 \draw[-> , thick] (2N+1) edge (2N+2);
                 \draw[-> , thick] (2N+2) edge (2N+3);
                   \draw[-> , thick] (4N-1) edge (4N);
                     \draw[-> , thick] (4N) edge (1);
                       \draw[-> , thick] (1) edge (2);

              \draw[-> , thick] (4N-2) edge (4N-1);
            \draw[-> , thick] (2N-1) edge (2N);
         
          \draw[-> , thick] (1) edge (2N+2);

            \draw[-> , thick] (4N) edge (2N+1);
            \draw[-> , thick] (2N+1) edge (4N-2);
            
            \draw[-> , thick] (2N+2) edge (4N-1);
             \draw[-> , thick] (2N+3) edge (4N);
              \draw[-> , thick] (2) edge (2N+3);
             \draw[-> , thick] (4N-2) edge (2N-1);
                        
            \end{scope}

 \node [draw,circle,fill=red!50,"{\footnotesize{$3N-4$}}"] (3N-4) at (16,0) {};
  \node [draw,circle,fill=red!50,"{\footnotesize{$3N-3$}}"] (3N-3) [right=of 3N-4] {};
  \node [draw,circle,fill=red!50,"{\footnotesize{$3N-2$}}"] (3N-2) [right=of 3N-3] {};
  \node [draw,circle,fill=red!50,"{\footnotesize{$3N-1$}}"] (3N-1) [right=of 3N-2] {};
  \node [draw,circle,fill=red!50,"{\footnotesize{$N$}}"] (N) [right=of 3N-1] {};
      \node (a6)[right=of N]{};

  \node [draw,circle,fill=red!50,"{\footnotesize{$N-5$}}" below] (N-5) at (16,-2.5) {} ;
  \node [draw,circle,fill=red!50,"{\footnotesize{$N-4$}}" below] (N-4) [right=of N-5] {} ; 
  \node [draw,circle,fill=red!50,"{\footnotesize{$N-3$}}" below] (N-3) [right=of N-4] {};
  \node [draw,circle,fill=red!50,"{\footnotesize{$N-2$}}" below] (N-2) [right=of N-3] {};
   \node [draw,circle,fill=red!50,"{\footnotesize{$N-1$}}" right] (N-1) [below right=of N] {};
   
    \node (r1) at (15.3,-2){};
    \node (r2) at (15.3,-0.7){};

       \node(c1) at (15.3,0){};
     \node(c2) at (15.3,-2.5){};

  \begin{scope}[>=Latex]

    \draw[- , thick] (3N-4) edge (r2);
  \draw[- , thick] (3N-4) edge (c1);
  \draw[- , thick] (N-5) edge (c2);

 \draw[-> , thick] (N-3) edge (3N-2);

    \draw[-> , thick] (N-4) edge (N-3);

    \draw[-> , thick] (3N-2) edge (N-5);
    \draw[- , thick] (3N-3) edge (r1);
    
    \draw[-> , thick] (3N-4) edge (3N-3);
    \draw[-> , thick] (3N-3) edge (3N-2);
    \draw[-> , thick] (3N-2) edge (3N-1);

    \draw[-> , thick] (N-5) edge (3N-4);
    \draw[-> , thick] (N-4) edge (3N-3);
    \draw[-> , thick] (3N-1) edge (N-4);
    \draw[-> , thick] (N-2) edge (3N-1);
    
    \draw[-> , thick] (N-2) edge (N);
  
    \draw[-> , thick] (N-5) edge (N-4);
    \draw[-> , thick] (N-2) edge (N-1);

    \draw[-> , thick] (3N) edge (N-3);
     
     \draw[-> , thick] (N) edge (3N+1);

             \draw[-> , thick] (3N-1) edge (3N);
            \draw[-> , thick] (N) edge (3N);
      
            \draw[-> , thick] (N-1) edge (3N);    
            
  \draw[-> , thick] (N-1) edge (3N+1);    
   
  \draw[-> , thick] (N-1) edge (4N+1);  
  \draw[-> , thick] (N) edge (4N+2);    
    \draw[-> , thick] (N-3) edge (N-2);    
      \draw[-> , thick] (3N+1) edge (N-2); 
 
  \end{scope}

\end{tikzpicture}
}
\end{center}

By relabelling all of the nodes in the resulting quiver to the permutation $\rho= (N-1,N)(N+1,N+2,\dots, 3N,3N+1)$, the quiver returns to the initial form. Thus it is mutation periodic up to permutation $\rho$ i.e.  $\mu_{N}\mu_{N-1} \cdots \mu_{1}\mu_{4N}\mu_{4N-1} \cdots \mu_{3N+2}\mu_{N+1} (Q) = \rho(Q)$.

\end{proof}

Therefore since this composition of mutations is mutation periodic, one can define the cluster map $\psi = \rho^{-1}\mu_{N}\mu_{N-1} \cdots \mu_{1}\mu_{4N}\mu_{4N-1} \cdots \mu_{3N+2}\mu_{N+1}$ such that its iteration on the initial cluster , 
$\vb{\tx}_{2N} = (\tx_{1},\tx_{2}, \dots, \tx_{4N}) $ \eqref{D2Ninitialcluster},
generates new cluster variables given by the following system of exchange relations, 
\begin{equation}
\begin{split}
\tx_{N+1}' \tx_{N+1}  &= b_{5}b_{6} \tx_{3N}\tx_{N+3} + \tx_{3N+2} \\ 
\tx_{3N+2}'\tx_{3N+2} &= b_{5}b_{6} \tx_{3N}\tx_{3N+1}\tx_{N+3}\tx_{N+4} + \tx_{N+1}' \tx_{3N+3} \tx_{N+2} \\ 
\tx_{3N+3}'\tx_{3N+3} &= b_{5}b_{6} \tx_{3N}\tx_{3N+1}\tx_{N+4}\tx_{N+5} + \tx_{3N+2}' \tx_{3N+4} \\ 
& \vdots \\
\tx_{4N-1}'\tx_{4N-1} &= b_{5}b_{6} \tx_{3N}\tx_{3N+1}\tx_{2N}\tx_{2N+1} + \tx_{4N-2}'\tx_{4N} \\ 
\tx_{4N}'\tx_{4N} &= b_{5}b_{6} \tx_{3N}\tx_{3N+1}\tx_{2N+1}\tx_{2N+2} + \tx_{4N-1}'\tx_{1} \\ 
\tx_{1}'\tx_{1} &= b_{5}b_{6} \tx_{3N}\tx_{3N+1}\tx_{2N+2}\tx_{2N+3} + \tx_{4N}'\tx_{2} \\ 
\tx_{2}'\tx_{2} &= b_{5}b_{6} \tx_{3N}\tx_{3N+1}\tx_{2N+3}\tx_{2N+4} + \tx_{1}'\tx_{3} \\ 
& \vdots \\ 
\tx_{N-2}'\tx_{N-2} &= b_{5}b_{6} \tx_{3N}^2\tx_{3N+1}\tx_{3N-1} + \tx_{N-3}'\tx_{N-1}\tx_{N} \\ 
\tx_{N-1}'\tx_{N-1} &= b_{5} \tx_{3N}\tx_{3N+1} + \tx_{N-2}' \\ 
\tx_{N}'\tx_{N} &= b_{6} \tx_{3N}\tx_{3N+1} + \tx_{N-2}' \\ 
\end{split}
\end{equation}
Upon setting the initial cluster as $\vb{\tx}_{2N}= (\xi_{N-2,0}, \dots, \xi_{1,0}, r_{0},s_{0}, \tau_{0},\tau_{1}, \dots, \tau_{2N},\chi_{0}, \eta_{1,0}, \dots , \eta_{N-2,0} )$ and denoting 
\begin{equation}
\begin{split}
&\tx_{i}'=\xi_{N-1-i,1} \ (1 \leq i \leq N-2), \quad  \tx_{N-1}' =s_{1}, \ \tx_{N}' =r_{1}, \quad \tx_{N+1}' = \tau_{2N+1}, \\
&  \tx_{3N+2}' =\chi_{1}, \quad \tx_{j}' = \eta_{j- 3N-1,1} \ ( 3N+2 \leq j \leq 4N),\\
\end{split}
\end{equation}
iteration of the cluster maps is equivalent to the following system of recursion relations. 
\begin{equation}\label{D2NLauretnrel}
\begin{split}
\tau_{n+2N+1}\tau_{n} &= b_{5}b_{6} \tau_{n+2N-1}\tau_{n+2} + \chi_{n} \\ 
\chi_{n+1}\chi_{n} &= b_{5}b_{6}\tau_{n+2N-1}\tau_{n+2N}\tau_{n+2}\tau_{n+3} + \eta_{1,n}\tau_{n+2N+1}\tau_{n+1} \\ 
\eta_{1,n+1}\eta_{1,n} &= b_{5}b_{6}\tau_{n+2N-1}\tau_{n+2N}\tau_{n+3} \tau_{n+4} + \chi_{n+1}\eta_{2,n} \\ 
&  \vdots \\
\eta_{N-3,n+1} \eta_{N-3,n} &= b_{5}b_{6}\tau_{n+2N-1}\tau_{n+2N} \tau_{n+N-1}\tau_{n+N} + \eta_{N-4,n+1}\eta_{N-2,n} \\ 
\eta_{N-2,n+1} \eta_{N-2,n} &= b_{5}b_{6}\tau_{n+2N-1}\tau_{n+2N} \tau_{n+N}\tau_{n+N+1} + \eta_{N-3,n+1}\xi_{N-2,n} \\ 
\xi_{N-2,n+1} \xi_{N-2,n} &= b_{5}b_{6} \tau_{n+2N-1}\tau_{n+2N}\tau_{n+N+1}\tau_{n+N+2} + \eta_{N-2,n+1}\xi_{N-3,n} \\
\xi_{N-3,n+1} \xi_{N-3,n} &= b_{5}b_{6} \tau_{n+2N-1}\tau_{n+2N}\tau_{n+N+2}\tau_{n+N+3} + \xi_{N-2,n+1}\xi_{N-4,n} \\
& \vdots \\ 
\xi_{1,n+1} \xi_{1,n} &= b_{5}b_{6} \tau_{n+2N-1}^{2}\tau_{n+2N} \tau_{n+2N-2} + \xi_{2,n+1} r_{n}s_{n} \\ 
s_{n+1}r_{n} &= b_{5} \tau_{n+2N-1} \tau_{n+2N} + \xi_{1,n+1} \\
r_{n+1}s_{n} &= b_{6} \tau_{n+2N-1} \tau_{n+2N} + \xi_{1,n+1} \\
\end{split}
\end{equation}

Writing the variable transformations, 
\begin{equation}
\begin{split}
&x_{1,n} =  \frac{\tau_{n+2N}\tau_{n}}{\tau_{n+2N-1}\tau_{n+1}}, \quad x_{2,n} =  \frac{\chi_{n}}{\tau_{n+2N-1}\tau_{n+2}}, \quad x_{3,n} = \frac{\eta_{1,n}}{\tau_{n+2N-1}\tau_{n+3}}, \\
&\dots, x_{N,n} = \frac{\eta_{N-2,n}}{\tau_{n+2N-1}\tau_{n+N}},x_{N+1,n} = \frac{\xi_{N-2,n}}{\tau_{n+2N-1}\tau_{n+N+1}},  \dots, \\
& x_{2N-2,n} = \frac{\xi_{1,n}}{\tau_{n+2N-1} \tau_{n+2N-2}},\quad x_{2N-1,n} = \frac{r_{n}}{\tau_{n+2N-1}}, \quad x_{2N,n} = \frac{s_{n}}{\tau_{n+2N-1}} \\
\end{split} 
\end{equation} 
the recursion relation reduces to expressions
\begin{equation}
\begin{split}
x_{1,n+1} x_{1,n} & = b_{5}b_{6} + x_{2,n} \\
x_{2,n+1} x_{2,n} & = b_{5}b_{6} + x_{3,n}x_{1,n+1} \\
& \vdots \\
x_{2N-3,n+1} x_{2N-3,n} & = b_{5}b_{6} + x_{2N-2,n}x_{2N-4,n+1} \\
x_{2N-2,n+1} x_{2N-2,n} & = b_{5}b_{6} + x_{2N-1,n}x_{2N,n}x_{2N-3,n+1} \\
x_{2N-1,n+1} x_{2N-1,n} & = b_{5} + x_{2N-2,n+1} \\
x_{2N,n+1} x_{2N,n} & = b_{6} + x_{2N-2,n+1} \\
\end{split}
\end{equation}
which correspond to the deformed type $D_{2N}$ map. This shows that the deformed type $D_{2N}$ can be expressed as a cluster map defined by the cluster mutations in the cluster algebra whose initial seed is composed of cluster $\vb{\tx} = (\tx_{1},\dots, \tx_{4N})$ and the quiver $Q_{D_{2N}}$.  
 
Since we restored the Laurent property of the deformed map, we can follow the procedure described in the section \ref{ss:tropD6} to determine the degree growths of the map and thereby compute the corresponding algebraic entropy, which provides criteria for detecting integrability. 

\section{Degree growth of deformed type $D_{2N}$ map } \label{S:degreegrowth}

In the previous section, we have shown that the one can construct the family of cluster maps from the generalized quiver that can be projected to the reduced space where the map is deformed type $D_{2N}$ map. 
Here we calculate the degree growth of the cluster map which correlates with the degree growth of deformed map. 

The cluster variables generated by iteration of the cluster maps $\psi_{D_{2N}}$, starting from the initial cluster $\vb{\tx}_{2N}= (\xi_{N-2,0}, \dots, \xi_{1,0}, r_{0},s_{0}, \tau_{0},\tau_{1}, \dots, \tau_{2N},\chi_{0}, \eta_{1,0}, \dots , \eta_{N-2,0} )$, can be written in the form, 
\begin{equation}\label{D2Nvariables}
\begin{split}
&\xi_{N-1-i,n} = \frac{P^{(i)}_{n}({\vb{\tx}})}{{\vb{\tx}}^{\vb{h}_{N-1-i,n}}}, \quad r_{n} = \frac{P^{(N-1)}_{n}(\vb{\tx})}{{\vb{\tx}}^{\vb{r}_{n}}}, \quad s_{n} = \frac{P^{(N)}_{n}(\vb{\tx})}{{\vb{\tx}}^{\vb{s}_{i,n}}}, \\
&\tau_{n} = \frac{P^{(N+1)}_{n}({\vb{\tx}})}{{\vb{\tx}}^{\vb{d}_{n}}}, \quad  \chi_{n} = \frac{P^{(N+2)}_{n}({\vb{\tx}})}{{\vb{\tx}}^{\vb{f}_{n}}},\quad \eta_{i,n} = \frac{P^{(i+N+2)}_{n}({\vb{\tx}})}{{\vb{\tx}}^{\vb{g}_{i,n}}} \\
\end{split}
\end{equation}
, for $1 \leq i \leq N-2$, where degree vector in each denominator  is $4N \times 1 $ vector (e.g. $\vb{d}_{n}= (h_{N-2,n}^{(N+1)},h_{N-3,n}^{(N+1)}, \dots, h_{1,n}^{(N+1)}, r^{(N+1)}_{n}, s^{(N+1)}_{n}, d^{(N+1)}_{n}, \dots, d^{(N+1)}_{n+2N}, f_{n}^{(N+1)}, g_{1,n}^{(N+1)},\dots,   g_{N-2,n}^{(N+1)} )$  which has initial data given by $ 4N \times 4N$ identity matrix, 
\begin{equation}\label{initialdvecD2N}
\qty(\vb{h}_{N-2,n},\vb{h}_{N-3,n}, \dots, \vb{h}_{1,n}, \vb{r}_{n}, \vb{s}_{n}, \vb{d}_{n}, \dots, \vb{d}_{n+2N}, \vb{f}_{n}, \vb{g}_{1,n},\dots,   \vb{g}_{N-2,n} ) = -I
\end{equation}
Substituting these variables into the relations \eqref{D2NLauretnrel} and comparing the denominator on both sides gives following (max,+) relations, 
\begin{equation}\label{maxplusD2N}
\begin{array}{rcl}
\vb{d}_{n+2N+1} + \vb{d}_{n} & = & \max(\vb{d}_{n+2N-1} + \vb{d}_{n+2}, \vb{f}_{n} ), \\
 \vb{f}_{n+1} + \vb{f}_{n} & = & \max(\vb{d}_{n+2N-1} + \vb{d}_{n+2N} +\vb{d}_{n+2} + \vb{d}_{n+3}, \vb{g}_{1,n} + \vb{d}_{n+2N+1} + \vb{d}_{n+1} ), \\
 \vb{g}_{1,n+1} + \vb{g}_{1,n} & = & \max (\vb{d}_{n+2N-1} + \vb{d}_{n+2N} +\vb{d}_{n+3} + \vb{d}_{n+4}, \vb{f}_{n+1} + \vb{g}_{2,n}) \\
 & \vdots \\ 
  \vb{g}_{N-3,n+1} + \vb{g}_{N-3,n} & = & \max(\vb{d}_{n+2N-1} + \vb{d}_{n+2N} +\vb{d}_{n+N-1} + \vb{d}_{n+N}, \vb{g}_{N-4,n+1} + \vb{g}_{N-2,n} ), \\
 \vb{g}_{N-2,n+1} + \vb{g}_{N-2,n} & = & \max(\vb{d}_{n+2N-1} + \vb{d}_{n+2N} +\vb{d}_{n+N} + \vb{d}_{n+N+1}, \vb{g}_{N-3,n+1} + \vb{h}_{N-2,n}), \\
  \vb{h}_{N-2,n+1} + \vb{h}_{N-2, n} & = & \max(\vb{d}_{n+2N-1} + \vb{d}_{n+2N} +\vb{d}_{n+N+1} + \vb{d}_{n+N+2}, \vb{g}_{N-2,n+1} + \vb{h}_{N-3,n}), \\
   \vb{h}_{N-3,n+1} + \vb{h}_{N-3, n} & = & \max(\vb{d}_{n+2N-1} + \vb{d}_{n+2N} +\vb{d}_{n+N+2} + \vb{d}_{n+N+4}, \vb{h}_{N-2,n+1} + \vb{h}_{N-4,n}), \\
  & \vdots \\
   \vb{h}_{1,n+1} + \vb{h}_{1, n}& = & \max(2\vb{d}_{n+2N -1 } + \vb{d}_{n+2N}+ \vb{d}_{n+2N-2}, \vb{h}_{2,n+1} + \vb{r}_{n} + \vb{s}_{n}), \\
    \vb{s}_{n+1} + \vb{r}_{ n}& = & \max(\vb{d}_{n+2N -1 } + \vb{d}_{n+2N}, \vb{h}_{1,n+1} ), \\
      \vb{r}_{n+1} + \vb{s}_{ n}& = & \max(\vb{d}_{n+2N -1 } + \vb{d}_{n+2N}, \vb{h}_{1,n+1} ), \\
\end{array}  
\end{equation}
Next we introduce quantities which is analogous to the tropical version of \eqref{vartransformD6} as following,
\begin{equation}
\begin{split}
\vb{X}_{1,n} = \vb{d}_{n+2N} + \vb{d}_{n} - \vb{d}_{n+2N-1} - \vb{d}_{n+1} ,& \quad \vb{X}_{2,n} = \vb{f}_{n} - \vb{d}_{n+2N-1} - \vb{d}_{n+2},\\
\vb{X}_{3,n} =  \vb{g}_{1,n} - \vb{d}_{n+2N-1} - \vb{d}_{n+3}, & \quad     \vb{X}_{4,n} = \vb{g}_{2,n} - \vb{d}_{n+2N-1} - \vb{d}_{n+4}, \\ 
&\vdots \\
  \vb{X}_{N,n} = \vb{g}_{N-2,n} - \vb{d}_{n+2N-1} - \vb{d}_{n+N}, & \quad   \vb{X}_{N+1,n} = \vb{h}_{N-2,n} - \vb{d}_{n+2N-1} - \vb{d}_{n+N+1}, \\ 
 & \vdots \\
  \vb{X}_{2N-2,n} = \vb{h}_{1,n} - \vb{d}_{n+2N-1} - \vb{d}_{n+2N-2}, & \quad   \vb{X}_{2N-1,n} = \vb{r}_{n} - \vb{d}_{n+2N-1}, \\ 
  \vb{X}_{2N,n} = \vb{s}_{n} - \vb{d}_{n+2N-1} 
\end{split}
\end{equation}
In this setting, the (max,+) relations above can be reformulated as follows, 
\begin{equation}\label{QuantityD2N}
\begin{split}
\vb{X}_{1,n+1} + \vb{X}_{1,n} &= \max(\vb{X}_{2,n}, 0) \\
\vb{X}_{2,n+1} + \vb{X}_{2,n} &=  \max(\vb{X}_{3,n} + \vb{X}_{1,n+1}, 0) \\
& \ \vdots \\ 
\vb{X}_{2N-3,n+1} + \vb{X}_{2N-3,n} &=  \max(\vb{X}_{2N-2,n} + \vb{X}_{2N-4,n+1}, 0) \\
\vb{X}_{2N-2,n+1} + \vb{X}_{2N-2,n} &=  \max(\vb{X}_{2N-1,n} + \vb{X}_{2N,n} + \vb{X}_{2N-3,n+1}, 0) \\
\vb{X}_{2N-1,n+1} + \vb{X}_{2N-1,n} &=  \max(\vb{X}_{2N-2,n+1}, 0) \\
\vb{X}_{2N,n+1} + \vb{X}_{2N,n} &=  \max(\vb{X}_{2N-2,n+1}, 0) \\
 \end{split}
\end{equation}
These relations can be obtained by tropicalizing the exchange relations  that define original cluster map $\varphi_{D_{2N}}$. Thus we have following result
This is tropical analogue of the system of equations given by the deformed map  $\tilde{\varphi}_{D_{2N}}$. Therefore we have following result. 
\begin{lm}
Each quantity $\vb{X}_{i,n}$, which satisfies the tropical analogue of the deformed map $\tilde{\varphi}_{D_{2N}}$,\eqref{QuantityD2N},  is periodic with period $2N$.
\end{lm}

\begin{proof}
This can be shown by taking same procedure which is used in Lemma \ref{periodtropD6}. 
\end{proof}

Such periodicity property of each $\vb{X}_{i,n}$ enables us to find an  explicit formula for growths of each degree vector, as shown below. 
 \begin{thm}
 Given that the initial data of the degree vectors $\vb{d}_{n}, \vb{f}_{n}, \vb{g}_{i,n}, \vb{h}_{i,n}, \vb{s}_{n},\vb{r}_{n}$  is specified by \eqref{initialdvecD2N}, the  system of equations \eqref{maxplusD2N} has a solution, expressed as follows  
 \begin{equation}
\begin{array}{rrr}
\vb{d}_{n} =\frac{n^2}{8N^2 - 4N}\vb{a} + O(n),& \quad  \vb{r}_{n} =  \frac{n^2}{8N^2 - 4N}\vb{a} + O(n),& \quad \vb{s}_{n} =\frac{n^2}{8N^2 - 4N}\vb{a} + O(n)  \\
\vb{f}_{i,n}=  \frac{n^2}{4N^2 - 2N}\vb{a} + O(n), & \quad \vb{h}_{i,n} =  \frac{n^2}{4N^2 - 2N}\vb{a} + O(n), & \quad \vb{g}_{i,n} =  \frac{n^2}{4N^2 - 2N}\vb{a} + O(n)
\end{array}
\end{equation}
 where $1\leq i \leq N-2$ and  $\vb{a}= (a_{k})_{1 \leq k \leq 4N}$, where $a_{l} = 2$ for $l \in \qty{1, \dots, N-2} \cup \qty{3N+2,\dots, 4N}$ and $a_{j}=1$ for $m=\qty{N-1,N}$. 
\end{thm}

\begin{proof}
By using the fact that $\vb{X}_{i,n} = \vb{X}_{i,n+2N}$ (periodicity of $\vb{X}_{i,n}$),   we can obtain the linear relations for the degree vector $\vb{d}_{n}$ from the quantity $\vb{X}_{1,n}$ as shown as follows, 
\begin{align*}
&(\cT^{2N} - 1) (\vb{d}_{n+2N} + \vb{d}_{n} - \vb{d}_{n+2N-1} - \vb{d}_{n+1}) = 0 \\
& \iff (\cT^{2N} - 1) (\cT^{2N} - \cT^{2N-1} - \cT + 1) = 0 \\ 
&\iff  (\cT-1 )^3 (\cT^{2N-1} + \cT^{2N-2} + \cdots + \cT + 1) (\cT^{2N-2} + \cT^{2N-3} + \cdots + \cT + 1)
\end{align*}
The roots of this equation are $\exp(\frac{2k\pi i }{2N})$, $\exp(\frac{2k\pi i}{2N-1})$ and 1 with multiplicity 3. Thus the solution is of the form
\begin{equation}
\vb{d}_{n} = \vb{c} n^2 + O(n)
\end{equation} 
where $\vb{c}$ is a constant vector. The constant $\vb{c}$ can be determined by factorising the $\cT$ equation in the following way, 
\begin{align*}
&(\cT^{2N} - 1) (\cT^{2N-1} - 1) (\cT-1)\vb{d}_{n} = 0 \\ 
& \iff (\cT^{2N} - 1) (\cT^{2N-1} - 1)\vb{d}_{n} =  (\cT^{2N} - 1) (\cT^{2N-1} - 1)\vb{d}_{n+1} \\ 
&\implies   (\cT^{2N} - 1) (\cT^{2N-1} - 1)\vb{d}_{n} =(8N^2 - 4N)\vb{c} =  \vb{a}
\end{align*}
With inductive approach, one can deduce that constant $\vb{a}= (a_{k})_{1 \leq k \leq 4N}$, where  $\vb{a}= (a_{k})_{1 \leq k \leq 4N}$, where $a_{l} = 2$ for $l \in \qty{1, \dots, N-2} \cup \qty{3N+2,\dots, 4N}$ and $a_{j}=1$ for $m=\qty{N-1,N}$.   
The remaining degree vectors can be derived by using periodicity property of quantities $\vb{X}_{i,n}$ together with the relation for  $\vb{d}_{n}$. For instance, applying the operator $\cL = (\cT^{2N} - 1) (\cT^{2N-1} - 1) (\cT-1)$ to $\vb{X}_{2,n}$, we have 
\begin{equation}
\begin{split}
&\cL \vb{X}_{2,n} = \cL \vb{f}_{n} - \underbrace{\cL \vb{d}_{n+2N-1}}_{0} - \underbrace{\cL\vb{d}_{n+2}}_{0} = 0\\
&\implies \vb{f}_{n} = \vb{c}_{2} n^2 + O(n)
\end{split}
\end{equation}
The constant vector $\vb{c}_{2}$ can be determined immediately from the $\vb{X}_{2,n}$.
\begin{equation}
\begin{split}
&(\cT^{2N} - 1) (\cT^{2N-1} - 1)\vb{f}_{n} = (\cT^{2N} - 1) (\cT^{2N-1} - 1)(\vb{d}_{n+2} + \vb{d}_{n+2N-1}) \\
&= 2\vb{c} + O(n)
\end{split}
\end{equation} 
Thus this d-vector is expressed $\vb{f}_{n} = \frac{1}{4N^2-2N}\vb{a} n^2 +O(n)$. Repeating this procedure to other quantities $\vb{X}_{i,n}$, we can acquire the desired results. 

\end{proof}
As a result, the degree growth of  cluster map $\psi_{D_{2N}}$ is quadratic which implies that the associated algebraic entropy vanishes.  Hence we arrive at the following conjecture. 
\begin{conj}
The discrete dynamics induced by the iteration of cluster maps $\psi_{D_{2N}}$ is a discrete integrable system. 
\end{conj}

\section{Conclusion and future work }

In this paper, we considered one of the main results introduced in \cite{hkm24}, namely, the deformation of periodic cluster map associated with Dynkin type $D_4$. We showed that the $D_4$ cluster map admits a 2-parameter integrable deformation and then lifted via Laurentification to a cluster map $\psi_{D_{4}^{(2)}}$ in 8 dimensions, with 2 frozen variables, that is generated by the quiver $Q_{D_{4}}^{(2)}$, (shown in Figure \ref{DeformedQD4}). We then demonstrated that quiver corresponding to the deformed type $D_{6}$ map can be constructed from $Q_{D_{4}}^{(2)}$ by inserting "ladder-shaped" quiver the local expansion (Figure \ref{extensionD4toD6}), in a similar manner to that of \cite{grab}. As a result, we can construct a family of quivers (illustrated in Figure \ref{D2Nquiver}) that generate two parametric cluster maps arising from deformed type $D_{2N}$ via Laurentification.

In recent studies of the type $D_4$ case with Hone and Mase, we have found a distinguished cluster map $\hat{\psi} $, built from a composition of 6 mutations and a specific permutation applied to the same cluster algebra. This map preserves the same invariant function (first integral of $\psi_{D_{4}^{(2)}}$) which is also an elliptic curve as it has genus 1. This means that each map acts as a translation by a point on the curve. Thus we expect these two to be independent to each other and hence describe the Mordell-Weil group of the corresponding elliptic surface of rank 2 like in the case of type $A_{3}$  (considered in \cite{hkm24}), which is under investigation.

 In the paper \cite{MOT}, Masuda, Okubo and Tsuda demonstrated that the translation of the Weyl group, composed with $R$ and other distinct reflections, gives rise to the discrete Painlev\'{e} equation $q$-$P_{VI}$.  In recent studies, we have found correspondence between cluster map $\hat{\psi} $ and a certain birational map $R$ that is one of the generating elements (\textit{reflection})of a Weyl group,  $W(D^{(1)}_{5})$. We expect that there is a close relation between cluster algebra with quiver $Q_{D_{4}}^{(2)}$ and $q$-$P_{VI}$, which is subject to future work.
\\
\\
\\
\textbf{Conflict of interest}:
The author declare that there is no conflict of interest.
\\
\\
\textbf{Data availability statement}:
No new data were created or analysed in this study.
\\
\\
\textbf{Acknowledgments}:
 This research was supported by the Grant-in-Aid for Scientific Research of Japan Society for the Promotion of Science, JSPS KAKENHI Grant Number 24KF0208. The author is grateful to Andrew Hone (University of Kent), Jan Grabowski (Lancaster University) and Takafumi Mase (The University of Tokyo) for their valuable advice and helpful comments for this paper.


\begin{thebibliography}{99}

\bibitem{fz2} 
S.~Fomin and A.~Zelevinsky,  
Cluster algebras. II. Finite type classification, Invent. Math. {\bf 154} (2003) 63--121.

\bibitem{fz4} 
S.~Fomin and A.~Zelevinsky,  
Cluster algebras: IV. Coefficients, Compos. Math. 143 (2007) 112--64.

\bibitem{fr} S. Fomin and N. Reading, Root systems and generalized associahedra, in PCMI Lecture Series volume; Geometric combinatorics, 63-131, IAS/Park City Math. Ser., 13, Amer. Math. Soc., Providence, RI, (2007).

\bibitem{grab} J.E. Grabowski, A.N.W. Hone and W. Kim, 
Deformed cluster maps of type $A_{2N}$, preprint (2024); 
{\tt arXiv:2402.18310} 

\bibitem{hk} A.N.W. Hone and T.E. Kouloukas,
Deformations of cluster mutations and invariant presymplectic forms, 
J. Alg. Comb. 57 (2023) 763--791.

\bibitem{kun} A. Kuniba, T. Nakanishi and J.  Suzuki, 
T-systems and Y-systems in integrable systems, J. Phys. A: Math. Theor. 44 (2011) 103001.

\bibitem{os} K. Oguiso and T. Shioda,  
The Mordell-Weil Lattice of a Rational Elliptic Surface, 
Commentarii Mathematici Universitatis Sancti Pauli 40 
(1991) 83--99. 

\bibitem{okubo} 
N. Okubo, 
Bilinear equations and
q-discrete Painlevé
equations satisfied by variables and coefficients in
cluster algebras, 
J. Phys. A: Math. Theor. 48 (2015) 355201. 


\bibitem{zam} Al.B. Zamolodchikov, 
On the thermodynamic Bethe ansatz equations for reflectionless ADE scattering theories,
Phys. Lett. B 253 (1991) 391--394.


	\bibitem{hkm24}
	A. N. W. Hone, W. Kim, T. Mase,
	New cluster algebras from old: integrability beyond Zamolodchikov periodicity,
	\textit{Journal of Physics A: Mathematical and Theoretical}
	57
	(2024):
	415201.

	\bibitem{mase17}
	T. Mase,
	Studies on spaces of initial conditions for nonautonomous mappings of the plane,
	\textit{Journal of Integrable Systems}
	3
	(2018):
	xyy010.

	\bibitem{hkha17}

Hamad. Khaled, Hone. Andrew, van der Kamp. Peter, Quispel. G,
QRT maps and related Laurent systems, 
\textit{Advances in Applied Mathematics} 
 (2017).

\bibitem{kanki}
M. Kanki, J. Mada, K. M. Tamizhmani and T. Tokihiro
Discrete Painlevé II equation over finite fields,
\textit{J. Phys. A: Math. Theor. 45 342001}

 (2012)
 
\bibitem{GRP}
Grammaticos. B, Ramani. A, Papageorgiou. V,
Do integrable mappings have the Painlev\'e property?, 
\textit{Phys. Rev. Lett.}
(1991)

\bibitem{MOT}
T. Masuda, N. Okubo,  T. Tsuda,
Birational Weyl group actions and q-Painlev\'{e} equations via mutation combinatorics in cluster algebras, preprint (2025),
\textit{ arXiv:2303.06704}


\bibitem{Shigeru}
S. Maeda,
Completely integrable symplectic mapping, 
\textit{Proceedings of the Japan Academy, Series A, Mathematical Sciences}
(1987)

\bibitem{FH2013}
A. Fordy and A. Hone,
Discrete Integrable Systems and Poisson Algebras From Cluster Maps
\textit{Communications in Mathematical Physics 527–584}, 
(2013)

\bibitem{FZ2001}
S.Fomin and A. Zelevinsky,
Cluster algebras I: Foundations,
\textit{Journal of the American Mathematical Society}
(2001)

\bibitem{FZ2002}
S. Fomin and A. Zelevinsky,
The Laurent phenomenon, 
\textit{Adv. in Appl. Math., 28, no.
2, 119–144}
 (2002)

\bibitem{OT2022}
N. Okubo and T. Suzuki, 
Generalized q-Painlevé VI Systems of Type (A2n+1+A1+A1)(1) Arising From Cluster Algebra,
 \textit{International Mathematics Research Notices, Volume 2022, Issue 9,Pages 6561–6607}
  (2022)
  
  \bibitem{hone2019cluster}
A. Hone, P. Lampe and T. Kouloukas, 
Cluster algebras and discrete integrability,
\textit{Nonlinear Systems and Their Remarkable Mathematical Structures},
 (2019)
\end{thebibliography}
\end{document}